\documentclass[11pt]{article}
\usepackage{amstext, amsmath,latexsym,amsbsy,amssymb,amsmath}

\usepackage{enumitem}
\newtheorem{definition}{Definition}[section]

\newtheorem{proposition}{Proposition}[section]
\newtheorem{remark}{Remark}[section]
\newenvironment{proof}{{\bf Proof\ }}{\QED\\}
\newtheorem{lemma}{Lemma}[section]

\numberwithin{equation}{section}
\newtheorem{theorem}{Theorem}[section]
\newtheorem{corollary}{Corollary}[section]

\newcommand{\QED}{\hspace*{\fill}\rule{2.5mm}{2.5mm}}

\usepackage{amssymb}\usepackage{graphicx}
\usepackage{tikz}

\usepackage{color}
\newcommand\qed{\hfill$\sqcap\kern-7.5pt\hbox{$\sqcup$}$}

\newcommand{\RR}{\mathbb{R}}

\newcommand{\beqn}{\begin{equation}}
\newcommand{\eeqn}{\end{equation}}
\newcommand{\bear}{\begin{eqnarray}}
\newcommand{\eear}{\end{eqnarray}}
\newcommand{\bean}{\begin{eqnarray*}}
\newcommand{\eean}{\end{eqnarray*}}

\newcommand{\cE}{\mathcal{E}}

\allowdisplaybreaks

\begin{document}
\title{On the dynamics of finite temperature trapped Bose gases}

\author{Avy Soffer\footnotemark[1] \and Minh-Binh Tran\footnotemark[2] 
}

\renewcommand{\thefootnote}{\fnsymbol{footnote}}

\footnotetext[1]{Mathematics Department, Rutgers University, New Brunswick, NJ 08903 USA.\\Email: soffer@math.rutgers.edu. 
}

\footnotetext[2]{Department of Mathematics, University of Wisconsin-Madison, Madison, WI 53706, USA. \\Email: mtran23@wisc.edu
}

\maketitle
\begin{abstract} 
The system that describes the dynamics of a Bose-Einstein Condensate (BEC) and the thermal cloud at finite temperature consists of a nonlinear Schrodinger (NLS) and a quantum Boltzmann (QB) equations. In such a system of trapped Bose gases at finite temperature, the QB equation corresponds to the evolution of the density distribution function of the thermal cloud and the NLS is the equation of the condensate.  The quantum Boltzmann collision operator in this temperature regime is the sum of two operators $C_{12}$ and $C_{22}$, which describe collisions of the condensate and the non-condenstate atoms and collisions between non-condensate atoms. Above the BEC critical temperature, the system is reduced to an equation containing only a collsion operator similar to $C_{22}$, which possesses a blow-up positive radial solution with respect to the $L^\infty$ norm (cf. \cite{EscobedoVelazquez:2015:FTB}).
On the other hand, at the very low temperature regime (only a portion of the transition temperature $T_{BEC}$), the system can be simplified into an equation of $C_{12}$, with a different (much higher order) transition probability, which
has a unique global classical positive radial solution with weighted $L^1$ norm (cf. \cite{AlonsoGambaBinh}). In our model, we first decouple the QB, which contains $C_{12}+C_{22}$, and the NLS equations, then show a global existence and uniqueness result for classical positive radial solutions to the spatially homogeneous kinetic system. Different from the case considered in \cite{EscobedoVelazquez:2015:FTB}, due to the presence of the BEC, the collision integrals are associated to sophisticated  energy manifolds rather than  spheres, since the particle energy is approximated  by the  Bogoliubov dispersion law. Moreover,  the mass of the full system is not conserved while it is conserved for the case considered in  \cite{EscobedoVelazquez:2015:FTB}. A new theory is then supplied.
\end{abstract}

{\bf Keyword:}
{Quantum kinetic theory; Bose-Einstein  condensate; quantum Boltzmann equation;  defocusing cubic nonlinear Schrodinger equation; quantum gases. 



\section{Introduction}

The study of kinetic equations has a very long history, starting with the classical Boltzmann equation, which provides a description of the dynamics of dilute monoatomic gases (cf. \cite{Cercignani:1975:TAB,Cercignani:1988:BEI,CercignaniIllnerPulvirenti:1994:TMT,Glassey:1996:CPK,Villani:2002:RMT}). As an attempt to extend the Boltzmann equation to deal with quantum gases, the Boltzmann-Nordheim (Uehling-Uhlenbeck) equation was introduced \cite{Nordheim:OTK:1928,UehlingUhlenbeck:TPI:1933}. However, the Boltzmann-Nordheim (Uehling-Uhlenbeck) equation fails to describe a Bose gas at temperatures which are close to and below the Bose-Einstein Condensate (BEC) critical temperature, due to the fact that its steady-state solution is a Bose-Einstein distribution in particle energies. Below the critical temperature, many-body effects modify the equilibrium distribution so that this distribution depends on quasiparticle energies. These
are accounted for by mean fields which break the the
unperturbed Hamiltonian $U (1)$ gauge symmetry. Therefore, a new description in terms of quasiparticles is required. Such a quantum kinetic theory was initiated by    Kirkpatrick and Dorfman \cite{KD1,KD2}, based on the rich body of research carried out in the period 1940-67 by Bogoliubov, Lee and Yang, Beliaev, Pitaevskii, Hugenholtz and Pines, Hohenberg and Martin, Gavoret and Nozi`eres, Kane and Kadanoff and many others. After the production of the first BECs, that later led Cornell, Wieman, and Ketterle to the 2001 Nobel Prize of Physics \cite{WiemanCornell,Ketterle,BarbaraGoss}, there has been an explosion of research on  the kinetic theory associated to BECs. Based on Kirkpatrick-Dorfman's works, Zaremba, Nikuni and Griffin successfully formulated a  self-consistent Gross-Pitaevskii-Boltzmann model, which is nowadays known as the ‘ZNG’ theory (cf. \cite{bijlsma2000condensate,ZarembaNikuniGriffin:1999:DOT}). Independent of the mentioned authors, Pomeau et. al. \cite{pomeau1999theorie} also proposed a similar model for the kinetics of BECs. Later, Gardinier, Zoller and collaborators  derived a Master Quantum Kinetic Equation (MQKE) for BECs, which returns to the ZNG model at the limits,  and  introduced the terminology ``Quantum Kinetic Theory'' in the series of papers \cite{QK1,QK3,QK5,QK0,QK6,QK4,QK2}.  The ZNG theory also gave the first quantitative predictions of vortex nucleation at finite temperatures \cite{williams2002dynamical}. Many other experiments have also confirmed the validity of the model (cf. \cite{proukakis2008finite}).  We refer to the review paper  \cite{anglin2002bose} for discussions on the condensate growth problem concerning the  MQKE model and  the books \cite{GriffinNikuniZaremba:2009:BCG,inguscio1999bose,ColdAtoms1} for more theoretical and experimental justifications of  the ZNG model, as well as  the tutorial article \cite{proukakis2008finite} for an easy introduction. Let us mention that besides the ZNG theory, there have been other works  describing the kinetics of BECs as well (see \cite{Allemand:DOF:2009,IG, kagan1997evolution,semikoz1995kinetics,semikoz1997condensation,Spohn:2010:KOT,Stoof:1999:CVI}, and references therein).

Let us first recall the ZNG model for finite temperature trapped bose gases, i.e. the temperature $T$ of the gas is below the transition temperature $T_{BEC}$ but above absolute zero. Denote $f(t,r,p)$ to be the density function of the Bose gas at time $t$, position $r$ and momentum $p$ and $\Phi(t,r)$ to be the wave function of the BEC. Employing the short-handed notation
$f_i=f(t,r,p_i)$, $i=1,2,3,4$. The Schr\"odinger (or the Gross-Pitaevski) equation for the condensates reads (cf. \cite{bijlsma2000condensate}):
\begin{equation}
\begin{aligned}\label{GPA}
i \hbar {\partial_t \Phi(t,r)} =&\ \Big(-\frac{\hbar^2 \Delta_{{r}}}{2m}+g[N_c(t,r) + 2n_n(t,r)] -i\Lambda_{12}[f](t,r) +V(r)\Big)\Phi(t,r), \ \ (t,r)\in\mathbb{R}_+\times\mathbb{R}^3,\\
\Lambda_{12}[f](t,r) = &\ \frac{\hbar}{2N_c}\Gamma_{12}[f](t,r),\\
\Gamma_{12}[f](t,r)= &\ \int_{\mathbb{R}^{3}}C_{12}[f](t,r,p)\frac{d p}{(2\pi \hbar)^3},\\ 
 n_n(t,r) \ = & \ \int_{\mathbb{R}^3}f(t,r,p)d p,
 \\
~~~\Phi(0,r)=& \ \Phi_0(r), \forall r\in\mathbb{R}^3,
\end{aligned}
\end{equation}
where $N_c(t,r)=|\Phi|^2(t,r)$ is the condensate density, $\hbar$ is the Planck constant, $g$ is the interaction coupling constant proportional to  the $s$-wave scattering length $a$, $V(r)$ is the confinement potential, and the operator $C_{12}$ can be found in   
the quantum Boltzmann equation for the non-condensate atoms (cf. \cite{bijlsma2000condensate}), written below:
\begin{eqnarray}\label{QBFullA}
{\partial_t f}(t,r,p)&+&\frac{p}{m}\cdot\nabla_{{r}} f(t,r,p) \ - \ \nabla_r U(t,r)\cdot \nabla_p f(t,r,p)\\\nonumber
&=&Q[f](t,r,p):=C_{12}[f](t,r,p)+C_{22}[f](t,r,p), (t,r,p)\in\mathbb{R}_+\times\mathbb{R}^3\times\mathbb{R}^3,\\\nonumber
C_{12}[f](t,r,p_1)&:=&\lambda_1N_c(t,r) \iint_{\mathbb{R}^{3}\times\mathbb{R}^{3}}K^{12}({p}_1,{p}_2,{p}_3)\delta(mv_c+ {p}_1-{p}_2-{p}_3)\delta(\mathcal{E}_c+\mathcal{E}_{{p}_1}-\mathcal{E}_{{p}_2}-\mathcal{E}_{{p}_3})\\\label{C12A}
& &\times[(1+f_1)f_2f_3-f_1(1+f_2)(1+f_3)]d p_2d p_3\\\nonumber
&&-2\lambda_1n_c(t,r)\iint_{\mathbb{R}^{3}\times\mathbb{R}^{3}}\delta(mv_c+{p}_2-{p}_1-{p}_3)\delta(\mathcal{E}_c+\mathcal{E}_{{p}_2}-\mathcal{E}_{{p}_1}-\mathcal{E}_{{p}_3})\\\nonumber
& &\times[(1+f_2)f_1f_3-f_2(1+f_1)(1+f_3)]d p_2d p_3,\\\label{C22A}
C_{22}[f](t,r,p_1)&:=&\lambda_2\iiint_{\mathbb{R}^{3}\times\mathbb{R}^{3}\times\mathbb{R}^{3}}K^{22}({p}_1,{p}_2,{p}_3,p_4)\delta({p}_1+{p}_2-{p}_3-{p}_4)\\\nonumber
& &\times\delta(\mathcal{E}_{{p}_1}+\mathcal{E}_{{p}_2}-\mathcal{E}_{{p}_3}-\mathcal{E}_{{p}_4})\times\\\nonumber
&&\times [(1+f_1)(1+f_2)f_3f_4-f_1f_2(1+f_3)(1+f_4)]d p_2d p_3d p_4,
\\\nonumber
f(0,r,p)&=&f_0(r,p), (r,p)\in\mathbb{R}^3\times\mathbb{R}^3,
\end{eqnarray}
where $\lambda_1=\frac{2g^2}{(2\pi)^2\hbar^4},$ $\lambda_2=\frac{2g^2}{(2\pi)^5\hbar^7}$, $m$ is the mass of the particles,  and $\mathcal{E}_{{p}}$ is the Hartree-Fock energy \cite{GriffinNikuniZaremba:2009:BCG}
\begin{equation}\label{HF}
\mathcal{E}_p=\frac{|p|^2}{2m}+U(t,r).
\end{equation}

Notice that $C_{22}$ is the Boltzmann-Nordheim (Uehling-Ulenbeck) quantum Boltzmann collision operator. 
If one writes  
\begin{equation}\label{def-PhiA}
\begin{aligned}
\Phi \ = & \ |\Phi(t,r)|e^{i\phi(r,t)},
\end{aligned}
\end{equation}
the condensate velocity can be defined as
\begin{equation}\label{def-vcA}
v_c(t,r) = \frac{\hbar}{m}\nabla \phi(t,r),
\end{equation}
and the condensate chemical potential is then
\begin{equation}\label{def-mucA}
 \mu_c 
=\ \frac{1}{\sqrt{n_c}}\left(-\frac{\hbar^2\Delta_r}{2m}+ V +g[2n_n+N_c]\right)\sqrt{N_c}.
\end{equation}
The potential $U$ and the condensate energy $\mathcal{E}_c$ are  written as follows
\begin{equation}\label{def-U}
U(t,r)\ = \ V(r) + 2g[N_c(t,r)+n_n(t,r)],
\end{equation}
and 
\begin{equation}\label{def-Ec}
\mathcal{E}_c(t,r) \ = \ \mu_c(t,r) \ +\ \frac{mv_c^2(t,r)}{2}.
\end{equation}

Notice that \eqref{C12} describes collisions of the condensate and the non-condensate atoms (condensate growth term), \eqref{C22} describes collisions between non-condensate atoms, and \eqref{GPA} is the defocusing nonlinear Schrodinger equation of the condensate (see Figure \ref{fig}). 
For the sake of simplicity, we denote $\lambda_1={\frac{2g^2N_c}{(2\pi)^2\hbar^4}}$ and $\lambda_2=\frac{2g^2}{(2\pi)^5\hbar^7}$.

The transition probability kernel $$K^{12}(p_1,p_2,p_3)=|A^{12}(|p_1|,|p_2|,|p_3|)|^2$$ of $C_{12}$  is given by the scattering amplitude (cf. \cite{E,ReichlGust:2012:CII,ReichlGust:2013:RRA,KD1,KD2,ReichlGust:2013:TTF})
\begin{equation}\label{Def:TransitionProbabilityKernel:K12}
\begin{aligned}
&~~~A^{12}(|p_1|,|p_2|,|p_3|):=\\
:=&~(u_{p_3}-v_{p_3})(u_{p_1}u_{p_2}+ v_{p_1}v_{p_2})+(u_{p_2}-v_{p_2})(u_{p_1}u_{p_3}+ v_{p_1}v_{p_3}) - (u_{p_1}-v_{p_1})(u_{p_2}v_{p_3}+ v_{p_2}u_{p_3}),
\end{aligned}
\end{equation}
where 
$$u^2_p=\frac{\frac{p^2}{2m}+gn_c+\mathcal{E}_p}{2\mathcal{E}_p},~~~ u^2_p-v^2_p=1.$$

The transition probability kernel $$K^{22}(p_1,p_2,p_3,p_4)=|A^{22}(p_1,p_2,p_3,p_4)|^2$$ of $C_{22}$  is given by the scattering amplitude (cf. (cf. \cite{E,ReichlGust:2012:CII,ReichlGust:2013:RRA,KD1,KD2,ReichlGust:2013:TTF}))
\begin{equation}\label{Def:TransitionProbabilityKernel:K22}
\begin{aligned}
&~~~A^{22}(|p_1|,|p_2|,|p_3|,|p_4|):=\\
:=&~u_{p_1}u_{p_2}u_{p_3}u_{p_4}+u_{p_1}v_{p_2}u_{p_3}v_{p_4}+u_{p_1}v_{p_2}v_{p_3}u_{p_4}+v_{p_1}u_{p_2}u_{p_3}v_{p_4}+v_{p_1}u_{p_2}v_{p_3}u_{p_4}+v_{p_1}v_{p_2}v_{p_3}v_{p_4}.
\end{aligned}
\end{equation}

When the temperature of the system is very low $T< 0.5T_{BEC}$ (cf. \cite{GriffinNikuniZaremba:2009:BCG}), the Hatree-Fock energy is no longer valid. A more general energy is used instead: the Bogoliubov dispersion law (cf. \cite{E,KD1,KD2})
\bear \label{def-E}
\mathcal{E}_p=\mathcal{E}(p)=\sqrt{\kappa_1 |p|^2 + \kappa_2 |p|^4}, \qquad \kappa_1=\frac{gN_c}{m}>0, \quad \kappa_2=\frac{1}{4m^2}>0.
\eear

Notice that the first rigorous proof of BEC in a physically realistic, continuum model was given in 2002 by Lieb and Seiringer (cf. \cite{lieb2002proof}). Besides the kinetic theory point of view, there are other approaches, valid with different physical assumptions, of understanding the dynamics of BECs and their thermal clouds, for instance, the works \cite{bach2016time,KirkpatrickSchlein:2013:ACL,DeckertFrohlichPickl:2016:DSW,GrillakisMachedon:2013:BMF,GrillakisMachedon:2013:PEA,GrillakisMachedonMargetis:2011:SOC, mitrouskas2016bogoliubov,seiringer2011excitation}, and cited references.

  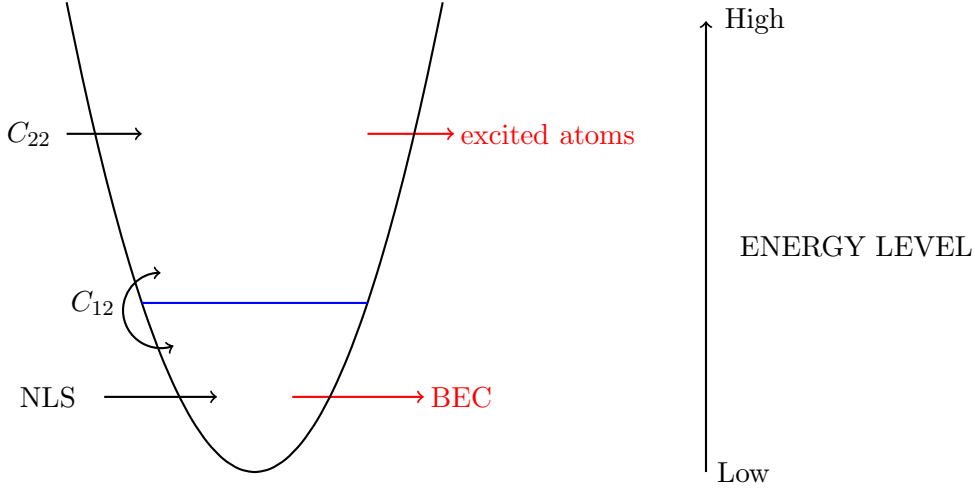
\begin{figure}

\begin{tikzpicture}
      \draw[thick, ->] (-2.5,4.5) -- ++(1,0) node[xshift=-1.5cm] {$C_{22}$};
      \draw[thick, ->, red] (1.5,4.5) -- ++(1.15,0) node[xshift=1.25cm] {excited atoms};
      \draw[thick, ->] (-2,1) -- ++(1.5,0) node[xshift=-2.25cm] {NLS};
      \draw[thick, ->, red] (0.5,1) -- ++(1.75,0) node[xshift=.5cm] {BEC};
      \draw[thick, ->] (6,0) node [xshift=.5cm]{Low} -- ++(0,3)  node[xshift=2.0cm] {ENERGY LEVEL} -- ++(0,3) node[xshift=.65cm] {High};
      \draw[ thick,domain=-2.5:2.5,smooth,variable=\x,black] plot ({\x},{\x*\x});
      \draw [thick, blue] (-1.5, 2.25) -- (1.5, 2.25);

      \draw[thick, <->] (-1.25,2.65) arc (90:290:.5);
      \node at (-2.15, 2.25) {$C_{12}$};
    \end{tikzpicture}
    \label{fig}\caption{The Bose-Einstein Condensate (BEC) and the excited atoms.}
    
   \end{figure}

\bigskip
{\bf The toy model}
\bigskip

Since the system \eqref{GPA}-\eqref{QBFullA} is too complicated, it is impossible to study all its properties in one single paper. As the first step to understand \eqref{GPA}-\eqref{QBFullA}, we impose a few simplifications. We suppose that the equation \eqref{QBFullA} is homogeneous in space and  the condensate density distribution function $N_c=|\Phi|^2$ can be considered as a constant $n_c$. The system is then reduced to the following toy model:

\begin{eqnarray}\label{QBFull}
\frac{\partial f}{\partial t} &=&Q[f]:=C_{12}[f]+C_{22}[f], (t,p)\in\mathbb{R}_+\times\mathbb{R}^3,\\\nonumber
f(0,p)&=&f_0(p), p\in\mathbb{R}^3,\end{eqnarray}
in which we rewrite $C_{12}$ and $C_{22}$ following the KD style \cite{KD1,KD2}, since it is simpler
\begin{eqnarray}\label{C12}
C_{12}[f](t,r,p_1)&:=&n_c\lambda_1\iint_{\mathbb{R}^{3}\times\mathbb{R}^{3}}K^{12}({p}_1,{p}_2,{p}_3)\delta({p}_1-{p}_2-{p}_3)\delta(\mathcal{E}_{{p}_1}-\mathcal{E}_{{p}_2}-\mathcal{E}_{{p}_3})\\\nonumber
& &\times[(1+f_1)f_2f_3-f_1(1+f_2)(1+f_3)]d{p}_2d{p}_3\\\nonumber
&&-2n_c\lambda_1\iint_{\mathbb{R}^{3}\times\mathbb{R}^{3}}K^{12}({p}_1,{p}_2,{p}_3)\delta({p}_2-{p}_1-{p}_3)\delta(\mathcal{E}_{{p}_2}-\mathcal{E}_{{p}_1}-\mathcal{E}_{{p}_3})\\\nonumber
& &\times[(1+f_2)f_1f_3-f_2(1+f_1)(1+f_3)]d{p}_2d{p}_3,\\\label{C22}
C_{22}[f](t,r,p_1)&:=&\lambda_2\iiint_{\mathbb{R}^{3}\times\mathbb{R}^{3}\times\mathbb{R}^{3}}K^{22}({p}_1,{p}_2,{p}_3,p_4)\delta({p}_1+{p}_2-{p}_3-{p}_4)\\\nonumber
& &\times\delta(\mathcal{E}_{{p}_1}+\mathcal{E}_{{p}_2}-\mathcal{E}_{{p}_3}-\mathcal{E}_{{p}_4})\times\\\nonumber
&&\times [(1+f_1)(1+f_2)f_3f_4-f_1f_2(1+f_3)(1+f_4)]d{p}_2d{p}_3d{p}_4,
\end{eqnarray}
where $\mathcal{E}_p$ is the Bogoliubov dispersion relation \eqref{def-E} and is valid when $T< 0.5T_{BEC}$ (cf. \cite{E}). The form of the Bogoliubov dispersion relation makes the study of \eqref{QBFull} much more difficult than the classical Boltzmann equation and it is the goal of our paper to develop techniques that can resolve this difficulty.

Let us mention that in real physical situations, the system of bosons is normally not spacial dependent. Moreover, the toy model \eqref{QBFull} does not really describe the full dynamics of the thermal cloud-condensate system since the equation does not conserve the total number of atoms. However, it is interesting to begin studying the ZNG model by understanding the spatially homogeneous kinetic equation with the Bogoliubov dispersion relation, in order to gain some insights into the full model. Such toy models are indeed useful and have been used in the physics community (cf. \cite{EPV,ReichlGust:2012:CII,ReichlGust:2013:RRA,gust2013transport}), under the assumption that the temperature of the system is very low $T<0.01T_{BEC}$. In such a low temperature system, the portion of excited atoms outside the condensate is very small, in comparison with the number of condensed atoms; and therefore $N_c$ can be regarded as a constant.  For instance, in \cite{ReichlGust:2013:RRA}, 
the authors consider only the kinetic part of the system, which is spatially homogeneous, as what we are assuming in our paper (equation (10) in \cite{ReichlGust:2013:RRA})
$$\partial_t f = C_{12}[f]+C_{22}[f] +C_{13}[f].$$
Indeed, the contribution of \cite{ReichlGust:2013:RRA} is that they could derive a new collision operator $C_{13}[f]$ which complements the ZNG theory in some  cases. We refer to  \cite{ReichlGust:2012:CII,ReichlGust:2013:RRA,gust2013transport} for the discussion about the reason why the number of the excitations (bogolon number) is not conserved. In those works, the solutions of the kinetic equations are shown to converge to equilibrium and explicit convergence rates are computed.  The theoretical findings are shown to be in perfect agreements with  experimental results and with the famous Lee-Yang theory. A similar model has  also been used in \cite{EPV}.  We refer to the physics paper \cite{ReichlTran} for more theoretical and experimental discussions on spacial homogeneous models describing the dynamics of low temperature dilute Bose gases.  

Notice that the two ways of writing \eqref{C12A}-\eqref{C22A} and \eqref{C12}-\eqref{C22} are equivalent as explained in \cite{ZarembaNikuniGriffin:1999:DOT}. Indeed, in \eqref{C12}-\eqref{C22}, $p_i$ represent the quasiparticle momentum in the local rest frame.


Let us mention  a difficulty, pointed out by Eckern (cf. \cite{E}), that arises from the form of the transition probability $A^{22}$.  Define the characteristic momentum for the crossover
between the linear and the quadratic part of the spectrum to be
$p_0=2mn_cg$, in which $g$ is the repulsive point interaction.  If all momenta are much smaller than $p_0$ i.e. $|p_1|, |p_2|, |p_3|<<p_0$, we obtain the following unphysical asymptotic behavior (cf. \cite{E})
$$|A^{22}(|p_1|,|p_2|,|p_3|,|p_4|)|^2\approx |p_1|^{-1}|p_2|^{-1}|p_3|^{-1}|p_4|^{-1}.$$ 
This question is still open in the physics community, as discussed in \cite{BennemannBennemann:TPL:1976,E,Nepomnyashchii:1978,PopovSeredniakov:1978}. In a mathematical point of view, there are many kinetic equations with singular kernels
for which it is possible to prove existence of solutions. However, such singularities in our case lead to the loss in moments of the solutions of the equation. An investigation of this sophisticated  question will be the scope of our forthcoming paper. As a consequence, to avoid this singular behavior, the following transition probability is chosen for mathematical convenience
\begin{equation}\label{Def:TransitionProbabilityKernel:K22A}
\begin{aligned}
K^{22}(p_1,p_2,p_3,p_4)=|A^{22}(|p_1|,|p_2|,|p_3|,|p_4|)|^2\chi_{\{|p_1|,|p_2|,|p_3|,|p_4|\geq p_0\}},
\end{aligned}
\end{equation}
where $\chi_{\{|p_1|,|p_2|,|p_3|,|p_4|\geq p_0\}}$ is the characteristic function of the set $\{|p_1|,|p_2|,|p_3|,|p_4|\geq p_0\}$: namely, it turns out that the dominant collision process of non-condensate atoms in the low temperature region with small momenta is the interactions between non-condensate atoms and the condensate; in this region, we there suppose that the effect of $C_{22}$ is much smaller than $C_{12}$. In other words, we assume that when $p$ is small $C_{12}$ is dominant and when $p$ is large $C_{22}$ is dominant, since $C_{12}$ describes the collisions between excited atoms at low quantum levels and atoms in the ground state, and $C_{22}$ describes collisions between excited atoms at high quantum levels.    As a consequence, it is reasonable to impose a cut-off on the kernel of $C_{22}$ and not on $C_{12}$. With this truncated transition probability, there exists a positive constant $\Gamma$ depending on $p_0$, such that
\begin{equation}\label{K22Bound}
\begin{aligned}
K^{22}(p_1,p_2,p_3,p_4)<\Gamma.
\end{aligned}
\end{equation}
~~\\

Note that some other mathematical results for quantum kinetic equations have been obtained in \cite{AlonsoGambaBinh,JinBinh,CraciunBinh,Binh9,GambaSmithBinh,GermainIonescuTran,ToanBinh,ToanBinh2,ReichlTran,soffer2016coupling} . Quantum kinetic equations have very similar formulations with the so-called wave turbulence kinetic equations. We refer to \cite{buckmaster2016analysis,buckmaster2016effective,EscobedoVelazquez:2015:OTT,FaouGermainHani:TWN:2016,germain2015high,germain2015continuous,germain2016continuous,LukkarinenSpohn:WNS:2011,Nazarenko:2011:WT,Spohn:WNW:2010,zakharov2012kolmogorov,Zakharov:1998:NWA} for more recent advances on the theory.

In the current work, we restrict our attention to  spatial homogeneous and radial solutions of \eqref{QBFull} $$f(0,p)=f_0(|p|),~~f(t,p)=f(t,|p|).$$

Note that in \cite{EscobedoVelazquez:2015:FTB}, the authors consider the Boltzmann-Nordheim (Uehling-Uhlenbeck) equation
\begin{eqnarray*}
\partial_t f_1&=&\iiint_{\mathbb{R}^{3}\times\mathbb{R}^{3}\times\mathbb{R}^{3}}\delta({p}_1+{p}_2-{p}_3-{p}_4)\delta(|{{p}_1}|^2+|{{p}_2}|^2-|{{p}_3}|^2-|{{p}_4}|^2)\times\\\nonumber
&&\times [(1+f_1)(1+f_2)f_3f_4-f_1f_2(1+f_3)(1+f_4)]d{p}_2d{p}_3d{p}_4,
\end{eqnarray*}
which, by the radially symmetry assumption of $f$, can be reduced to the equivalent form
\begin{eqnarray*}
\partial_t f_1&=&\iiint_{\mathbb{R}_+\times\mathbb{R}_+\times\mathbb{R}_+ }\frac{\min\{|p_1|,|p_2|,|p_3|,|p_4|\}|p_1||p_2||p_3||p_4|}{|p_1|^2}\\
&&~~\times\delta(|{p_1}|^2+|{p_2}|^2-|{p_3}|^2-|{p_4}|^2)[f_3f_4(1+f_1+f_2)\\
&&~~-f_1f_2(1+f_3+f_4)]d|{p}_2|d|{p}_3|d|{p}_4|,
\end{eqnarray*}
By the same argument as in \cite{EscobedoVelazquez:2015:FTB}, $C_{22}$ also has the following  form
\begin{equation}\label{C22New}\begin{aligned}
C_{22}[f]=&~~\kappa_3\iiint_{\mathbb{R}_+\times\mathbb{R}_+\times\mathbb{R}_+ }K^{22}(p_1,p_2,p_3,p_4)\frac{\min\{|p_1|,|p_2|,|p_3|,|p_4|\}|p_1||p_2||p_3||p_4|}{|p_1|^2}\\
&~~\times\delta(\mathcal{E}_{p_1}+\mathcal{E}_{p_2}-\mathcal{E}_{p_3}-\mathcal{E}_{p_4})[f_3f_4(1+f_1+f_2)-f_1f_2(1+f_3+f_4)]d|{p}_2|d|{p}_3|d|{p}_4|,\end{aligned}
\end{equation}
where $\kappa_3$ is some positive constant.\\
Equation \eqref{QBFull} can be simplified as follows
\begin{equation}\label{QB}
\frac{\partial f}{\partial t}=Q[f]=C_{12}[f]+C_{22}[f],~~~ f(0,p)=f_0(|p|), \forall p\in\mathbb{R}^3.
\end{equation}
~~

We are interested in the existence and uniqueness of strong, classical and radial solutions of \eqref{QB}.
\begin{definition}\label{StrongSolution}
A function $f$ is defined to be a strong radial solution in $C([0,T),X)\cap C^1((0,T),Y)$,  for some function spaces $X,Y$, to \eqref{QB}, where $C_{22}$ is of the form \eqref{C22New}, $K_{22}$ is of the form \eqref{K22Bound}, $\mathcal{E}_p$ is  the Bogoliubov dispersion law and $N_c$ is assume to the the constant $n_c$, if and only if $f$ satisfies 
\begin{equation*}
\frac{\partial f}{\partial t}=Q[f]=C_{12}[f]+C_{22}[f],~~~ f(0,p)=f_0(|p|), \mbox{ for a.e. } p\in\mathbb{R}^3.
\end{equation*}
Moreover $f(t,p)=f(t,|p|)$ for all $(t,p)\in [0,T)\times\mathbb{R}^3$.

\end{definition}

{ Bosons are sensitive to temperature.  When the temperature is below the transition temperature $T<T_{BEC}$, the BEC is formed.   When we lower the temperature $T$, the behavior of the quasi-particles change. This can be seen clearly through the Bogoliubov dispersion relation, which depends on the density of the condensate and the temperature since $g$ depends on $T$. In the lower temperature range when $T$ is only a portion of  $T_{BEC}$, sometimes, we can suppose (cf. \cite{E,EPV})   that the interaction between bosons, i.e. the $C_{22}$ collision operator, is negligible, and the BEC is very stable. In this case, the system can be reduced to a kinetic equation involving the $C_{12}$ collision operator only:
\begin{equation}\label{Peierls}
\frac{\partial f}{\partial t}=C_{12}[f],~~~ f(0,p)=f_0(p), \forall p\in\mathbb{R}^3.
\end{equation}
In this  regime, the transition probability takes the form $\mathcal{C}_K|p_1||p_2||p_3|,$ which is unbounded, while \eqref{Def:TransitionProbabilityKernel:K12} is bounded. In the series of beautiful works \cite{ArkerydNouri:2012:BCI,ArkerydNouri:AMP:2013,ArkerydNouri:2015:BCI}, the study of \eqref{Peierls} has been done for the first time. In \cite{AlonsoGambaBinh} it has been proved that \eqref{Peierls} has a unique positive radial solution, based on an argument of propagation of polynomial and exponential moments. We will see later that, unlike \eqref{Peierls}, polynomial and exponential moments of solutions of \eqref{QBFull} are not propagating on the time interval $[0,\infty)$, due to the presence of the collision operator $C_{22}$.  In \cite{ToanBinh}, it is prove that the solution of \eqref{Peierls} is bounded from below by a Gaussian. In other words, the operator $C_{12}$ is ``strongly'' positive. 
\\

Above the BEC critical temperature, the density of the condensate $n_c$ is $0$, then $C_{12}=0$. Equation \eqref{QB} is reduced to the Boltzmann-Nordheim (Uehling-Uhlenbeck) equation 
\begin{equation}\label{UU}
\frac{\partial f}{\partial t}=C_{22}[f],~~~ f(0,p)=f_0(|p|), \forall p\in\mathbb{R}^3,
\end{equation}
which has a blow-up positive radial solution in the $L^\infty$ norm if the mass of the initial data is too concentrated around the origin (cf. \cite{EscobedoVelazquez:2015:FTB}). Note that in this temperature regime, the transition probability is $K^{22}=1$ (cf. \cite{ReichlGust:2012:CII,gust2013transport}), which is different from the regime considered in this paper. The existence of a global weak and measure solution for the equation was treated in \cite{Lu:2004:OID,Lu:2005:TBE,Lu:2013:TBE}. In \cite{BriantEinav:2016:OTC}, local existence and uniqueness results, with respect to the $L^\infty$  norm, were obtained for the Boltzmann-Nordheim (Uehling-Uhlenbeck) equation. Let us mention that when the temperature is above the BEC critical temperature, the energy is of the form $\frac{p^2}{2m}$. The collision of two microscopic boxes of particles with momenta $p_1$ and $p_2$ changes the momenta into $p_3$ and $p_4$; and the conservation laws read: 
$$|p_1|^2+|p_2|^2=|p_3|^2+|p_4|^2, \ \ \ p_1+p_2=p_3+p_4. $$
Since $p_1$, $p_2$, $p_3$, $p_4$ belong to the sphere centered at $\frac{p_1+p_2}{2}$ with radius $\frac{|p_1-p_2|}{2}$, the collision operator $C_{22}$ can be expressed as a integration on a sphere, following the strategy represented in \cite{Carleman:1933:TEI,Villani:2002:RMT} for the classical Boltzmann operator.

In our case $\mathcal{E}_p$ is approximated by the Bogoliubov dispersion law \eqref{def-E},  the collision operators are integrals on much more complicated manifolds. Classical techniques used for the classical Boltzmann equation cannot be applied. New estimates on 
 energy manifolds, such as,
 $$\mathcal{E}_{p_1}=\mathcal{E}_{p_2}+\mathcal{E}_{p_3}, \ \ \ p_1=p_2+p_3. $$
 are then required. 
 
  Moreover, \eqref{UU} conserves the mass of the solution, while the full equation \eqref{QB} does not. As a consequence, estimating the mass of the solution to \eqref{QB} is  a crucial task.

Let us emphasize that due to the presence of the $C_{12}$ term, which is much more complicated than the classical Boltzmann collision operators due to its non-symmetry structure, \eqref{QB} is much more complicated than the  Boltzmann-Nordheim (Uehling-Uhlenbeck) equation, as it has already been noticed in a series of beautiful works \cite{ArkerydNouri:2012:BCI,ArkerydNouri:AMP:2013,ArkerydNouri:2015:BCI}, where the study of $C_{12}$ has been done for the first time. In order to study $C_{12}$,  the authors of \cite{AlonsoGambaBinh} have developed special techniques, based on the ideas of propagation and creation of exponential and polynomial moments for only the collision operator $C_{12}$. In our case, as it is shown later, the mass of the solution of \eqref{QB} is not conserve and the presence of $C_{12}+C_{22}$ makes the problem different from the case considered in \cite{AlonsoGambaBinh}.  Indeed, in the works \cite{ArkerydNouri:2012:BCI,ArkerydNouri:AMP:2013,ArkerydNouri:2015:BCI},  the authors impose a  cut-off on $C_{12}$ for small momentums. As a consequence, the authors approximate the Bogoliubov dispersion relation by a simplified one, which significantly simplifies the analysis. However, we do not need to impose this cut-off on $C_{12}$. Since $C_{12}$ is used in its most general form in our case, it is very important that the Bogoliubov dispersion relation is kept. 

Moreover, we tried the strategy of \cite{Arkeryd:1972:OBE} used to prove the existence and uniqueness of classical solutions to the classical Boltzmann equation, but it does not seem applicable. The main reason is that for \eqref{QB}, we can only establish bounds on  weighted $L^1$ norms of the solution. Bounds on weighted $L^p$ norms is still an open question and the H-theorem is not useful in this case. As a consequence, the  Dunford-Pettis theorem cannot be used and the strategy of \cite{Arkeryd:1972:OBE}, even though very powerful, cannot be directly applied. We then have to develop new ideas to show that \eqref{QB} indeed has a global and classical solution in weighted $L^1$  spaces. Our result is different from the results considered in \cite{Lu:2004:OID,Lu:2005:TBE,Lu:2013:TBE} about the existence and uniqueness of global weak solution for the Boltzmann-Nordheim (Uehling-Uhlenbeck) equation. Note that our method also works for the collision kernel $K^{12}$ of the more complicated form $C|p_1|^\varrho|p_2|^\varrho|p_3|^\varrho$, $(\varrho\ge 0)$.

Let us define
\begin{equation}\label{L1Space}
L^1_m(\mathbb{R}^3)=\left\{~~f~~\Big{|}~~ \|f\|_{L^1_m}:= \int_{\mathbb{R}^3}|p|^m|f(p)|dp<\infty\right\},
\end{equation}
\begin{equation}\label{L1Space1}
\mathcal{L}^1_{m}(\mathbb{R}^3)=\left\{~~f~~\Big{|}~~ \|f\|_{\mathcal{L}^1_{m}}:= \int_{\mathbb{R}^3}|f(p)|\mathcal{E}_p^{m/2}dp<\infty\right\},
\end{equation}
\begin{equation}\label{L1Space1}
\mathbb{L}^1_{m}(\mathbb{R}^3)=\left\{~~f~~\Big{|}~~ \|f\|_{\mathbb{L}^1_{m}}:= \int_{\mathbb{R}^3}|f(p)|\left(1+\mathcal{E}_p^{m/2}\right)dp<\infty\right\}.
\end{equation}
Our main result is the following theorem.
\begin{theorem}\label{Theorem:ExistenceKinetic}
Suppose that $f_0(p)=f_0(|p|)\geq 0$, and
$$\int_{\mathbb{R}^3}(1+\mathcal{E}_p)f_0(p)dp<\infty.$$

For any time interval $[0,T]$, let $n$, $n^*$ be two positive integers, $n>1$, $n_*$ is an odd number, $n^*>n+4$. For any positive number $\mathcal{R}$, there exists $\mathfrak{c}_{n_*}$ depending on $\mathcal{R}$ and $T$  satisfying $\mathfrak{c}_{n_*}(\mathcal{R},T)$ tends to infinity as $T$ or $\mathcal{R}$ tends to infinity, such that if
$$\int_{\mathbb{R}^3}\mathcal{E}_p^{n^*}f_0(p)<\mathfrak{c}_{n_*}(\mathcal{R},0), \int_{\mathbb{R}^3}f_0(p)<\mathcal{R},$$
then there exists a unique classical positive  radial solution $$f(t,p)=f(t,|p|)\in C^0([0,T],\mathbb{L}^1_{2n}(\mathbb{R}^3))\cap C^1((0,T),\mathbb{L}^1_{2n}(\mathbb{R}^3)) $$ of \eqref{QB} where $C_{22}$ is of the form \eqref{C22New}, $K_{22}$ is of the form \eqref{K22Bound}, $\mathcal{E}_p$ is  the Bogoliubov dispersion law and $N_c$ is assume to the the constant $n_c$.
\end{theorem}~~\\

One of the key ingredients of the proof of Theorem \ref{Theorem:ExistenceKinetic} is the following theorem about the existence and unique of solutions to ODEs on Banach spaces. The theorem has an  inspiration from  \cite{ABCL:2016,AlonsoGambaBinh,Bressan,Martin}. Notice that different from the previous cases considered in \cite{ABCL:2016,AlonsoGambaBinh}, we do not have the propagation of polynomial and exponential moments of the solution, as a consequence, we introduce new ideas to deal with this difficulty. Those ideas are discussed in Remarks \ref{RM1}, \ref{RM2} and \ref{RM3}.
\begin{theorem}\label{Theorem:ODE} Let $[0,T]$ be a time interval, $E:=(E,\|\cdot\|)$ be a Banach space, $\mathcal{S}_T$  be a bounded, convex and closed subset of $E$, and $Q:\mathcal{S}_T\rightarrow E$ be an operator  satisfying the following properties:
\begin{itemize}
\item [$(\mathfrak{A})$] Let $\|\cdot\|_*$ be a different norm of $E$, satisfying $\|\cdot\|_*\leq C_E\|\cdot\|$ for some universal constant $C_E$, and the function
\begin{eqnarray*}
|\cdot |_*: E&\longrightarrow&\mathbb{R}\\
  u&\longrightarrow& |u|_*,
\end{eqnarray*}
satisfying $$|u+v|_*\le |u|_*+|v|_*, \mbox{ and } \ \ \ |\alpha u|_*=\alpha|u|_*$$ for all $u$, $v$ in $E$ and $\alpha\in\mathbb{R}_+$.
\\ Moreover, $$|u|_*=\|u\|_*, \forall u\in\mathcal{S}_T,\ \  |u|_*\leq\|u\|_*\leq C_E\|u\|, \forall u\in E,$$ 
 $$|Q(u)|_*\le C_*(1+|u|_*), \forall u\in \mathcal{S}_T,$$ and $$\mathcal{S}_T\subset \overline{B_*\Big(O,(2R_*+1)e^{(C_*+1)T}\Big)}:=\overline{\Big\{u\in E \Big{|} \|u\|_*\le (2R_*+1)e^{(C_*+1)T}\Big\}},$$ for some positive constant $R_*\ge 1$.
\item [$(\mathfrak{B})$] Sub-tangent condition
\begin{equation*}
\liminf_{h\rightarrow0^+}h^{-1}\text{dist}\big(u+hQ[u],\,\mathcal{S}_T\big)=0,\qquad \forall\,u\in\mathcal{S}_T\cap B_*\Big(O,(2R_*+1)e^{(C_*+1)T}\Big)\,,
\end{equation*}

\item [$(\mathfrak{C})$] H\"{o}lder continuity condition
\begin{equation*}
\big\|Q[u] - Q[v]\big\| \leq C\|u - v\|^{\beta},\quad \beta\in(0,1), \quad \forall\,u,v\in\mathcal{S}_T\,,
\end{equation*}
\item [$(\mathfrak{D})$] one-side Lipschitz condition
\begin{equation*}
\big[ Q[u] - Q[v], u - v \big] \leq C\|u - v\|,\qquad \forall\,u,v\in\mathcal{S}_T\,,
\end{equation*}
where $$\big[ \varphi,\phi \big]: = \lim_{h\rightarrow 0^{-}}h^{-1}\big(\| \phi + h\varphi \| - \| \phi \| \big).$$
\end{itemize}
Then the equation 
\begin{equation}\label{Theorem_ODE_Eq}
\partial_t u=Q[u] \mbox{ on } [0,T]\times E,~~~~u(0)=u_0 \in \mathcal{S}_T\cap B_*(O,R_*)
\end{equation}
has a unique solution in $C^1((0,T),E)\cap C([0,T],\mathcal{S}_T)$.
\end{theorem}

\begin{remark}\label{RM1}
Note that for \eqref{QB}, the mass is not conserved. We indeed prove that it grows exponentially in Section \ref{Sec:Mass}. As a consequence, in  Theorem \ref{Theorem:ODE}, besides the norm $\|\cdot\|$ of the Banach space $E$, we also need the second norm $\|\cdot\|_*$ and the ball $$B_*\Big(O,(2R_*+1)e^{(C_*+1)T}\Big),$$ 
which take the crucial  role in controlling the mass of the solution on the time interval $[0,T]$. 
\\ Thanks to the control on the mass, we can later  prove that the collision operator $Q$ in \eqref{QB} is indeed Holder continuous, which means Condition $(\mathfrak{C})$ of Theorem \ref{Theorem:ODE} is satisfied. 
\end{remark}

\begin{remark}\label{RM2} In Theorem  \ref{Theorem:ODE}, $|\cdot|_*$ is a function from $E$ to $\mathbb{R}$, that coincides with the second norm in $\|\cdot\|_*$ in the set $\mathcal{S}_T$. This is due to the fact that, we will choose $\mathcal{S}_T$ to be a subset of the positive cone of $E=\mathbb{L}^1_{2n}(\mathbb{R}^3)$. 
\end{remark}

\begin{remark}\label{RM3} In Condition $(\mathfrak{B})$ of Theorem  \ref{Theorem:ODE}, we do not consider the boundary case where $$\|u\|_*=(2R_*+1)e^{(C_*+1)T}.$$ Our idea of the proof is to start with an initial condition $u(0)$ in the intersection of $\mathcal{S}_T$ and the ball $B_*(O,R_*)$, and make $u(t)$ evolve as long as $$\|u(t)\|_*<(2R_*+1)e^{(C_*+1)T}.$$ This idea is realized, in a discrete way,  in Part 2 of the proof of  Theorem  \ref{Theorem:ODE}.
\end{remark}

The plan of the paper is as follows: 
\begin{itemize}
\item Section \ref{Sec:QB} is devoted to the proof of Theorem \ref{Theorem:ExistenceKinetic}. This proof is divided into several steps:
\begin{itemize}
\item In Section \ref{Sec:Conservation}, basic properties of Equation \eqref{QB} are presented. We prove that solutions of \eqref{QB} conserve momentum and energy in Section \ref{Sec:MomentEnergy}. However, different from the Boltzmann-Nordheim (Uehling-Uhlenbeck) equation \eqref{UU}, the mass is not conserved for the full equation. Therefore, estimating the mass is a crucial task. Notice that different from previous studies (cf. \cite{EscobedoVelazquez:2015:FTB}), where the energy is $$\mathcal{E}_p=\frac{p^2}{2m};$$ in our case, due to the presence of the condensate, the energy is approximated by the Bogoliubov dispersion law \eqref{def-E}. This requires new estimates on the energy surfaces.
Section \ref{Sec:EnergySurface} is devoted to such estimates. Based on these estimates, in Section \ref{Sec:Mass}, we provide a bound of the mass of solutions to Equation \eqref{QB} on a finite time interval $[0,T].$
\item As a key ingredient of the proof of Theorem \ref{Theorem:ExistenceKinetic}, we show in Section \ref{Sec:MomentEstimates} that polynomial moments with arbitrary high orders of solutions of \eqref{QB} are bounded on a finite time interval $[0,T]$, which is the content of Proposition \ref{Propo:MomentsPropa}. Note that different from the very low temperature regimes considered in \cite{AlonsoGambaBinh}, in our regimes, polynomial moments are not propagating an created on $[0,\infty)$. The strategy of the proof of the proposition is to estimate moments of the collision operators $C_{12}$ and $C_{22}$, which are done in Sections \ref{Sec:C12} and \ref{Sec:C22} using results on energy surfaces of Section \ref{Sec:EnergySurface}. Based on these estimates, we obtain a differential inequality for finite time moments of high orders in Section \ref{Sec:MomentsEstimate}, which leads to the desired results of Proposition \ref{Propo:MomentsPropa}. 
\item In Section \ref{Sec:HolderEstimate}, we prove that the collision operators $C_{12}$ and $C_{22}$ are Holder continuous, thanks to Proposition \ref{Propo:MomentsPropa}. In order to do this, we decompose $C_{22}$ as the sum of two  operators $C_{22}^1$ and $C_{22}^2$, where the first one is of second order and the second one is of third order. The operators $C_{12}$, $C_{22}^1$ and $C_{22}^2$ are proven to be Holder continuous in Sections \ref{Sec:HolderEstimateC12}, \ref{Sec:HolderEstimateC221} and \ref{Sec:HolderEstimateC222}, respectively, on any time interval $[0,T]$. 
\item Using Theorem \ref{Theorem:ODE}, we prove in Section \ref{Sec:ExistenceKinetic} that Equation \eqref{QB} has a unique positive, radial solution on any time interval $[0,T]$.
\end{itemize}
\item The proof of Theorem \ref{Theorem:ODE}  is given in Section \ref{Appendix}.
\end{itemize}
\section{The quantum Boltzmann equation}\label{Sec:QB}
\subsection{Mass, momentum and energy of  solutions of the kinetic equation}\label{Sec:Conservation}
We will  make use of the following notation 
\begin{equation}\label{Def:MomentOrderk}
m_k[f]=\int_{\mathbb{R}^3}\mathcal{E}^k(p_1)f(p_1)dp_1.
\end{equation}
For convenience, we introduce 
\begin{equation}\label{CollisionOperator}
C_{12}[f]=C^1_{12}[f]+C^2_{12}[f] 
\end{equation}
with 
$$
\begin{aligned}C^1_{12}[f]&:=\iint_{\mathbb{R}^3\times\mathbb{R}^3} \mathcal{K}^{12}(p_1, p_2, p_3)
\Big [f(p_2)f(p_3)-f(p_1)(f(p_2)+f(p_3)+1) \Big ]dp_2dp_3
\\
C^2_{12}[f]&:=-2\iint_{\mathbb{R}^3\times\mathbb{R}^3}\mathcal{K}^{12}(p_2, p_1, p_3)\Big [f(p_1)f(p_3)-f(p_2)(f(p_1)+f(p_3)+1)\Big ]dp_2dp_3,
\end{aligned}$$
where the collision kernel is defined by 
$$\mathcal{K}^{12}(p_1, p_2, p_3)=\lambda_1 n_c K^{12}(p_1,p_2,p_3)\Big(\delta (\mathcal{E} (p_1)-\mathcal{E} (p_2)-\mathcal{E} (p_3))  \delta (p_1-p_2-p_3)\Big ).$$
We also define the energy surfaces/resonance manifolds
\begin{equation}\label{def-Sp}\begin{aligned}
S_p^0:& = \Big \{ p_*\in \mathbb{R}^3~:~\mathcal{E}(p-p_*) + \mathcal{E}(p_*)= \mathcal{E}(p) \Big\}
\\
S^1_p : &= \Big \{p_*\in \mathbb{R}^3~:~ \mathcal{E}(p+p_*) = \mathcal{E}(p)+\mathcal{E}(p_*) \Big\}\\
S^2_p : &= \Big \{p_*\in \mathbb{R}^3~:~ \mathcal{E}(p_*) = \mathcal{E}(p)+\mathcal{E}(p_*-p) \Big\}
\end{aligned}\end{equation}
for all $p \in \mathbb{R}^3 \setminus\{0\}$ and the functions
\begin{equation}\label{def-Hp}\begin{aligned}
H_0^p(x):& = \mathcal{E}(p-x) + \mathcal{E}(x) - \mathcal{E}(p),
\\
H_1^p(x) : &= \mathcal{E}(p+x) - \mathcal{E}(p)-\mathcal{E}(x),\\
H_2^p(x) : &= \mathcal{E}(x) -\mathcal{E}(p)-\mathcal{E}(x-p).
\end{aligned}\end{equation}

 Set
$$\bar{K}^{12}(p_1, p_2, p_3)=\lambda_1 n_c{K}^{12}(p_1, p_2, p_3),$$
by the nature of the Dirac delta function, the collision operators can be expressed under the form of the following surface integrals
$$
\begin{aligned}C^1_{12}[f]&:=\int_{S_{p_1}^0}\bar{K}^{12}_0(p_1, p_1 - p_3, p_3)\Big [f(p_1 - p_3)f(p_3)-f(p_1)(f(p_1 - p_3)+f(p_3)+1) \Big ]\; d\sigma(p_3)
\\
C^2_{12}[f]&:=2\int_{S_{p_1}^1}{\bar{K}_1^{12}(p_1 + p_3, p_1, p_3)}\Big [f(p_1 + p_3)(f(p_1)+f(p_3)+1) - f(p_1)f(p_3)\Big ]\; d\sigma(p_3),
\end{aligned}$$
where
$$ \bar{K}^{12}_0(p_1, p_1 - p_3, p_3)=\frac{\bar{K}^{12}(p_1, p_1 - p_3, p_3)
}{|\nabla  H^0_p(p_3)|},\ \ \  \ \bar{K}^{12}_1(p_1 + p_3, p_1, p_3)=\frac{\bar{K}^{12}(p_1 + p_3, p_1, p_3)}{|\nabla H^1_p(p_3)|}.$$

We also split $C_{12}[f]$ as the sum of gain and loss terms: 
\begin{equation}\label{Def:C12gainloss} 
C_{12}[f] = C_{12}^\mathrm{gain}[f]  -  C_{12}^\mathrm{loss}[f]\end{equation}
with 
$$
\begin{aligned}C_{12}^\mathrm{gain}[f]&:=\int_{S_{p_1}^0} \bar{K}^{12}_0(p_1, p_1 - p_3, p_3)f(p_1 - p_3)f(p_3) \; d\sigma(p_3)
\\&\quad + 2\int_{S_{p_1}^1}\bar{K}^{12}_1(p_1 + p_3, p_1, p_3)f(p_1 + p_3) \Big (f(p_1)+f(p_3)+1\Big )\; d\sigma(p_3),
\\
C_{12}^\mathrm{loss}[f]&:= fC_{12}^-[f],\\
C_{12}^-[f]&:=  \int_{S_{p_1}^0} \bar{K}^{12}_0(p_1, p_1 - p_3, p_3) \Big (f(p_1 - p_3)+f(p_3)+1\Big ) \; d\sigma(p_3) 
\\&\quad + 2\int_{S_{p_1}^1}\bar{K}^{12}_1(p_1 + p_3, p_1, p_3) f(p_3) \; d\sigma(p_3).
\end{aligned}$$

Similar as for $C_{12}$, we also split $C_{22}$ into gain and loss operators, as follows
\begin{equation}\label{C22GainLoss}
C_{22}[f]=C_{22}^{\mathrm{gain}}[f]-C_{22}^{\mathrm{loss}}[f],
\end{equation}
where
\begin{eqnarray*}
C_{22}^{\mathrm{gain}}[f]&:=&\lambda_2\iiint_{\mathbb{R}^{3\times3}}\mathcal{K}^{22}({p}_1,{p}_2,{p}_3,p_4)(1+f({p}_1))(1+f({p}_2))f({p}_3)f({p}_4)d{p}_2d{p}_3d{p}_4,\\\label{C22Loss}
C_{22}^{\mathrm{loss}}[f]&:=&f C_{22}^{-}[f],\\\nonumber
C_{22}^{-}[f]&:=&\lambda_2\iiint_{\mathbb{R}^{3\times3}}\mathcal{K}^{22}({p}_1,{p}_2,{p}_3,p_4)f(p_2)(1+f({p}_3))(1+f({p}_4))d{p}_2d{p}_3d{p}_4,\end{eqnarray*}
and
$$\mathcal{K}^{22}({p}_1,{p}_2,{p}_3,p_4)=\lambda_2 K^{22}({p}_1,{p}_2,{p}_3,p_4)\delta({p}_1+{p}_2-{p}_3-{p}_4)\delta(\mathcal{E}_{{p}_1}+\mathcal{E}_{{p}_2}-\mathcal{E}_{{p}_3}-\mathcal{E}_{{p}_4}).$$
 We also split $Q$ into the sum of a gain and a loss operators
\begin{equation}\label{QGainLoss}
Q[f]=Q^{\mathrm{gain}}[f]- Q^{\mathrm{loss}}[f]\,,
\end{equation}
where 
\begin{equation*}\label{Subtangent:E2}
Q^{\mathrm{gain}}[f]=C_{12}^{\mathrm{gain}}[f]+C_{22}^{\mathrm{gain}}[f],
\end{equation*}
\begin{equation*}\label{Subtangent:E3}
Q^{\mathrm{loss}}[f]=C_{12}^{\mathrm{loss}}[f]+C_{22}^{\mathrm{loss}}[f],
\end{equation*}
and 
$$Q^{\mathrm{loss}}[f]=fQ^{-}[f],$$
with
$$Q^{-}[f]=C_{12}^{-}[f]+C_{22}^{-}[f].$$
\subsubsection{Conservation of momentum and energy and  the H-Theorem}\label{Sec:MomentEnergy}
In this section, we obtain the basic properties of smooth solutions of \eqref{QB}.  
\begin{lemma}\label{Lemma:WeakFormulation}
There holds 
\begin{eqnarray*}
&&\int_{\mathbb{R}^3}Q[f](p_1)\varphi(p_1)dp_1\\
&=&\iiint_{\mathbb{R}^3\times\mathbb{R}^3\times\mathbb{R}^3} R_{12}[f](p_1, p_2, p_3) 
\Big( \varphi(p_1)-\varphi(p_2)-\varphi(p_3) \Big) \; dp_1dp_2dp_3\\
&&+\frac{1}{2}\iiint_{\mathbb{R}^3\times\mathbb{R}^3\times\mathbb{R}^3\times\mathbb{R}^3} R_{22}[f](p_1, p_2, p_3, p_4) 
\Big( \varphi(p_1) +\varphi(p_2)-\varphi(p_3)-\varphi(p_4) \Big) \; dp_1dp_2dp_3dp_4
,
\end{eqnarray*}
for any smooth test function $\varphi$, where
 \begin{eqnarray*}
R_{12}[f](p_1,p_2,p_3)&=&\lambda_1 n_c K^{12}({p}_1,{p}_2,{p}_3)\delta({p}_1-{p}_2-{p}_3)\delta(\mathcal{E}_{{p}_1}-\mathcal{E}_{{p}_2}-\mathcal{E}_{{p}_3})\\\nonumber
& &\times[(1+f({p}_1))f({p}_2)f({p}_3)-f({p}_1)(1+f({p}_2))(1+f({p}_3))],\\\nonumber
R_{22}[f](p_1,p_2,p_3,p_4)&=&\lambda_2 K^{22}({p}_1,{p}_2,{p}_3,p_4) \delta({p}_1+{p}_2-{p}_3-{p}_4)\delta(\mathcal{E}_{{p}_1}+\mathcal{E}_{{p}_2}-\mathcal{E}_{{p}_3}-\mathcal{E}_{{p}_4})\\\nonumber
& &\times[(1+f({p}_1))(1+f({p}_2))f({p}_3)f({p}_4)\\
&&-f({p}_1)f({p}_2)(1+f({p}_3))(1+f({p}_4))].
\end{eqnarray*}
\end{lemma}
\begin{proof}
By a view of \eqref{QB}, we have 
\begin{eqnarray*}
&&\int_{\mathbb{R}^3}C_{12}[f](p_1)\varphi(p_1)dp_1+ \int_{\mathbb{R}^3}C_{22}[f](p_1)\varphi(p_1)dp_1= I_1+I_2,
\end{eqnarray*}
where 
\begin{eqnarray*}
I_1&:=& \iiint_{\mathbb{R}^3\times\mathbb{R}^3\times\mathbb{R}^3}  \Big(R_{12}[f](p_1, p_2, p_3)-R_{12}[f](p_2, p_1, p_3)-R_{12}[f](p_3, p_2, p_1) \Big ) \varphi(p_1) \; dp_1 dp_2dp_3,\\
 I_2&:=& \iiint_{\mathbb{R}^3\times\mathbb{R}^3\times\mathbb{R}^3\times\mathbb{R}^3}  R_{22}[f](p_1, p_2, p_3, p_4)\varphi(p_1) \; dp_1 dp_2dp_3dp_4.
\end{eqnarray*}
By switching the variables $p_1\leftrightarrow p_2$, $p_1\leftrightarrow p_3$ in the integrals of  $I_1$ and $(p_1,p_2)\leftrightarrow (p_2,p_1)$,  $(p_1,p_2)\leftrightarrow (p_3,p_4)$ in the integrals of  $I_2$,  respectively, as in \cite{ToanBinh,AlonsoGambaBinh,EscobedoVelazquez:2015:FTB},  the lemma follows at once.  
\end{proof}

As a consequence, we obtain the following two corollaries. 

\begin{corollary}[Conservation of momentum and energy] Smooth solutions $f(t,p)$ of \eqref{QB} satisfy 
\begin{eqnarray}\label{Coro:ConservatioMomentum}
\int_{\mathbb{R}^3}f(t,p)pdp&=&\int_{\mathbb{R}^3}f_0(p)pdp\\\label{Coro:ConservatioEnergy}
\int_{\mathbb{R}^3}f(t,p)\mathcal{E}(p)dp&=&\int_{\mathbb{R}^3}f_0(p)\mathcal{E}(p)dp
\end{eqnarray}
for all $t\ge 0$. 
\end{corollary}
\begin{proof} This follows from Lemma \ref{Lemma:WeakFormulation} by taking $\varphi(p) = p$ or $\mathcal{E}(p)$.
 \end{proof}
\begin{corollary}[H-Theorem] Smooth solutions $f(t,p)$ of \eqref{QB} satisfy 
 $$\frac{d}{dt}\int_{\mathbb{R}^3}\left[f(t,p)\log f(t,p)-(1+f(t,p))\log(1+f(t,p))\right]dp\leq 0$$
A radial symmetric equilibrium of the equation has the following form \begin{equation}\label{def-equilibrium}
f_\infty(p)=\frac{1}{e^{c\mathcal{E}(p)}-1}\end{equation}
where $c$ is some positive constant.
\end{corollary}
\begin{proof}
Observe that
$$\partial_t\int_{\mathbb{R}^3}\left[f(t,p)\log f(t,p)-(1+f(t,p))\log(1+f(t,p))\right]dp=\int_{\mathbb{R}^3}\partial_t f(t,p)\log\left(\frac{f(t,p)}{f(t,p)+1}\right)dp,$$
and $$
\begin{aligned}
&\int_{\mathbb{R}^3}Q[f](t,p)\varphi(t,p)dp 
\\=&~~\lambda_1 n_c\iiint_{\mathbb{R}^3\times\mathbb{R}^3\times\mathbb{R}^3}
K^{12}({p}_1,{p}_2,{p}_3)\delta({p}_1-{p}_2-{p}_3)\delta(\mathcal{E}_{{p}_1}-\mathcal{E}_{{p}_2}-\mathcal{E}_{{p}_3})
\\&~~ \times (1+f(t,p_1))(1+f(t,p_2))(1+f(t,p_3))\left( \frac{f(t,p_2)}{f(t,p_2)+1}\frac{f(t,p_3)}{f(t,p_3)+1}-\frac{f(t,p_1)}{f(t,p_1)+1} \right ) \\
&~~\times [\varphi(p_1)-\varphi(p_2)-\varphi(p_3)]dp_1dp_2dp_3
\\&~~+\frac{\lambda_2}{2} \iiint_{\mathbb{R}^3\times\mathbb{R}^3\times\mathbb{R}^3\times\mathbb{R}^3}
{K}^{22}(p_1, p_2, p_3, p_4) \delta({p}_1+{p}_2-{p}_3-{p}_4)\delta(\mathcal{E}_{{p}_1}+\mathcal{E}_{{p}_2}-\mathcal{E}_{{p}_3}-\mathcal{E}_{{p}_4})\\
&~~
\times(1+f(t,p_1))(1+f(t,p_2))(1+f(t,p_3))(1+f(t,p_4))\times\\
&\times\left( \frac{f(t,p_3)}{f(t,p_3)+1}\frac{f(t,p_4)}{f(t,p_4)+1}-\frac{f(t,p_1)}{f(t,p_1)+1}\frac{f(t,p_2)}{f(t,p_2)+1} \right ) 
\\&~~ \times  [\varphi(p_1)+\varphi(p_2)-\varphi(p_3)-\varphi(p_4)]dp_1dp_2dp_3dp_4.
\end{aligned}$$
 Notice that $$(\alpha-\beta)\log \left(\frac{\alpha}{\beta}\right) \ge 0.$$ In the above inequality, the equality holds if and only if $\alpha = \beta$. Now suppose that $f_\infty(p)$ is a radial symmetric equilibrium. By Lemma \ref{Lemma:WeakFormulation} with $\varphi(p)=\log\left(\frac{f_\infty(p)}{f_\infty(p)+1}\right)$, we obtain 
$$\int_{\mathbb{R}^3}Q[f_\infty](p)\varphi(p)dp\leq 0.$$
This yields the inequalities in the H-theorem:
\begin{eqnarray*}
\frac{f_\infty(p_2)}{f_\infty(p_2)+1}\frac{f_\infty(p_3)}{f_\infty(p_3)+1}-\frac{f_\infty(p_1)}{f_\infty(p_1)+1}&=&0,\\
\frac{f_\infty(p_2')}{f_\infty(p_2')+1}\frac{f_\infty(p_1')}{f_\infty(p_1')+1}-\frac{f_\infty(p_4')}{f_\infty(p_4')+1}\frac{f_\infty(p_3')}{f_\infty(p_3')+1}&=&0.
\end{eqnarray*}
Setting   
$h(p)=\log \left(\frac{f_\infty(p)}{f_\infty(p)+1}\right)$, with the notice that $h$ is radial symmetric, we get the following set of equations
\begin{equation}\label{id-h}h(p_2)+h(p_3)=h(p_1),\end{equation}
and
\begin{equation}\label{id-h2}h(p_3')+h(p_4')=h(p_2')+h(p_1').\end{equation}
Let us consider \eqref{id-h}. In particular, by  the conservation law $$p_1=p_2+p_3,$$ the function $h(p)$ possesses the following property $$h(p_2 + p_3) = h(p_2) + h(p_3),$$ for all  $(p_2, p_3) \in \mathbb{R}^6$ satisfying
$$ \mathcal{E} (p_2 + p_3) = \mathcal{E} (p_2) + \mathcal{E} (p_3) .$$
As a consequence, since $h$ is radial symmetric, $$h \circ \mathcal{E}^{-1} (\alpha + \beta) = h \circ \mathcal{E}^{-1} (\alpha) + h \circ \mathcal{E}^{-1} (\beta),$$
where  $p_2 = \mathcal{E}^{-1} (\alpha)$ and $p_3= \mathcal{E}^{-1} (\beta)$. Notice that $\alpha, \beta$ can  take arbitrary values in $\mathbb{R}_+$, which implies $h \circ \mathcal{E}^{-1} (\alpha) =  - c\alpha$ for some positive constant  $c$ and for all $\alpha\ge 0$. Hence $h(p) =  - c\mathcal{E}(p)$, for all $p \in \mathbb{R}^3$. Identity \eqref{def-equilibrium} is proved. \end{proof} 
\subsubsection{Resonance manifolds/energy surfaces}\label{Sec:EnergySurface}
We establish the following estimates on the energy surface integrals on $S_p^1$ and $S_p^2$ following the strategy proposed in \cite{ToanBinh}. 

\begin{lemma}\label{lem-Sp} Let $S_p^0$ be defined as in \eqref{def-Sp}. The following estimate  holds
\begin{equation}\label{lem-Sp-e2} \int_{S_{p}^0}\frac{K^{12}(p,w,p-w)|w|^{k_1}|p-w|^{k_2}}{|\nabla H_0^p(w)|}d\sigma(w) \ge c_1 |p|^{k_1+k_2+1}\;  \min\{ 1, |p|\}^{k_1+k_2+7},\end{equation}
where $k_1, k_2$ is are non-negative constants.\\

Moreover, for any function ${F}(\cdot):\mathbb{R}^3\to\mathbb{R}$ which is  radial  and positive 
$$F(u)=F(|u|),$$  we have 
\begin{equation}\label{lem-Sp-e1} \int_{S_p^0} \frac{{F}(|w|)}{|\nabla H_0^p(|w|)|} d\sigma(w) \le c_2  \int_0^{|p|} |u| {F}(|u|)\; d|u|,\end{equation}
for some positive constant $c_2$ independent of $p$.
\end{lemma}
\begin{proof} By definition $S_p^0$ is the surface containing all $w$ satisfying  $$\mathcal{E}(p-w) + \mathcal{E}(w)= \mathcal{E}(p).$$
For $w=0$ and $p$, the above identity is automatically satisfied, hence $\{0,p\}\subset S_p^0$. If we consider $\mathcal{E}(\varrho)$ as a function of $|\varrho|$: $\mathcal{E}(\varrho)=\mathcal{E}(|\varrho|)$, then
$$\mathcal{E}'(|\varrho|)=\frac{\kappa_1+ 2\kappa_2|\varrho|^2}{\sqrt{\kappa_1+ \kappa_2|\varrho|^2}}>0,$$
 which means that $\mathcal{E}(|\varrho|)$ is strictly increasing. Since for all $w\in S_p^0\backslash\{0,p\}$, $\mathcal{E}(|p-w|)< \mathcal{E}(|p|)$ and $\mathcal{E}(|w|)< \mathcal{E}(|p|)$, by the monotonicity of  $\mathcal{E}(|\varrho|)$, we have $|w| < |p|$ and $|p-w|< |p|$, for all $w\in S_p\backslash\{0,p\}$. As a consequence, the energy surface $S_p^0$ is a subset of  $\overline{B(0,|p|)} \cap \overline{B(p,|p|)}$.
Now, recall $$H_0^p(w): = \mathcal{E}(p-w)  + \mathcal{E}(w) - \mathcal{E}(p).$$
The directional derivative of $H_0^p$ in the direction of $w$ can be computed as
\begin{equation}\label{DG} \nabla_w H_0^p = \frac{ w - p}{|p-w|} \mathcal{E}'(|p -w|) + \frac{w}{|w|} \mathcal{E}'(|w|).\end{equation}
For $w$ of the form $w=\gamma p+q e_0, \gamma,q\in\mathbb{R}_+$, $e_0 \cdot p =0$, the  derivative of $H$ with respect to $q$  is
\begin{equation}\label{Dq-G} \partial_q H_0^p = \partial_q w \cdot \nabla_w H_0^p = e_0 \cdot \nabla_w H =q |e_0|^2 \left[ \frac{\mathcal{E}'(p -w)}{|p-w|} + \frac{\mathcal{E}'(w) }{|w|}\right]   > 0,\end{equation}
which means that $H(w)$ is strictly increasing with respect to $q$. 
\\ For $q=0$ and $\gamma\in(0,1)$, we will show that 
\begin{equation}\label{GgammapNegative}
H_0^p(w)=H_0^p(\gamma p) <0.
\end{equation} Let us start by the following true fact
$$\sqrt{\left(\kappa_1+\kappa_2\gamma^2|p|^2\right)\left(\kappa_1+\kappa_2(1-\gamma)^2|p|^2\right)}<\kappa_1+\kappa_2(\gamma^2-\gamma+2)|p|^2\ \ \ \mbox{ for } p\ne 0.$$
Multiplying both sides of the above inequality with $2\gamma(1-\gamma)|p|^2$ yields
$$2\sqrt{\left(\kappa_1\gamma^2+\kappa_2\gamma^4|p|^2\right)\left(\kappa_1(1-\gamma)^2+\kappa_2(1-\gamma)^4|p|^2\right)}|p|^2< 2\kappa_1\gamma(1-\gamma)|p|^2+2\kappa_2\gamma(1-\gamma)(\gamma^2-\gamma+2)|p|^4.$$
Adding $\kappa_1\gamma^2|p|^2+\kappa_2\gamma^4|p|^4 + \kappa_1(1-\gamma)^2|p|^2+\kappa_2(1-\gamma)^4|p|^4$ to both sides of the above inequality, we obtain \begin{equation*}
\begin{aligned}
&~~\kappa_1\gamma^2|p|^2+\kappa_2\gamma^4|p|^4 + \kappa_1(1-\gamma)^2|p|^2+\kappa_2(1-\gamma)^4|p|^4 \\
&~~ + 2\sqrt{\left(\kappa_1\gamma^2+\kappa_2\gamma^4|p|^2\right)\left(\kappa_1(1-\gamma)^2+\kappa_2(1-\gamma)^4|p|^2\right)}|p|^2
\\
<&~~\kappa_1|p|^2+\kappa_2|p|^4.
\end{aligned}
\end{equation*}
Rearranging the terms in the above inequality and taking the square root gives
$$\sqrt{\kappa_1\gamma^2|p|^2+\kappa_2\gamma^4|p|^4}+\sqrt{\kappa_1(1-\gamma)^2|p|^2+\kappa_2(1-\gamma)^4|p|^4} < \sqrt{\kappa_1|p|^2+\kappa_2|p|^4},$$
and \eqref{GgammapNegative} is proved.
\\ As a consequence, for a unit vector $e_0$ which is orthogonal to $p$, the surface $S_p$ and the  set $\mathcal{P}_\gamma= \{ \gamma p +  qe_0, q\in\mathbb{R}_+ \}$ intersect at only one point, for each $\gamma \in (0,1)$. Define the intersection by $W_\gamma=\gamma p +q_\gamma e_0$. Since
$$\mathcal{E}(p-W_\gamma) + \mathcal{E}(W_\gamma)= \mathcal{E}(p),$$
then $\mathcal{E}(W_\gamma)< \mathcal{E}(p)$; there holds
$$|W_\gamma|=\sqrt{\gamma^2|p|^2+|q_\gamma|^2}<|p|, ~~~ |W_\gamma-p|=\sqrt{(1-\gamma)^2|p|^2+|q_{1-\gamma}|^2}<|p|$$
which implies 
\begin{equation}\label{bd-qa} |q_\gamma|< {|p|},\end{equation} 
and
\begin{equation}\label{WEstimate}\gamma|p|< |W_\gamma|<|p|,~~~~(1-\gamma)|p|<|p-W_\gamma|<|p|.\end{equation} 
Taking the derivative with respect to $\gamma$ of the identity
$$H_0^p(W_\gamma) =0$$
 yields: 
\begin{equation}\label{dG-qa}
\begin{aligned}
 0 &= \partial_\gamma W_\gamma \cdot \nabla_w H_0^p=\partial_\gamma W_\gamma \cdot \left(\frac{ W_\gamma - p}{|p-W_\gamma|} \mathcal{E}'(|p -W_\gamma|) + \frac{W_\gamma}{|W_\gamma|} \mathcal{E}'(|W_\gamma|)\right)\\
 &= \frac12 \partial_\gamma |W_\gamma|^2  \left[ \frac{\mathcal{E}'(|p -W_\gamma|)}{|p-W_\gamma|} + \frac{\mathcal{E}'(|W_\gamma|) }{|W_\gamma|}\right]  - |p|^2 \frac{\mathcal{E}'(|p -W_\gamma|)}{|p-W_\gamma|} 
 \\
 &= 
  \frac12 \partial_\gamma |q_\gamma|^2  \left[ \frac{\mathcal{E}'(|p -W_\gamma|)}{|p-W_\gamma|} + \frac{\mathcal{E}'(|W_\gamma|) }{|W_\gamma|}\right]  + \gamma |p|^2 \frac{\mathcal{E}'(|W_\gamma|) }{|W_\gamma|} - (1-\gamma) |p|^2 \frac{\mathcal{E}'(|p -W_\gamma|)}{|p-W_\gamma|} 
 \end{aligned}\end{equation}
 where the identities $\partial_\gamma W_\gamma=p$,  $|W_\gamma|^2 = \gamma^2 |p|^2 + |q_\gamma|^2$ have been used. 
 \\ With the notice that $\mathcal{E}'(|W_\gamma|)>0$, the above identity yields
 \begin{equation}\label{da-w} \frac 12 \partial_\gamma |q_\gamma|^2 \le (1-\gamma) |p|^2 \end{equation}
 for all $p$ and all $\gamma \in (0,1)$. 
\\ We now provide an estimate on $q_\gamma$. In order to do this, let us consider two cases $|p|\geq1$ and $|p|<1$. 
\begin{itemize}
\item Case 1: $|p|\geq1$. Observe that at $\gamma =\frac12 $, due to the symmetry of the geometry
 $$|W_{1/2}| = |W_{1/2} - p|,$$ 
which implies $$2\cE(W_{1/2}) = \cE(p).$$ Noting that $|W_{1/2}|^2 = \frac 14 |p|^2 + |q_{1/2}|^2$, yields
$$4\left[\kappa_1\left(\frac 14 |p|^2 + |q_{1/2}|^2\right)+\kappa_2\left(\frac 14 |p|^2 + |q_{1/2}|^2\right)^2\right]=\kappa_1|p|^2+\kappa_2|p|^4,$$
then
$$ \kappa_2\Big( \frac 14 |p|^2 + |q_{1/2}|^2 \Big)^2  +  \kappa_1|q_{1/2}|^2 =  \frac{\kappa_2}4 |p|^4,$$
which implies 
\begin{equation}\label{q12}c_0  |p|^2 = c_0  |p|^2 \min\Big \{ 1, |p|^2\Big \}
\le |q_{1/2}|^2 \le C_0  |p|^2 \min\Big \{ 1, |p|^2\Big \}=C_0  |p|^2 \end{equation}
for some constants $c_0, C_0$, independent of $|p|$. 
\\ Combining \eqref{da-w}, \eqref{q12} and the fact that
$$ |q_\gamma|^2 = |q_{1/2}|^2 - \int_\gamma^{\frac12} \partial_{\gamma'} |q_{\gamma'}|^2 \; d\gamma' $$
yields
\begin{equation}\label{q-largep}|q_\gamma|^2 \ge c_0|p|^2 - 2\left|\gamma - \frac 12\right| |p|^2 \ge \frac 12 c_0 |p|^2\end{equation}
for all $\gamma$ satisfying $\left|\gamma - \frac 12 \right|\le \frac{c_0}{4}$. 
\item Case 2: $|p|$ is  small. Recall that 
\begin{equation}\label{q12bb}\begin{aligned}
\Big(\mathcal{E}(w) &+ \mathcal{E}(p-w)\Big)^2  - \mathcal{E}(p)^2
\\&=  \kappa_1 (|p-w|^2+|w|^2-|p|^2)  + \kappa_2 (|p-w|^4  +|w|^4-|p|^4) + 2 \mathcal{E}(w) \mathcal{E}(p-w) 
\\&=  2\kappa_1w\cdot (w-p)  + 2 \kappa_2 w\cdot (w-p) \Big(|w|^2 + |w-p|^2+|p|^2\Big) 
\\&\quad  - 2\kappa_2 |w|^2 |p-w|^2+ 2 \mathcal{E}(w) \mathcal{E}(p-w) ,
\end{aligned}\end{equation}
which leads to
\begin{equation}\label{q12b}- w\cdot (w-p) \Big( \kappa_1+ \kappa_2 |w|^2 + \kappa_2|w-p|^2+\kappa_2|p|^2\Big) = \mathcal{E}(w) \mathcal{E}(p-w)  - \kappa_2 |w|^2 |p-w|^2
\end{equation}
for all $w\in S_p$, in which the right hand side can be computed explicitly as
$$\begin{aligned}
 \mathcal{E}(w) &\mathcal{E}(p-w)  - \kappa_2 |w|^2 |p-w|^2
 \\&= 
  |w| |p-w| \sqrt{(\kappa_1 + \kappa_2 |w|^2 )(\kappa_1 + \kappa_2|w-p|^2)} - \kappa_2 |w|^2 |p-w|^2
 \\& =  |w| |p-w| \frac{ \kappa_1 \Big( \kappa_1 + \kappa_2 |w|^2 + \kappa_2|w-p|^2 \Big)}{\sqrt{(\kappa_1 + \kappa_2 |w|^2 )(\kappa_1 + \kappa_2|w-p|^2)} + \kappa_2 |w| |p-w|} .
\end{aligned}$$
 We will develop an asymptotic expansion of the above expression in term of $|p|$. In order to do this, we observe that
$$\sqrt{\left(1 + \frac{\kappa_2}{\kappa_1} |w|^2 \right)\left(1 + \frac{\kappa_2}{\kappa_1} |w-p|^2\right)} =1+\frac{\kappa_2}{2\kappa_1} (|w|^2+|w-p|^2)  +\mathcal{O}(|p|^4),$$
which leads to
\begin{equation}\label{q12c}\begin{aligned}
 \mathcal{E}(w) &\mathcal{E}(p-w)  - \kappa_2 |w|^2 |p-w|^2
 \\& =  |w| |p-w| \Big( \kappa_1 + \kappa_2 |w|^2 + \kappa_2|w-p|^2 \Big) 
 \\&\quad \times \Big( 1 -\frac12 \frac{ \kappa_2}{\kappa_1} ( |w|^2 +|w-p|^2) - \frac{ \kappa_2}{\kappa_1} |w||w-p| + \mathcal{O}(|p|^4)\Big) 
 \\& =  |w| |p-w| \Big( \kappa_1 + \frac 12\kappa_2 |w|^2 + \frac12\kappa_2|w-p|^2- { \kappa_2} |w||w-p| +\mathcal{O}(|p|^4)\Big)
  \\& =  |w| |p-w| \Big( \kappa_1 + \kappa_2 |w|^2 + \kappa_2|w-p|^2 + \kappa_2 |p|^2 \Big) -
\\&~~~- \frac{\kappa_2}{2} |w| |w-p| \Big(|w|^2 + |w-p|^2+2|w||w-p|+2|p|^2\Big)\Big(1+ \mathcal{O}(|p|^2)\Big).
   \end{aligned}\end{equation}
Define $\rho_\gamma$ be the angle between $W_\gamma$ and $W_\gamma -p$, then $W_\gamma\cdot (W_\gamma-p) = |W_\gamma| |W_\gamma-p| \cos \rho_\gamma$, which, together with \eqref{q12b}-\eqref{q12c}, leads to
 $$ 1 + \cos \rho_\gamma =  \frac{\kappa_2}{2} \frac{\Big(|W_\gamma|^2 + |W_\gamma-p|^2+2|W_\gamma||W_\gamma-p|+2|p|^2\Big)\Big(1+ \mathcal{O}(|p|^2)\Big) }{\kappa_1 + \kappa_2 |W_\gamma|^2 + \kappa_2|W_\gamma-p|^2 + \kappa_2 |p|^2} = \mathcal{O}(|p|^2).$$
 Hence $\sin \rho_\gamma=\mathcal{O}(|p|)$. The area of the parallelogram formed by $W_\gamma$ and $W_\gamma -p$ can be computed as 
 $$ 2 |p| |q_\gamma| = |W_\gamma \times (W_\gamma -p)| = |W_\gamma| |W_\gamma-p| \sin \rho_\gamma,
 $$ 
which, together with \eqref{WEstimate}, implies that there exist  universal constants $c_2,c_3$ satisfying 
\begin{equation}\label{q-smallp} c_3\gamma (1-\gamma) |p|^2 \le |q_\gamma|\le c_2 |p|^2\end{equation}
for all $\gamma \in (0,1)$.
\end{itemize}
 The two inequalities \eqref{q-largep} and \eqref{q-smallp} are the two estimates we need to obtain \eqref{lem-Sp-e2}. To continue, we  parametrize the surface $S_p^0$ as follows: We choose $p^\perp$ to be a vector in $\mathcal{P}_0 = \{ p\cdot q =0\}$ and  $e_\theta$ to  be the unit vector in $ \mathcal{P}_0$ so that the angle between $p^\perp$ and $e_\theta$ is $\theta$. The surface $S_p$ can be represented as  
 $$ S_p^0 = \Big\{ W (\gamma,\theta) = \gamma p + |q_\gamma| e_\theta ~:~ \theta \in [0,2\pi], ~\gamma \in [0,1]\Big\} .$$
Notice that the vector $\partial_\theta e_\theta$ is orthogonal to both vectors  $p$ and $e_\theta$,  the surface area can be computed as 
\begin{equation}\label{dS}\begin{aligned}
 d \sigma (w) &= |\partial_\gamma W_\gamma \times \partial_\theta W_\gamma| d\gamma d\theta  = \Big |(p+\partial_\gamma |q_\gamma| e_\theta) \times |q_\gamma| \partial_\theta e_\theta \Big | d\gamma d\theta
 \\
&= \Big |(|q_\gamma | p+\frac12\partial_\gamma  |q_\gamma |^2 e_\theta) \times  \partial_\theta e_\theta \Big | d\gamma  d\theta
\\
 &=
 \sqrt{|p|^2  |q_\gamma | ^2+ \frac14| \partial_\gamma  (|q_\gamma |^2)|^2}d\gamma  d\theta.
 \end{aligned}\end{equation}
 
 It is straightforward from the  identity \eqref{dG-qa} that
 \begin{equation}\label{proposition:T2L2:E8}
\begin{aligned}
  \partial_\gamma |q_\gamma|^2  &
 &= 
2|p|^2 \frac{\gamma   \frac{\mathcal{E}'(|W_\gamma|) }{|W_\gamma|} +(\gamma-1)  \frac{\mathcal{E}'(|p -W_\gamma|)}{|p-W_\gamma|}}{\frac{\mathcal{E}'(|p -W_\gamma|)}{|p-W_\gamma|} + \frac{\mathcal{E}'(W_\gamma) }{|W_\gamma|}}.
 \end{aligned}\end{equation}

Now, let us compute $|\nabla H_0^p|$ under the new parametrization
\begin{equation*}
\begin{aligned}
 |\nabla H_0^p|^2 \ =  &\  |p|^2\left[\gamma   \frac{\mathcal{E}'(|W_\gamma|) }{|W_\gamma|} +(\gamma-1)  \frac{\mathcal{E}'(|p -W_\gamma|)}{|p-W_\gamma|}\right]^2\\
 & \ + |q_\gamma|^2 \left[\frac{\mathcal{E}'(|p -W_\gamma|)}{|p-W_\gamma|} + \frac{\mathcal{E}'(|W_\gamma|) }{|W_\gamma|}\right]^2,
\end{aligned}
\end{equation*}
which, in companion with \eqref{proposition:T2L2:E8}, implies
\begin{equation}\label{proposition:T2L2:E9}
\begin{aligned}
 |\nabla  H_0^p|^2 \ =  &\  \frac{\left|\partial_\gamma |q_\gamma|^2\right|^2}{4|p|^2}\left[\frac{\mathcal{E}'(|p -W_\gamma|)}{|p-W_\gamma|} + \frac{\mathcal{E}'(|W_\gamma|) }{|W_\gamma|}\right]^2\\
 & \ + |q_\gamma|^2 \left[\frac{\mathcal{E}'(|p -W_\gamma|)}{|p-W_\gamma|} + \frac{\mathcal{E}'(|W_\gamma|) }{|W_\gamma|}\right]^2,
\end{aligned}
\end{equation}

Using the fact that $\mathcal{E}'(x)\ge c x $ for all $x \in \mathbb{R}_+$, we get the following lower bound on  $|\nabla H_0^p|$
\begin{equation}\label{proposition:T2L2:E10}
\begin{aligned}
 |\nabla H_0^p| \ =  &\  \frac{\sqrt{ \frac{\left|\partial_\gamma |q_\gamma|^2\right|^2}{4} +  |q_\gamma|^2|p|^2}}{|p|}\left[\frac{\mathcal{E}'(|p -W_\gamma|)}{|p-W_\gamma|} + \frac{\mathcal{E}'(|W_\gamma|) }{|W_\gamma|}\right].
\end{aligned}
\end{equation}

With  \eqref{q-largep} and  \eqref{q-smallp}, we are now able to estimate the integral
$$Z:=\int_{S_{p}^0}\bar{K}_0^{12}(p,w,p-w)|w|^{k_1}|p-w|^{k_2}d\sigma(w).$$
Notice that $$K^{12}(p,w,p-w)\geq C(|p|\wedge p_0)(|p-w|\wedge p_0)(|w|\wedge p_0) ,$$
where $C$ is some positive constant varying from line to line. As a result, $Z$ can be bounded from below by $CZ'$, where  $Z'$ is defined as
$$Z':=\int_{S_{p}^0}(|p|\wedge p_0)(|w|\wedge p_0)(|p-w|\wedge p_0)|w|^{k_1}|p-w|^{k_2}d\sigma(w).$$
By \eqref{dS}, $Z'$ can be rewritten as
$$\int_0^{2\pi}\int_0^1\frac{|p|(|p|\wedge p_0)(|w|\wedge p_0)(|p-w|\wedge p_0)|w|^{k_1}|p-w|^{k_2}}{\frac{\mathcal{E}'(|p-w|}{|p-w|}+\frac{\mathcal{E}'(|w|)}{|w|}}d\gamma  d\theta.$$
Due to \eqref{q-largep}, for $|p|$ large, and $\gamma\in\left[\frac{2-c_0}{4},\frac{2+c_0}{4}\right]$, 
$$|W_\gamma|^2\ge|q_\gamma|^2 \ge \frac 12 c_0 |p|^2$$
and
$$|p-W_\gamma|^2\ge|q_\gamma|^2 \ge \frac 12 c_0 |p|^2.$$

We therefore can bound
$$\frac{|p|}{\frac{\mathcal{E}'(|p-w|)}{|p-w|}+\frac{\mathcal{E}'(|w|)}{|w|}} \ge \frac{|p|}{2\frac{\mathcal{E}'(c_0 |p|)}{c_0 |p|}} \ge \frac{c_0|p|^2}{2\frac{\kappa_1+2\kappa_2|c_0p|^2}{\sqrt{\kappa_1+\kappa_2|c_0p|^2}}},$$
where we have use the fact that $\frac{\mathcal{E}'(|\varrho|)}{|\varrho|}$ is decreasing with respect to $|\varrho|$. Since $|p|$ is large, 
$$\frac{|p|}{\frac{\mathcal{E}'(|p-w|)}{|p-w|}+\frac{\mathcal{E}'(|w|)}{|w|}} \ge C|p|,$$
for some positive constant $C>0$. 

Therefore, $Z'$ can be estimated as follows
\begin{eqnarray*}
Z'&\geq& C\int_0^{2\pi}\int_{\frac{1-c_0}{2}}^{\frac{1+c_0}{2}}(|p|\wedge p_0)\left(\left|\sqrt{\frac{c_0}{2}}|p|\right|\wedge p_0\right)^2\left|\sqrt{\frac{c_0}{2}}|p|\right|^{k_1+k_2} |p|d\gamma  d\theta\\
&\geq&C(|p|\wedge1)^3|p|^{k_1+k_2+1}\\
&\geq&C|p|^{k_1+k_2+1},
\end{eqnarray*}
where $C$ is some positive constant varying from line to line.\\
Thanks to \eqref{q-smallp}, for $p$ small, on the interval $\gamma\in \left[\frac{1}{3},\frac{1}{2}\right]$, 
$$|W_\gamma|^2\ge|q_\gamma|^2 \ge c_1 |p|^4$$
and
$$|p-W_\gamma|^2\ge|q_\gamma|^2 \ge c_1  |p|^4.$$

We therefore can bound
$$\frac{|p|}{\frac{\mathcal{E}'(|p-w|)}{|p-w|}+\frac{\mathcal{E}'(|w|)}{|w|}} \ge \frac{|p|}{2\frac{\mathcal{E}'(c_1 |p|^2)}{c_1 |p|^2}} \ge \frac{c_1|p|^3}{2\frac{\kappa_1+2\kappa_2c_1^2|p|^4}{\sqrt{\kappa_1+\kappa_2c_1^2|p|^4}}},$$
where we have use the fact that $\frac{\mathcal{E}'(|\varrho|)}{|\varrho|}$ is decreasing with respect to $|\varrho|$. Since $|p|$ is small, 
$$\frac{|p|}{\frac{\mathcal{E}'(|p-w|)}{|p-w|}+\frac{\mathcal{E}'(|w|)}{|w|}} \ge C|p|^3,$$
for some positive constant $C>0$. 

Therefore, $Z'$ can be estimated as follows
\begin{eqnarray*}
Z'&\geq& \int_0^{2\pi}\int_{\frac{1}{3}}^{\frac{1}{2}}(|p|\wedge p_0)\left(\left|\sqrt{{c_1}}|p|^2\right|\wedge p_0\right)^2\left|{\sqrt{c_1}}|p|\right|^{2k_1+2k_2} C|p|^3d\gamma  d\theta\\
&\geq&C(|p|\wedge 1)^5|p|^{2k_1+2k_2+3}\\
&\geq&C|p|^{2k_1+2k_2+8}.
\end{eqnarray*}
The above shows that \eqref{lem-Sp-e2} holds true.\\

As for the surface integral of a radial function $\mathcal{G}(|w|)$, we  introduce the radial variable $ u = |W_\alpha | = \sqrt{\alpha^2 |p|^2 + |q_\alpha|^2}$. We compute $2u du =\partial_\alpha |W_\alpha|^2 d\alpha$ and hence 
$$\frac{1}{|\nabla H_0^p|}d \sigma (w) = \frac{|p|}{2\left[\frac{\mathcal{E}'(|p -W_\gamma|)}{|p-W_\gamma|} + \frac{\mathcal{E}'(|W_\gamma|) }{|W_\gamma|}\right]\partial_\alpha |W_\alpha|^2 } u du d\theta .
$$
Using \eqref{dG-qa}, we compute 
$$\begin{aligned}
\frac{|p|}{2\left[\frac{\mathcal{E}'(|p -W_\gamma|)}{|p-W_\gamma|} + \frac{\mathcal{E}'(|W_\gamma|) }{|W_\gamma|}\right]\partial_\alpha |W_\alpha|^2 }
& =  \frac{1}{4|p|\frac{\mathcal{E}'(|p -W_\gamma|)}{|p-W_\gamma|}} 
&\le  \frac{1}{4{\mathcal{E}'(|p|)}},
\end{aligned}$$
where we have used the fact that $\frac{\mathcal{E}'(|p -W_\gamma|)}{|p-W_\gamma|}\ge \frac{\mathcal{E}'(|p|)}{|p|}$.
This, together with the bound
$$\mathcal{E}'(|p|)\ge C,$$
for some positive constant $C$,  proves that 
\begin{equation}\label{dS-du} d \sigma (w) \le C u du d\theta,
\end{equation}
for some positive constants $C$. This yields the upper bound on the surface integral.  \end{proof}

\begin{lemma}\label{lem-Sp1} Let $S_p^1$ be defined as in \eqref{def-Sp} and $F$ be an arbitrary function in $ L^1_1(\mathbb{R}_+)$. 
There is a positive constant $C_0$ so that 
$$  \int_{S_p^1} \frac{{F}(|w|)}{|\nabla H_1^{p_1}(|w|)|} \; d\sigma (w) \le \frac{C_0}{|p|} \|  uF(\cdot)\|_{L^1(\mathbb{R}_+)}$$
uniformly in $p\in \mathbb{R}^3.$
\end{lemma}
  \begin{proof} Let us recall  that $S_p^1$ is the surface consisting of $w$ satisfying $$\mathcal{E}(p+w)  =  \mathcal{E}(w) + \mathcal{E}(p).$$  
First, we compute \begin{equation}\label{wp-com} \begin{aligned}
\mathcal{E}&(p+w)^2 - \Big( \mathcal{E}(p) + \mathcal{E}(w)\Big)^2  
\\&=   \kappa_1 |p+w|^2 + \kappa_2 |p+w |^4  - \kappa_1 (|p|^2 + |w|^2) 
 - \kappa_2 (|p|^4 + |w|^4) - 2\mathcal{E}(p) \mathcal{E}(w)
\\&=  2 \kappa_1w\cdot p  +2 \kappa_2 w\cdot p (|p|^2 + |w|^2 + |p+w|^2) 
 + 2 \kappa_2  |p|^2 |w|^2 - 2\mathcal{E}(p) \mathcal{E}(w).
\end{aligned}\end{equation}
Now, since $\kappa_1 \not =0$, it follows that  $ \kappa_2  |p|^2 |w|^2  < \mathcal{E}(p) \mathcal{E}(w)$. This, in combination with \eqref{wp-com}, proves that if $w \in S_p^1 \setminus  \{ 0\}$, then $w\cdot p >0$. Let us calculate the derivative of $H_1^{p_1}$   
$$ \nabla_w  H_1^{p} = \frac{p+w}{|p+w|} \mathcal{E}'(|p+w|) - \frac{w}{|w|} \mathcal{E}'(|w|),$$
where $$\mathcal{E}'(|\varrho|)=\frac{2\kappa_1+ 4\kappa_2|\varrho|^2}{\sqrt{\kappa_1+ \kappa_2|\varrho|^2}}.$$
The  derivative at $w = \gamma p$ with $\gamma \in \mathbb{R}_+$ can be determined using the previous formulation
$$ \partial_\gamma H_1^{p} =\partial_\gamma w\cdot\nabla_w H_1^{p} |_{w=\gamma p}=|p|  \mathcal{E}'((1+\gamma )p ) - |p|  \mathcal{E}'(\gamma p) .$$
By  the monotonicity of  $\mathcal{E}(p)$ with respect to the length $|p|$, it follows that $\mathcal{E}'((1+\gamma )p ) > \mathcal{E}'(\gamma p)$, and hence $\partial_\gamma H_1^{p} >0$ for all $\gamma > 0$. Since $H_1^{p}(0) = 0$, $H_1^{p}(\gamma p)> 0$ for all positive $\gamma$.

Now, let us consider  all the points $W_\gamma =\gamma p + q$, with $q \cdot p = 0$, for each fixed $\gamma>0$. The directional derivative of $H_1^{p}$ at $W_\gamma = \gamma p+q$ in the direction of $q\not= 0$ satisfies 
$$ q \cdot \nabla_w H_1^{p} = |q|^2 \left[ \frac{\mathcal{E}'(p+W_\gamma)}{|p+W_\gamma|} - \frac{\mathcal{E}'(W_\gamma) }{|W_\gamma|}\right]  < 0$$
in which the fact that $\mathcal{E}'(p)/|p|$ is strictly decreasing in $|p|$ has been used. By a view of \eqref{wp-com}, the sign of $H_1^{p}(w)$, with $W_\gamma = \gamma p + q$, is the same with the quantity 
$$
\begin{aligned}
\gamma |p|^2 & \Big( \kappa_1 + \kappa_2 (|p|^2 + |W_\gamma|^2 + |p+W_\gamma|^2) \Big) +
  \kappa_2  |p|^2 |W_\gamma|^2  - \mathcal{E}(p) \mathcal{E}(W_\gamma) 
\\ &= \gamma |p|^2  \Big( \kappa_1 + 2\kappa_2 (|p|^2 + \gamma |p|^2+ |W_\gamma|^2 ) \Big) 
\\&\quad -  \frac{ (\kappa_1 |p|^2 + \kappa_2 |p|^4) (\kappa_1 |W_\gamma|^2 + \kappa_2 |W_\gamma|^4) - \kappa_2 ^2 |p|^4 |W_\gamma|^4}
 {\sqrt{\kappa_1 |p|^2 + \kappa_2 |p|^4}\sqrt{\kappa_1 |W_\gamma|^2 + \kappa_2 |W_\gamma|^4}  + \kappa_2  |p|^2 |W_\gamma|^2 }
 \\ &= \gamma |p|^2  \Big( \kappa_1 +  2\kappa_2 (|p|^2 + \gamma |p|^2+ |W_\gamma|^2 ) \Big) 
\\&\quad -  \frac{ \kappa_1^2 |W_\gamma|^2  |p|^2 + \kappa_1\kappa_2 |W_\gamma|^2|p|^4  + \kappa_1 \kappa_2 |p|^2 |W_\gamma|^4 }
 {\sqrt{\kappa_1 |p|^2 + \kappa_2 |p|^4}\sqrt{\kappa_1 |W_\gamma|^2 + \kappa_2 |W_\gamma|^4}  + \kappa_2  |p|^2 |W_\gamma|^2 }.
 \end{aligned}$$
This yields that $H_1^{p}(\gamma p + q)<0$ as long as 
$$
\begin{aligned}
 \gamma  & <  \frac{\kappa_1  (\kappa_1 + \kappa_2 |p|^2)  + \kappa_1 \kappa_2 |W_\gamma|^2}
 {\sqrt{\kappa_1 |p|^2 + \kappa_2 |p|^4}\sqrt{\frac{\kappa_1}{|W_\gamma|^2}+ \kappa_2}  + \kappa_2  |p|^2 } \frac{1}{ \Big( \kappa_1 +  2\kappa_2 (|p|^2 + \gamma |p|^2+ |W_\gamma|^2 ) \Big)}.
 \end{aligned}$$
 
Since $H_1^{p}(\gamma p)>0$ and $q \cdot \nabla_w H_1^{p}<0$, for a given direction $q$, there exists $q_\gamma$, such that $q_\gamma\cdot q>0$ and $q_\gamma$ is parallel with $q$, if an only if 
\begin{equation}
\label{qinfinity} \lim_{q\to \infty}H_1^{p}(\gamma p + q) <0
\end{equation}

Taking $q \to \infty$ (and so $|W_\gamma|\to \infty$), we obtain \eqref{qinfinity}  if and only if    
\begin{equation}\label{range-a}\gamma < \gamma_p:= 
\frac12 \frac{ \kappa_1 }
 { \kappa_2|p|^2 + 2\sqrt {\kappa_2} \sqrt{\kappa_1 |p|^2 + \kappa_2 |p|^4} }.\end{equation}
 In particular, we note that 
\begin{equation}\label{bd-ap}\gamma_p|p| (1+|p|) \le C_0, \qquad \forall ~ p\in \RR^3\end{equation}
for some positive constant $C_0$.

Hence, for positive values of $\gamma$ satisfying \eqref{range-a}, there is a unique $|q_v|$ so that $\overline G(\gamma p + q) =0$, for all $|q| = |q_\gamma|$. Moreover, from the continuity of $H_1^{p}(W_\gamma)$, $|q_\gamma|$ is continuously differentiable with respect to $\gamma$. For $\gamma\ge \gamma_p$, $H_1^{p}(\gamma p + q)>0$, for all $q$ so that $q\cdot p =0$.  

Now, we can parametrize  the surface $S^1_p$ as follows:
\begin{equation}S^1_p =\Big \{ w(\gamma,\theta) = \gamma p + |q_\gamma| e_\theta~:~\gamma \in [0, \gamma_p),~ \theta \in [0,2\pi]\Big\},\end{equation}
in which $\gamma_p$ and $|q_\gamma|$ are defined as above and $e_\theta$ is the unit vector rotating around $p$ and on the orthogonal plane to $p$. As in \eqref{dS}, we have
$$
 d \sigma (w) = \sqrt{|p|^2  |q_\gamma| ^2+ \frac14| \partial_\gamma (|q_\gamma|^2)|^2}d\gamma d\theta 
 $$
and hence, the surface integral is estimated by 
$$ 
\begin{aligned}
\int_{S_p^1}\frac{ F(|w|)}{|\nabla H_1^{p}(w)|} \; d \sigma(w) 
&=\iint_{ [0,2\pi]\times [0,\gamma_p]} \frac{F(|\gamma p + q_\gamma|)}{|\nabla H_1^{p}(|\gamma p + q_\gamma|)|}  \sqrt{|p|^2  |q_\gamma| ^2+ \frac14| \partial_\gamma (|q_\gamma|^2)|^2}d\gamma d\theta .
\end{aligned}$$
Let us introduce the variable $ u = |W_\gamma | = \sqrt{\gamma^2 |p|^2 + |q_\gamma|^2}$. We compute 
$$ 2u du =\partial_\gamma |W_\gamma|^2 d\gamma$$ 
and hence \begin{equation}\label{bd-intSp2} 
\begin{aligned}
\int_{S_p^1} \frac{ F(|w|)}{|\nabla H_1^{p}(w)|} \; d \sigma(w) \le 2\pi  \int_0^\infty F(u) \frac{ \sqrt{|p|^2  |q_\gamma| ^2+ \frac14| \partial_\gamma (|q_\gamma|^2)|^2} }{2\partial_\gamma |W_\gamma|^2|\nabla H_1^{p}(|\gamma p + q_\gamma|)| } u du .
\end{aligned}\end{equation}
We recall that $H_1^{p}(W_\gamma) =0$ and hence 
$$0 =  \partial_\gamma W_\gamma \cdot \nabla_w H_1^{p} = \frac12 \partial_\gamma|W_\gamma|^2  \Big[ \frac{\mathcal{E}'(p+w_\gamma)}{|p+W_\gamma|} - \frac{\mathcal{E}'(W_\gamma) }{|w_\gamma|}\Big] + |p|^2 \frac{\mathcal{E}'(p+w_\gamma)}{|p+W_\gamma|} $$
which leads to \begin{equation}\label{bd-low123}|p|^2 \frac{\mathcal{E}'(p+W_\gamma)}{|p+W_\gamma|} = \frac12 \partial_\gamma|W_\gamma|^2 \left[\frac{\mathcal{E}'(W_\gamma) }{|W_\gamma|}-\frac{\mathcal{E}'(p+W_\gamma)}{|p+W_\gamma|}\right],\end{equation}
and 
\begin{equation}\label{bd-low123b}
\partial_\gamma |q_\gamma|^2 = 2\frac{-\gamma|p|^2\frac{\mathcal{E}'(|W_\gamma|)}{|W_\gamma|}+(1+\gamma)|p|^2\frac{\mathcal{E}'(|p+W_\gamma|)}{|p+W_\gamma|}}{\left[\frac{\mathcal{E}'(W_\gamma) }{|W_\gamma|}-\frac{\mathcal{E}'(p+W_\gamma)}{|p+W_\gamma|}\right]}.
\end{equation}

We deduce from \eqref{bd-low123b} that
 \begin{equation*}
\begin{aligned}
 |\nabla  H_1^{p}|^2 \ =  &\  |p|^2\left[-\gamma   \frac{\mathcal{E}'(|W_\gamma|) }{|W_\gamma|} +(\gamma+1)  \frac{\mathcal{E}'(|p+W_\gamma|)}{|p+W_\gamma|}\right]^2 \ + |q_\gamma|^2 \left[\frac{\mathcal{E}'(|p +W_\gamma|)}{|p+W_\gamma|} - \frac{\mathcal{E}'(|W_\gamma|) }{|W_\gamma|}\right]^2\\
  =  &\  \frac{\left|\partial_\gamma |q_\gamma|^2\right|^2}{4|p|^2}\left[\frac{\mathcal{E}'(|p +W_\gamma|)}{|p+W_\gamma|} - \frac{\mathcal{E}'(|W_\gamma|) }{|W_\gamma|}\right]^2 \ + |q_\gamma|^2 \left[\frac{\mathcal{E}'(|p +W_\gamma|)}{|p+W_\gamma|} - \frac{\mathcal{E}'(|W_\gamma|) }{|W_\gamma|}\right]^2,
\end{aligned}
\end{equation*}
which implies 
\begin{equation*}
\begin{aligned}
\int_{S_p^1} \frac{ F(|w|)}{|\nabla H_1^{p}(w)|} \; d \sigma(w) \le \pi  \int_0^\infty  \frac{F(u)|p|}{\partial_\gamma |W_\gamma|^2\left[\frac{\mathcal{E}'(|W_\gamma|)}{|W_\gamma|}-\frac{\mathcal{E}'(|p+W_\gamma|)}{|p+W_\gamma|}\right] } u du .
\end{aligned}\end{equation*}

The above and \eqref{bd-low123} yield
\begin{equation*}
\begin{aligned}
\int_{S_p^1} \frac{ F(|w|)}{|\nabla H_1^{p}(w)|} \; d \sigma(w) \le \pi  \int_0^\infty  \frac{F(u)}{2 |p| \frac{\mathcal{E}'(|p+w|)}{|p+w|}} u du.
\end{aligned}\end{equation*}

Using the fact that $$\frac{\mathcal{E}'(|p+w|)}{|p+w|}\ge C$$
for some positive constant $C$, we obtain
\begin{equation*}
\begin{aligned}
\int_{S_p^1} \frac{ F(|w|)}{|\nabla H_1^{p}(w)|} \; d \sigma(w) \le \frac{C}{|p|} \int_0^\infty  F(u)u du.
\end{aligned}\end{equation*}

\end{proof}
\begin{lemma}\label{lem-Sp2} Let $S_p^2$ be defined as in \eqref{def-Sp} and $F$ be an arbitrary function in $L^1_1(\mathbb{R}_+)$. 
There are positive constants $C_0$ so that 
$$  \int_{S_p^2} \frac{{F}(|w|)}{|\nabla H_2^{p}(|w|)|} \; d\sigma (w) \le \frac{ C_0}{|p|} \|  uF(\cdot)\|_{L^1(\mathbb{R}_+)}$$
uniformly in $p\in \mathbb{R}^3.$
\end{lemma}
  \begin{proof} We observe that 
\begin{eqnarray*}
S^2_p&=&\{p_*~~|~~\mathcal{E}(p_*)=\mathcal{E}(p)+\mathcal{E}(p_*-p)\}\\
&=&\{p_*+p~~|~~\mathcal{E}(p_*+p)=\mathcal{E}(p)+\mathcal{E}(p_*)\}\\
&=&p+S_p^1.
\end{eqnarray*}
The above identity means that the same argument of Lemma \ref{lem-Sp1} could be applied and the conclusion of the lemma follows.
\end{proof}

\subsubsection{Boundedness of the total mass for the kinetic equation}\label{Sec:Mass}
\begin{proposition}\label{Propo:Mass} Suppose that the positive radial initial condition $f_0(p)=f_0(|p|)$ satisfies
$$\int_{\mathbb{R}^3}f_0(p_1)dp_1<\infty, \int_{\mathbb{R}^3}f_0(p_1)\mathcal{E}(p_1)dp_1<\infty.$$
There exist universal positive constants $\mathcal{C}_1$, $\mathcal{C}_2$ such that  
the mass of the positive radial solution $f(t,p)=f(t,|p|)$ of \eqref{QB} could be bounded as
$$\int_{\mathbb{R}^3}f(t,p_1)dp_1\leq \mathcal{C}_1e^{\mathcal{C}_2t}.$$
\end{proposition}
\begin{proof}
First, observe that the constant function $1$ can be used as the test function for \eqref{QB}, to get
\begin{equation}\label{Propo:Mass:E1}
\begin{aligned}
\frac{d}{dt}\int_{\mathbb{R}^3} f(p_1)dp_1 ~~&=~~ \int_{\mathbb{R}^3}C_{12}[f](p_1)dp_1+\int_{\mathbb{R}^3}C_{22}[f](p_1)dp_1,
\end{aligned}
\end{equation}
with the notice that 
$$\int_{\mathbb{R}^3}C_{22}[f](p_1)dp_1=0,$$
and
\begin{eqnarray*}
\int_{\mathbb{R}^3}C_{12}[f](p_1)dp_1&=&\lambda_1 n_c\iiint_{\mathbb{R}^{3}\times\mathbb{R}^{3}\times\mathbb{R}^{3}}K^{12}({p}_1,{p}_2,{p}_3)\delta({p}_1-{p}_2-{p}_3)\delta(\mathcal{E}_{{p}_1}-\mathcal{E}_{{p}_2}-\mathcal{E}_{{p}_3})\\\nonumber
& &\times[f(p_1)+2f(p_1)f(p_2)-f(p_2)f(p_3)]dp_1d{p}_2d{p}_3.
\end{eqnarray*}
From the above computations, we can see that the control of the total mass really comes from estimating the collision operator $C_{12}$, since the integral of $C_{22}$ is already $0$. Set 
$$J_1=\lambda_1 n_c\iiint_{\mathbb{R}^{3}\times\mathbb{R}^{3}\times\mathbb{R}^{3}}K^{12}({p}_1,{p}_2,{p}_3)\delta({p}_1-{p}_2-{p}_3)\delta(\mathcal{E}_{{p}_1}-\mathcal{E}_{{p}_2}-\mathcal{E}_{{p}_3})f(p_1)dp_1d{p}_2d{p}_3$$
and
$$J_2=2\lambda_1 n_c\iiint_{\mathbb{R}^{3}\times\mathbb{R}^{3}\times\mathbb{R}^{3}}K^{12}({p}_1,{p}_2,{p}_3)\delta({p}_1-{p}_2-{p}_3)\delta(\mathcal{E}_{{p}_1}-\mathcal{E}_{{p}_2}-\mathcal{E}_{{p}_3})f(p_1)f(p_2)dp_1d{p}_2d{p}_3,$$
to get
\begin{equation}\label{Propo:Mass:E2}
\frac{d}{dt}\int_{\mathbb{R}^3} f(p_1)dp_1 = \int_{\mathbb{R}^3} Q[f](p_1)dp_1\le  J_1+J_2,
\end{equation}
note that in the above inequality, we have dropped the negative term containing $f(p_2)f(p_3)$.\\
Now, $J_1$ can be estimated the following way, by using the definition of the Dirac functions $\delta({p}_1-{p}_2-{p}_3)$, $\delta(\mathcal{E}_{{p}_1}-\mathcal{E}_{{p}_2}-\mathcal{E}_{{p}_3})$ and the boundedness of $K^{12}({p}_1,{p}_2,{p}_1-p_2)$  
$$\begin{aligned}
J_1 ~~=&~~\lambda_1 n_c\iint_{\mathbb{R}^{3}\times\mathbb{R}^{3}}K^{12}({p}_1,{p}_2,{p}_1-p_2)\delta(\mathcal{E}_{{p}_1}-\mathcal{E}_{{p}_2}-\mathcal{E}_{{p}_1-p_2})f(p_1)dp_1d{p}_2
\\
~~\le&~~C\int_{\mathbb{R}^3}f(p_1)\left(\int_{S^0_{p_1}}\frac{1}{|\nabla H_0^{p_1}|}d\sigma(p_2)\right)dp_1,
\end{aligned}
$$
which, by Lemma \ref{lem-Sp}, can be bounded as
$$\begin{aligned}
J_1 ~~\le&~~C\int_{\mathbb{R}^3}f(p_1)|p_1|^2dp_1.
\end{aligned}
$$
Using the fact that $|p_1|^2$ is dominated by $\mathcal{E}_{p_1}$ up to a constant, yields
\begin{equation}\label{Propo:Mass:E3}\begin{aligned}
J_1 ~~\le&~~C\int_{\mathbb{R}^3}f(p_1)\mathcal{E}_{p_1}dp_1
~~\le ~~ C,
\end{aligned}
\end{equation}
where $C$ is a constant varying from line to line and the last inequality follows from the conservation of energy \eqref{Coro:ConservatioEnergy}.
\\ It remains to estimate $J_2$. By a straightforward use of the  definition of the Dirac functions $\delta({p}_1-{p}_2-{p}_3)$ and $\delta(\mathcal{E}_{{p}_1}-\mathcal{E}_{{p}_2}-\mathcal{E}_{{p}_3})$
$$\begin{aligned}
J_2 ~~=&~~2\lambda_1 n_c\iint_{\mathbb{R}^{3}\times\mathbb{R}^{3}}K^{12}({p}_1,{p}_2,{p}_1-p_2)\delta(\mathcal{E}_{{p}_1}-\mathcal{E}_{{p}_2}-\mathcal{E}_{{p}_1-p_2})f(p_1)f(p_2)dp_1d{p}_2
\\
~~=&~~2\lambda n_c\int_{\mathbb{R}^3}f(p_2)\left(\int_{S_{p_2}^2}K^{12}_2({p}_1,{p}_2,{p}_1-p_2)f(p_1)d\sigma(p_1)\right)dp_2,
\end{aligned}
$$
which, by Lemma \ref{lem-Sp2}, can be bounded as
$$\begin{aligned}
J_2 ~~\le&~~C\int_{\mathbb{R}^3}f(p_2)\left(\int_{S_{p_2}^2}K^{12}_2({p}_1,{p}_2,{p}_1-p_2)f(p_1)d\sigma(p_1)\right)dp_2\\ ~~\le&~~C\int_{\mathbb{R}^3}f(p_2)\left(\int_{\mathbb{R}^3}\frac{K^{12}({p}_1,{p}_2,{p}_1-p_2)}{|p_1||p_2|}f(p_1)dp_1\right)dp_2.
\end{aligned}
$$
Since $\frac{K^{12}({p}_1,{p}_2,{p}_1-p_2)}{|p_1||p_2|}$ is bounded by $|p_1-p_2|$, up to a constant, $J_2$ is dominated by
$$\begin{aligned}
J_2 ~~\le&~~C\int_{\mathbb{R}^3}f(p_2)\left(\int_{\mathbb{R}^3}f(p_1)|p_1-p_2|dp_1\right)dp_2\\
~~\le&~~C\int_{\mathbb{R}^3}f(p_2)\left(\int_{\mathbb{R}^3}f(p_1)(|p_1|+|p_2|)dp_1\right)dp_2\\
~~\le&~~C\left(\int_{\mathbb{R}^3}f(p_2)\mathcal{E}(p_2)dp_2\right)\left(\int_{\mathbb{R}^3}f(p_1)dp_1\right),
\end{aligned}
$$
notice that $C$ is a positive constant varying from line to line and we have just used the fact that $|p|$ is bounded by $\mathcal{E}(p)$ up to a constant, which by the conservation of energy \eqref{Coro:ConservatioEnergy}, implies
\begin{equation}\label{Propo:Mass:E4}\begin{aligned}
J_2 ~~
~~\le&~~C\left(\int_{\mathbb{R}^3}f(p_1)dp_1\right),
\end{aligned}
\end{equation}
Combining \eqref{Propo:Mass:E2}, \eqref{Propo:Mass:E3} and \eqref{Propo:Mass:E4} leads to
\begin{equation}\label{Propo:Mass:E5}
\frac{d}{dt}\int_{\mathbb{R}^3} f(p_1)dp_1  = \int_{\mathbb{R}^3} Q[f](p_1)dp_1 \leq C^*\left(1+\int_{\mathbb{R}^3} f(p_1)dp_1 \right),
\end{equation}
for some positive constant $C^*$, 
which implies the conclusion of the Proposition.
\end{proof}
\subsection{Finite time moment estimates of  the solution to the kinetic equation}\label{Sec:MomentEstimates}
\subsubsection{Estimating $C_{12}$}\label{Sec:C12}
\begin{proposition}\label{Propo:C12} For any positive, radial function $f(p)=f(|p|)$, for any $n\in\mathbb{N}$, there exists a universal positive constant $\mathcal{C}$ depending on $n$, such that the following bound on the collision operator $C_{12}$ holds true
\begin{equation}\label{Propo:C12:1}
\int_{\mathbb{R}^3}C_{12}[f](p_1)\mathcal{E}^n(p_1)dp_1\le \sum_{k=1}^{n-1}\mathcal{C} (m_k[f]+m_{k-1}[f])(m_{n-k-1}[f]+m_{n-k}[f])-\mathcal{C} m_{n+1}[f]+\mathcal{C} m_1[f].
\end{equation}
\end{proposition}
\begin{proof}
For the sake of simplicity, we denote $m_k[f]$ by $m_k$.
By a view of Lemma \ref{Lemma:WeakFormulation}, 
\begin{equation}\label{Propo:C12:E1}\begin{aligned}
~~&~~\int_{\mathbb{R}^3}C_{12}[f](p_1)\mathcal{E}^n(p_1)dp_1=\\
~~=&~~n_c\lambda_1\iiint_{\mathbb{R}^{3\times 3}}K^{12}(p_1,p_2,p_3)\delta(p_1-p_2-p_3)\delta(\mathcal{E}_{p_1}-\mathcal{E}_{p_2}-\mathcal{E}_{p_3})\times\\
~~&~~\times[f(p_2)f(p_3)-f(p_1)-2f(p_1)f(p_2)][\mathcal{E}^n_{p_1}-\mathcal{E}^n_{p_2}-\mathcal{E}^n_{p_3}]dp_1dp_2dp_3.
\end{aligned}
\end{equation}
By  the definition of $\delta(\mathcal{E}_{p_1}-\mathcal{E}_{p_2}-\mathcal{E}_{p_3})$, the term $\mathcal{E}^n_{p_1}-\mathcal{E}^n_{p_2}-\mathcal{E}^n_{p_3}$ could be rewritten as
\begin{eqnarray*}
(\mathcal{E}_{p_2}+\mathcal{E}_{p_3})^n-\mathcal{E}^n_{p_2}-\mathcal{E}^n_{p_3}=\sum_{k=1}^{n-1} {{n}\choose{k}}\mathcal{E}^k_{p_2}\mathcal{E}^{n-k}_{p_3},
\end{eqnarray*}
which yields
$$\begin{aligned}
~~&~~\int_{\mathbb{R}^3}C_{12}[f](p_1)\mathcal{E}^n(p_1)dp_1=\\
~~=&~~n_c\lambda_1\iiint_{\mathbb{R}^{3\times 3}}K^{12}(p_1,p_2,p_3)\delta(p_1-p_2-p_3)\delta(\mathcal{E}_{p_1}-\mathcal{E}_{p_2}-\mathcal{E}_{p_3})\times\\
~~&~~\times[f(p_2)f(p_3)-f(p_1)-2f(p_1)f(p_2)]\left[\sum_{k=1}^{n-1}{{n}\choose{k}}\mathcal{E}^k_{p_2}\mathcal{E}^{n-k}_{p_3}\right]dp_1dp_2dp_3.
\end{aligned}
$$
Dropping the term containing $-2f(p_1)f(p_2)$, the above quantity could be bounded as
\begin{equation}\label{Propo:C12:E2}\begin{aligned}
~~&~~\int_{\mathbb{R}^3}C_{12}[f](p_1)\mathcal{E}^n(p_1)dp_1\le L_1+L_2,
\end{aligned}
\end{equation}
where
$$\begin{aligned}
L_1:=
~~&~~n_c\lambda_1\iiint_{\mathbb{R}^{3\times 3}}K^{12}(p_1,p_2,p_3)\delta(p_1-p_2-p_3)\delta(\mathcal{E}_{p_1}-\mathcal{E}_{p_2}-\mathcal{E}_{p_3})\times\\
~~&~~\times f(p_2)f(p_3)\left[\sum_{k=1}^{n-1}{{n}\choose{k}}\mathcal{E}^k_{p_2}\mathcal{E}^{n-k}_{p_3}\right]dp_1dp_2dp_3\\
L_2:=~~&~~-n_c\lambda_1\iiint_{\mathbb{R}^{3\times 3}}K^{12}(p_1,p_2,p_3)\delta(p_1-p_2-p_3)\delta(\mathcal{E}_{p_1}-\mathcal{E}_{p_2}-\mathcal{E}_{p_3})\times\\
~~&~~\times f(p_1)\left[\sum_{k=1}^{n-1}{{n}\choose{k}}\mathcal{E}^k_{p_2}\mathcal{E}^{n-k}_{p_3}\right]dp_1dp_2dp_3.
\end{aligned}
$$
Let us first look at $L_1$. By the definition of $\delta(p_1-p_2-p_3)$,
$$\begin{aligned}
L_1=
~~&~~n_c\lambda_1\iint_{\mathbb{R}^{3\times 2}}K^{12}(p_2+p_3,p_2,p_3)\delta(\mathcal{E}_{p_2+p_3}-\mathcal{E}_{p_2}-\mathcal{E}_{p_3})\times\\
~~&~~\times f(p_2)f(p_3)\left[\sum_{k=1}^{n-1}{{n}\choose{k}}\mathcal{E}^k_{p_2}\mathcal{E}^{n-k}_{p_3}\right]dp_2dp_3,
\end{aligned}
$$
which by the boundedness of $K^{12}$, could be bounded as
$$\begin{aligned}
L_1\le
~~&~~C\iint_{\mathbb{R}^{3}\times \mathbb{R}^{3}}\delta(\mathcal{E}_{p_2+p_3}-\mathcal{E}_{p_2}-\mathcal{E}_{p_3})\times\\
~~&~~\times f(p_2)f(p_3)\left[\sum_{k=1}^{n-1}{{n}\choose{k}}\mathcal{E}^k_{p_2}\mathcal{E}^{n-k}_{p_3}\right]dp_2dp_3\\
\le
~~&~~\sum_{k=1}^{n-1}C\int_{\mathbb{R}^{3}}
 f(p_2)\mathcal{E}^k_{p_2}\left[\int_{S_{p_2}^1}f(p_3)\frac{\mathcal{E}^{n-k}_{p_3}}{|\nabla H_2^{p_2}|}d\sigma(p_3)\right]dp_2.
\end{aligned}
$$
Applying Lemma \ref{lem-Sp1} to the above inequality leads to
$$\begin{aligned}
L_1\le
~~&~~\sum_{k=1}^{n-1}C\int_{\mathbb{R}^{3}}
 f(p_2)\frac{\mathcal{E}^k_{p_2}}{|p_2|}\left[\int_{\mathbb{R}^{3}}f(p_3)\frac{\mathcal{E}^{n-k}_{p_3}}{|p_3|}dp_3\right]dp_2,
\end{aligned}
$$
where $C$ is some constant varying from line to line.\\
Observe that $$\frac{\mathcal{E}^{n-k}_{p_3}}{|p_3|}\le C\left({\mathcal{E}^{n-k-1}_{p_3}}+{\mathcal{E}^{n-k}_{p_3}}\right), \frac{\mathcal{E}^k_{p_2}}{|p_2|}\le C\left({\mathcal{E}^{k-1}_{p_2}}+{\mathcal{E}^{k}_{p_2}}\right) $$
which implies
\begin{equation}\label{Propo:C12:E3}\begin{aligned}
L_1\le
~~&~~\sum_{k=1}^{n-1}C\left[\int_{\mathbb{R}^{3}}
 f(p_1)\mathcal{E}^k_{p_1}dp_1\right]\left[\int_{\mathbb{R}^{3}}
 f(p_1)\mathcal{E}^{n-k-1}_{p_1}dp_1+\int_{\mathbb{R}^{3}}
 f(p_1)\mathcal{E}^{n-k}_{p_1}dp_1\right]\\
 ~~\le&~~\sum_{k=1}^{n-1}C [m_k+m_{k-1}][m_{n-k-1}+m_{n-k}].
\end{aligned}
\end{equation}
Now, by the definition of $\delta(p_1-p_2-p_3)$ and $\delta(\mathcal{E}_{p_1}-\mathcal{E}_{p_2}-\mathcal{E}_{p_1-p_2})$, the second term $L_2$ can be rewritten as
$$\begin{aligned}
L_2=~~&~~-n_c\lambda_1\iint_{\mathbb{R}^{3\times 3}}K^{12}(p_1,p_2,p_1-p_2)\delta(\mathcal{E}_{p_1}-\mathcal{E}_{p_2}-\mathcal{E}_{p_1-p_2})\times\\
~~&~~\times f(p_1)\left[\sum_{k=1}^{n-1}{{n}\choose{k}}\mathcal{E}^k_{p_2}\mathcal{E}^{n-k}_{p_1-p_2}\right]dp_1dp_2\\
\le~~&~~-\sum_{k=1}^{n-1}C\int_{\mathbb{R}^{3}}f(p_1)\left[\int_{S_{p_1}^0}K_0^{12}(p_1,p_2,p_1-p_2)\mathcal{E}^k_{p_2}\mathcal{E}^{n-k}_{p_1-p_2}d\sigma(p_2)\right]dp_1
.
\end{aligned}
$$
Since $$\mathcal{E}^k_{p_2}\mathcal{E}^{n-k}_{p_1-p_2}\geq C\left[|p_2|^k|p_1-p_2|^{n-k}+|p_2|^{2k}|p_1-p_2|^{2(n-k)}\right],$$
where $C$ is some positive constant varying from line to line, $L_2$ can be estimated as follows
$$\begin{aligned}
~~&~~L_2\le\\
\le&-\sum_{k=1}^{n-1}C\int_{\mathbb{R}^{3}}f(p_1)\left[\int_{S_{p_1}^0}K_0^{12}(p_1,p_2,p_1-p_2)\left(|p_2|^k|p_1-p_2|^{n-k}+|p_2|^{2k}|p_1-p_2|^{2(n-k)}\right)d\sigma(p_2)\right]dp_1,
\end{aligned}
$$  
which, due to Lemma \ref{lem-Sp}, can be bounded by
$$\begin{aligned}
L_2~~\le &~~-C\int_{\mathbb{R}^{3}}f(p_1)\left((|p_1|\wedge 1)^{n+7}|p_1|^{n+1}+(|p_1|\wedge 1)^{2n+7}|p_1|^{2n+1}\right)dp_1.
\end{aligned}
$$  
Splitting the integral on $\mathbb{R}^3$ into two integrals on $|p_1|>1$ and $|p_1|\le 1$ yields
$$\begin{aligned}
L_2~~\le &~~-C\int_{|p_1|>1}f(p_1)\left(|p_1|^{n+2}+|p_1|^{2n+1}\right)dp_1\\
~~&~~ -C\int_{|p_1|\le 1}f(p_1)\left(|p_1|^{2n+7}+|p_1|^{4n+7}\right)dp_1\\
~~\le &~~-C\int_{|p_1|>1}f(p_1)\left(|p_1|^{n+1}+|p_1|^{2n+2}\right)dp_1
,
\end{aligned}
$$  
where $C$ is some positive constant varying from line to line and we have used the inequality $-|p_1|^{n+1}>-|p_1|^{n+2}$ for $|p_1|>1$. 
Adding and subtracting the right hand side of the above inequality with an integral on the domain $|p_1|\leq1$, we obtain
$$\begin{aligned}
L_2
~~\le &~~-C\left[\int_{\mathbb{R}^3}f(p_1)\left(|p_1|^{n+1}+|p_1|^{2n+2}\right)dp_1-\int_{|p_1|\leq1}f(p_1)\left(|p_1|^{n+1}+|p_1|^{2n+2}\right)dp_1
\right]\\
~~\le &~~-C\left[\int_{\mathbb{R}^3}f(p_1)\left(|p_1|^{n+1}+|p_1|^{2n+2}\right)dp_1-\int_{|p_1|\leq1}|p_1|f(p_1)dp_1
\right],
\end{aligned}
$$  
where the last inequality is due to the fact that we are integrating on $|p_1|\le 1$. Bounding the integral on $|p_1|\leq 1$ by the integral on the full space $\mathbb{R}^3$, we get
$$\begin{aligned}
L_2
~~\le &~~-C\int_{\mathbb{R}^3}f(p_1)\left(|p_1|^{n+1}+|p_1|^{2n+2}\right)dp_1+C\int_{\mathbb{R}^3}|p_1|f(p_1)dp_1
.
\end{aligned}
$$
By the inequality
$$|p_1|^{n+1}+|p_1|^{2n+2}\geq C\mathcal{E}^{n+1}_{p_1},$$
we obtain the following estimate on $L_2$
  \begin{equation}\label{Propo:C12:E4}\begin{aligned}
L_2\le
~~&~~-C m_{n+1}+C m_1.
\end{aligned}
\end{equation}
Combining \eqref{Propo:C12:E2}, \eqref{Propo:C12:E3} and \eqref{Propo:C12:E4}, we get the conclusion of the Proposition. 
\end{proof}
\subsubsection{Estimating $C_{22}$}\label{Sec:C22}
\begin{proposition}\label{Propo:C22} For any positive, radial function $f(p)=f(|p|)$, for any $n\in\mathbb{N}$, $n>2$, $n$ is odd, there exists a universal positive constant $\mathcal{C}$ depending on $n$, such that the following bound on the collision operator $C_{22}$ holds true
\begin{equation}\label{Propo:C22:1}
\begin{aligned}
&~~\int_{\mathbb{R}^3}C_{22}[f](p_1)\mathcal{E}_{p_1}^{n}dp_1\le\\
\le &~~\mathcal{C}\sum_{0\le i, j, k<n;~ i+j+k=n}\sum_{s=0}^{k+1}m_{i+s}\left(m_{j+k-s}+m_{j+k-s+1/2}\right)+\\
&~~+\mathcal{C}\sum_{0\le i, j, k<n;~ i+j+k=n:~j,k>0}m_{i}\left(m_{j-1}+m_{j-1/2}\right)\left(m_{k-1}+m_{k-1/2}\right).
\end{aligned}
\end{equation}
\end{proposition}
\begin{proof} For the sake of simplicity, we denote $m_k[f]$ by $m_k$. We first observe that, by a spherical change of variables
$$\int_{\mathbb{R}^3}C_{22}[f](p_1)\mathcal{E}_{p_1}^{n}dp_1=C\int_{\mathbb{R}_+}C_{22}[f](p_1)|p_1|^2\mathcal{E}_{p_1}^{n}d|p_1|,$$
where $C$ is some universal constant varying from line to line, 
and
$$
\begin{aligned}
&~~\int_{\mathbb{R}^3}C_{22}[f](p_1)\mathcal{E}_{p_1}^{n}dp_1=\\
=&~~\kappa_3\int_{\mathbb{R}_+^{4}}K^{22}(p_1,p_2,p_3,p_4){\min\{|p_1|,|p_2|,|p_3|,|p_4|\}}|p_1||p_2||p_3||p_4|\delta(\mathcal{E}_{p_1}+\mathcal{E}_{p_2}-\mathcal{E}_{p_3}-\mathcal{E}_{p_4})\\
&~~\times [f(p_3)f(p_4)(1+f(p_1)+f(p_2))-f(p_1)f(p_2)(1+f(p_3)+f(p_4))]\mathcal{E}_{p_1}^{n}d|{p}_1|d|{p}_2|d|{p}_3|d|{p}_4|.
\end{aligned}
$$
By the classical change of variables $(p_1,p_2)\leftrightarrow (p_2,p_1)$,  $(p_1,p_2)\leftrightarrow (p_3,p_4)$ (cf. \cite{Villani:2002:RMT}), the above equation could be expressed in the following way
$$
\begin{aligned}
&~~\int_{\mathbb{R}^3}C_{22}[f](p_1)\mathcal{E}_{p_1}^{n}dp_1=\\
=&~~C\int_{\mathbb{R}_+^{4}}K^{22}(p_1,p_2,p_3,p_4){\min\{|p_1|,|p_2|,|p_3|,|p_4|\}}|p_1||p_2||p_3||p_4|\delta(\mathcal{E}_{p_1}+\mathcal{E}_{p_2}-\mathcal{E}_{p_3}-\mathcal{E}_{p_4})\\
&~~\times f(p_1)f(p_2)(1+f(p_3)+f(p_4))\Big[\mathcal{E}_{p_4}^{n}+\mathcal{E}_{p_3}^{n}-\mathcal{E}_{p_2}^{n}-\mathcal{E}_{p_1}^{n}\Big]d|{p}_1|d|{p}_2|d|{p}_3|d|{p}_4|,
\end{aligned}
$$
where $C$ is some universal constant varying from line to line.\\
Taking into account the fact that $p_3$ and $p_4$ are symmetric, and using the definition of the Dirac function to get $\mathcal{E}_{p_4}=\mathcal{E}_{p_1}+\mathcal{E}_{p_2}-\mathcal{E}_{p_3}$, one obtains
\begin{equation}\label{Propo:C22:E1}
\begin{aligned}
&~~\int_{\mathbb{R}^3}C_{22}[f](p_1)\mathcal{E}_{p_1}^{n}dp_1=\\
=&~~C\int_{\mathbb{R}_+^{4}}K^{22}(p_1,p_2,p_3,p_4){\min\{|p_1|,|p_2|,|p_3|,|p_4|\}}|p_1||p_2||p_3||p_4|\delta(\mathcal{E}_{p_1}+\mathcal{E}_{p_2}-\mathcal{E}_{p_3}-\mathcal{E}_{p_4}) \\
&~~\times f(p_1)f(p_2)(1+2f(p_3))\Big[(\mathcal{E}_{p_1}+\mathcal{E}_{p_2}-\mathcal{E}_{p_3})^{n}+\mathcal{E}_{p_3}^{n}-\mathcal{E}_{p_2}^{n}-\mathcal{E}_{p_1}^{n}\Big]d|{p}_1|d|{p}_2|d|{p}_3|d|{p}_4|.
\end{aligned}
\end{equation}
Notice that for $$\mathcal{E}(|p|)=\sqrt{\kappa_1|p|^2+\kappa_2|p|^4},$$
its derivative is bounded from below as
\begin{equation}\label{Propo:C22:E1b}
\mathcal{E}'(|p|)=\frac{\kappa_1+2\kappa_2|p|^2}{\sqrt{\kappa_1+\kappa_2|p|^2}}\geq C|p|,\end{equation}
where $C$ is some universal constant varying from line to line, which means $C|p_4|d|p_4|$ can be bounded by $d\mathcal{E}_{p_4}$. As a consequence, the following estimate on the right hand side of \eqref{Propo:C22:E1} follows
\begin{equation}\label{Propo:C22:E1a}
\begin{aligned}
&~~\int_{\mathbb{R}^3}C_{22}[f](p_1)\mathcal{E}_{p_1}^{n}dp_1=\\
\le&~~C\int_{\mathbb{R}_+^{4}}K^{22}(p_1,p_2,p_3,p_4){\min\{|p_1|,|p_2|,|p_3|\}}|p_1||p_2||p_3|\delta(\mathcal{E}_{p_1}+\mathcal{E}_{p_2}-\mathcal{E}_{p_3}-\mathcal{E}_{p_4}) \\
&~~\times f(p_1)f(p_2)(1+2f(p_3))\Big[(\mathcal{E}_{p_1}+\mathcal{E}_{p_2}-\mathcal{E}_{p_3})^{n}+\mathcal{E}_{p_3}^{n}-\mathcal{E}_{p_2}^{n}-\mathcal{E}_{p_1}^{n}\Big]d|{p}_1|d|{p}_2|d|{p}_3|d\mathcal{E}_{p_4},
\end{aligned}
\end{equation}
where, we have used the fact that $$\min\{|p_1|,|p_2|,|p_3|,|p_4|\}\le \min\{|p_1|,|p_2|,|p_3|\}.$$
Since $n$ is an odd number, applying Newton formula to the term $(\mathcal{E}_{p_1}+\mathcal{E}_{p_2}-\mathcal{E}_{p_3})^{n}+\mathcal{E}_{p_3}^{n}-\mathcal{E}_{p_2}^{n}-\mathcal{E}_{p_1}^{n}$ yields
\begin{equation}\label{Propo:C22:E2}\begin{aligned}
(\mathcal{E}_{p_1}+\mathcal{E}_{p_2}-\mathcal{E}_{p_3})^{n}+\mathcal{E}_{p_3}^{n}-\mathcal{E}_{p_2}^{n}-\mathcal{E}_{p_1}^{n}
=&~~\sum_{0\le i, j, k<n;~ i+j+k=n} C_{i,j,k,n}\mathcal{E}_{p_1}^i\mathcal{E}_{p_2}^j\mathcal{E}_{p_3}^k.
\end{aligned}
\end{equation}
 Plugging \eqref{Propo:C22:E2} into \eqref{Propo:C22:E1a}, integrating with respect to $d\mathcal{E}_4$ and using the bound \eqref{K22Bound} leads to
\begin{equation}\label{Propo:C22:E3}
\begin{aligned}
&~~\int_{\mathbb{R}^3}C_{22}[f](p_1)\mathcal{E}_{p_1}^{n}dp_1\le\\
\le &~~C\int_{\{\mathcal{E}_{p_1}+\mathcal{E}_{p_2}\ge\mathcal{E}_{p_3}\}}{\min\{|p_1|,|p_2|,|p_3|\}}|p_1||p_2||p_3|f(p_1)f(p_2)(1+2f(p_3))\\
&~~\times\left[\sum_{0\le i, j, k<n;~ i+j+k=n} |C_{i,j,k,n}|\mathcal{E}_{p_1}^i\mathcal{E}_{p_2}^j\mathcal{E}_{p_3}^k\right]d|p_1|d|p_2|d|p_3|\\
\le&~~C\sum_{0\le i, j, k<n;~ i+j+k=n}\int_{\{\mathcal{E}_{p_1}+\mathcal{E}_{p_2}\ge\mathcal{E}_{p_3}\}}{\min\{|p_1|,|p_2|,|p_3|\}}|p_1||p_2||p_3|\\
&~~\times f(p_1)f(p_2)(1+2f(p_3))\mathcal{E}_{p_1}^i\mathcal{E}_{p_2}^j\mathcal{E}_{p_3}^kd|p_1|d|p_2|d|p_3|.
\end{aligned}
\end{equation}
In order to estimate the right hand side of \eqref{Propo:C22:E3}, we estimate each term containing $f(p_1)f(p_2)$ and 2$f(p_1)f(p_2)f(p_3)$ seperately. 
\\ Let us first look at the term containing $f(p_1)f(p_2)$
\begin{equation}\label{Propo:C22:E4}
\begin{aligned}
L_1:=&~~C\sum_{0\le i, j, k<n;~ i+j+k=n}\int_{\{\mathcal{E}_{p_1}+\mathcal{E}_{p_2}\ge\mathcal{E}_{p_3}\}}{\min\{|p_1|,|p_2|,|p_3|\}}|p_1||p_2||p_3|\\
&~~\times f(p_1)f(p_2)\mathcal{E}_{p_1}^i\mathcal{E}_{p_2}^j\mathcal{E}_{p_3}^kd|p_1|d|p_2|d|p_3|\\
\le&~~C\sum_{0\le i, j, k<n;~ i+j+k=n}\int_{\{\mathcal{E}_{p_1}+\mathcal{E}_{p_2}\ge\mathcal{E}_{p_3}\}}{\min\{|p_1|,|p_2|\}}|p_1||p_2|\\
&~~\times f(p_1)f(p_2)\mathcal{E}_{p_1}^i\mathcal{E}_{p_2}^j\mathcal{E}_{p_3}^kd|p_1|d|p_2|d\mathcal{E}_{p_3},\end{aligned}
\end{equation}
where we have used \eqref{Propo:C22:E1b} to get $|p_3|dp_3\le Cd\mathcal{E}_{p_3}$ and the fact that
$$\min\{|p_1|,|p_2|,|p_3|\}\le \min\{|p_1|,|p_2|\}.$$
In \eqref{Propo:C22:E4}, integrating with respect to $d\mathcal{E}_{p_3}$ leads to
\begin{equation}\label{Propo:C22:E5}
\begin{aligned}
L_1\le&~~C\sum_{0\le i, j, k<n;~ i+j+k=n}\int_{\mathbb{R}_+^2}{\min\{|p_1|,|p_2|\}}|p_1||p_2|f(p_1)f(p_2)\\
&~~\times \mathcal{E}_{p_1}^i\mathcal{E}_{p_2}^j\frac{(\mathcal{E}_{p_1}+\mathcal{E}_{p_2})^{k+1}}{k+1}d|p_1|d|p_2|,\end{aligned}
\end{equation}
where $C$ is some universal constant varying from line to line.\\
Again, by Newton formula  
\begin{equation}\label{Newton}
(\mathcal{E}_{p_1}+\mathcal{E}_{p_2})^{k+1}=\sum_{0}^{k+1} {{k+1}\choose{s}}\mathcal{E}_{p_1}^s\mathcal{E}_{p_2}^{k+1-s},
\end{equation}
which, together with \eqref{Propo:C22:E4} leads to
\begin{equation}\label{Propo:C22:E6}
\begin{aligned}
&~~L_1\le\\
\le&~~C\sum_{0\le i, j, k<n;~ i+j+k=n}\sum_{s=0}^{k+1}\int_{\mathbb{R}_+^2}{\min\{|p_1|,|p_2|\}}|p_1||p_2|f(p_1)f(p_2)\mathcal{E}_{p_1}^{i+s}\mathcal{E}_{p_2}^{j+k+1-s}d|p_1|d|p_2|\\
\le&~~C\sum_{0\le i, j, k<n;~ i+j+k=n}\sum_{s=0}^{k+1}\int_{\mathbb{R}_+^2}|p_1|^2|p_2|f(p_1)f(p_2)\mathcal{E}_{p_1}^{i+s}\mathcal{E}_{p_2}^{j+k+1-s}d|p_1|d|p_2|.
\end{aligned}
\end{equation}
Note that integrals of $d|p_1|$ and $d|p_2|$ in \eqref{Propo:C22:E6} are separated and
it is straightforward that the integral of $d|p_1|$ can be computed, by a spherical coordinate change of variables, as
\begin{equation}\label{Propo:C22:E7}
\begin{aligned}
\int_{\mathbb{R}_+}|p_1|^2f(p_1)\mathcal{E}_{p_1}^{i+s}d|p_1|=\int_{\mathbb{R}^3}f(p_1)\mathcal{E}_{p_1}^{i+s}dp_1=m_{i+s}.
\end{aligned}
\end{equation}
Now, for the second integral concerning $d|p_2|$, by the inequality
$$ \mathcal{E}_{p_2}\leq C(|p_2|+|p_2|^2),$$
for some positive constant $C$,
one gets
$$
\begin{aligned}
\int_{\mathbb{R}_+}|p_2|f(p_2)\mathcal{E}_{p_2}^{j+k+1-s}d|p_2|
~~\le&~~~~C\int_{\mathbb{R}_+}(|p_2|^2+|p_2|^3)f(p_2)\mathcal{E}_{p_2}^{j+k-s}d|p_2|\\
~~\le&~~~~C\int_{\mathbb{R}^3}(1+|p_2|)f(p_2)\mathcal{E}_{p_2}^{j+k-s}dp_2,
\end{aligned}
$$
which, by the inequality
$$ \mathcal{E}_{p_2}^{1/2}\geq C|p_2|,$$
implies that 
\begin{equation}\label{Propo:C22:E8}
\begin{aligned}
\int_{\mathbb{R}_+}|p_2|f(p_2)\mathcal{E}_{p_2}^{j+k+1-s}d|p_2|
~~\le&~~~~C\int_{\mathbb{R}^3}\left(1+\mathcal{E}_{p_2}^{1/2}\right)f(p_2)\mathcal{E}_{p_2}^{j+k-s}dp_2\\
~~\le&~~~~C\left(m_{j+k-s}+m_{j+k-s+1/2}\right).
\end{aligned}
\end{equation}
Combining \eqref{Propo:C22:E6}, \eqref{Propo:C22:E7} and \eqref{Propo:C22:E8} lead to
\begin{equation}\label{Propo:C22:E9}
\begin{aligned}
&~~L_1
\le&~~C\sum_{0\le i, j, k<n;~ i+j+k=n}\sum_{s=0}^{k+1}m_{i+s}\left(m_{j+k-s}+m_{j+k-s+1/2}\right).
\end{aligned}
\end{equation}
Now, for the term containing $2f(p_1)f(p_2)f(p_3)$, by bounding the integral on $\{\mathcal{E}_{p_1}+\mathcal{E}_{p_2}\ge \mathcal{E}_{p_3}\}$ by the integral on $\mathbb{R}^3_+$, we get 
\begin{equation}\label{Propo:C22:E10}
\begin{aligned}
L_2:=&~~C\sum_{0\le i, j, k<n;~ i+j+k=n}\int_{\{\mathcal{E}_{p_1}+\mathcal{E}_{p_2}\ge\mathcal{E}_{p_3}\}}{\min\{|p_1|,|p_2|,|p_3|\}}|p_1||p_2||p_3|\\
&~~\times 2f(p_1)f(p_2)f(p_3)\mathcal{E}_{p_1}^i\mathcal{E}_{p_2}^j\mathcal{E}_{p_3}^kd|p_1|d|p_2|d|p_3|\\
\le&~~C\sum_{0\le i, j, k<n;~ i+j+k=n}\int_{\mathbb{R}_+^3}{\min\{|p_1|,|p_2|,|p_3|\}}|p_1||p_2||p_3|\\
&~~\times f(p_1)f(p_2)f(p_3)\mathcal{E}_{p_1}^i\mathcal{E}_{p_2}^j\mathcal{E}_{p_3}^kd|p_1|d|p_2|d|p_3|,\end{aligned}
\end{equation}
where $C$ is some universal constant varying from line to line.\\
Notice that there are only two cases: $i,j,k>0$ and one of $i,j,k$ is $0$. Indeed, due to the condition that $i+j+k=n$ and $0\leq i,j,k<n$, the case where two of the index $i,j,k$ 
are $0$ will not happen. Therefore, we can suppose without loss of generality that $i\ge 0$ and $j,k>0$.\\
The terms on the right hand side of \eqref{Propo:C22:E10} can be estimated as
 \begin{equation}\label{Propo:C22:E11}
\begin{aligned}
&~~\int_{\mathbb{R}_+^3}{\min\{|p_1|,|p_2|,|p_3|\}}|p_1||p_2||p_3| f(p_1)f(p_2)f(p_3)\mathcal{E}_{p_1}^i\mathcal{E}_{p_2}^j\mathcal{E}_{p_3}^kd|p_1|d|p_2|d|p_3|\\
\le&\int_{\mathbb{R}_+}|p_1|^2\mathcal{E}_{p_1}^if(p_1)d|p_1|\int_{\mathbb{R}_+}|p_2|\mathcal{E}_{p_2}^jf(p_2)d|p_2|\int_{\mathbb{R}_+}|p_3|\mathcal{E}_{p_3}^kf(p_3)d|p_3|~~.\end{aligned}
\end{equation}
For each term on the right hand side of \eqref{Propo:C22:E11}, one can write, by the spherical coordinate change of variables
 \begin{equation}\label{Propo:C22:E12}
\int_{\mathbb{R}_+}|p_1|^2\mathcal{E}_{p_1}^if(p_1)d|p_1|=\int_{\mathbb{R}^3}\mathcal{E}_{p_1}^if(p_1)dp_1=m_{i},
\end{equation}
\begin{equation}\label{Propo:C22:E13}
\int_{\mathbb{R}_+}|p_2|\mathcal{E}_{p_2}^jf(p_2)d|p_2|\le C\left(m_{j-1}+m_{j-1/2}\right),
\end{equation}
 \begin{equation}\label{Propo:C22:E14}
\int_{\mathbb{R}_+}|p_3|\mathcal{E}_{p_3}^kf(p_3)d|p_3|\le C\left(m_{k-1}+m_{k-1/2}\right),
\end{equation}
where \eqref{Propo:C22:E13} and \eqref{Propo:C22:E14} are obtained by exactly the same manner as  \eqref{Propo:C22:E8}.
\\ Combining \eqref{Propo:C22:E11}, \eqref{Propo:C22:E12}, \eqref{Propo:C22:E13} and \eqref{Propo:C22:E14} yields
\begin{equation}\label{Propo:C22:E15}
\begin{aligned}
&~~\int_{\mathbb{R}_+^3}{\min\{|p_1|,|p_2|,|p_3|\}}|p_1||p_2||p_3| f(p_1)f(p_2)f(p_3)\mathcal{E}_{p_1}^i\mathcal{E}_{p_2}^j\mathcal{E}_{p_3}^kd|p_1|d|p_2|d|p_3|\\
~~\le&~~Cm_{i}\left(m_{j-1}+m_{j-1/2}\right)\left(m_{k-1}+m_{k-1/2}\right)~~.\end{aligned}
\end{equation}
The two inequalities \eqref{Propo:C22:E10} and \eqref{Propo:C22:E15} yield
\begin{equation}\label{Propo:C22:E16}
\begin{aligned}
L_2
\le&~~C\sum_{0\le i, j, k<n;~ i+j+k=n:~j,k>0}m_{i}\left(m_{j-1}+m_{j-1/2}\right)\left(m_{k-1}+m_{k-1/2}\right),\end{aligned}
\end{equation}
where $C$ is some universal constant varying from line to line.\\
From \eqref{Propo:C22:E3}, \eqref{Propo:C22:E9} and \eqref{Propo:C22:E16}, we get
$$\begin{aligned}
&~~\int_{\mathbb{R}^3}C_{22}[f](p_1)\mathcal{E}_{p_1}^{n}dp_1\le\\
\le &~~C\sum_{0\le i, j, k<n;~ i+j+k=n}\sum_{s=0}^{k+1}m_{i+s}\left(m_{j+k-s}+m_{j+k-s+1/2}\right)+\\
&~~+C\sum_{0\le i, j, k<n;~ i+j+k=n;~j,k>0}m_{i}\left(m_{j-1}+m_{j-1/2}\right)\left(m_{k-1}+m_{k-1/2}\right).
\end{aligned}
$$\end{proof}
\subsubsection{Finite time moment estimates}\label{Sec:MomentsEstimate}
\begin{proposition}\label{Propo:MomentsPropa}
Suppose that $f_0(p)=f_0(|p|)$ is a positive radial initial condition and 
$$\int_{\mathbb{R}^3}f_0(p)\mathcal{E}_pdp<\infty,~~~\int_{\mathbb{R}^3}f_0(p)dp<\infty,$$
then for any finite time interval $[0,T]$, and for any $n\geq 1$, the positive radial solution $f(t,p)=f(t,|p|)$ of \eqref{QB} satisfies
$$\sup_{t\in[\tau,T]}\int_{\mathbb{R}^3}f(t,p)\mathcal{E}_p^ndp<C_\tau,\ \ \ \forall  \ 0<\tau\le T,$$
where $C_\tau$ is a constant depending on $\tau$.\\
If $$\int_{\mathbb{R}^3}f_0(p)\mathcal{E}_p^ndp<\infty, $$
then
$$\sup_{t\in[0,T]}\int_{\mathbb{R}^3}f(t,p)\mathcal{E}_p^ndp<\infty.$$
\end{proposition}
In order to prove Proposition \ref{Propo:MomentsPropa}, we would need the following Holder inequality.
\begin{lemma}\label{Lemma:Holder}
Let $f$ be a function in $L^1(\mathbb{R}^3)\cap L^1_n(\mathbb{R}^3)$, then
$$\|f\|_{L^1_k}\le \mathcal{C}\|f\|_{L^1_n}^{\frac{k}{n}},$$
where $\mathcal{C}$ is a constant depending on $\|f\|_{L^1}$, $k$ and $n$.
\end{lemma}
\begin{proof}
By Holder inequality, we have
\begin{eqnarray*}
\int_{\mathbb{R}^3}|p|^kf(p)dp &\le& \left(\int_{\mathbb{R}^3}|f(p)|dp\right)^{\frac{n-k}{n}}\left(\int_{\mathbb{R}^3}|p|^n|f(p)|dp\right)^{\frac{k}{n}}\\
&\le& C\left(\|f\|_{L^1},k,n\right)\left(\int_{\mathbb{R}^3}|p|^nf(p)dp\right)^{\frac{k}{n}}.
\end{eqnarray*}
\end{proof}
\begin{proof}[of Proposition \ref{Propo:MomentsPropa}]
Fix a time interval $[0,T]$. It is sufficient to prove Proposition \ref{Propo:MomentsPropa} for $n\in\mathbb{N}$, $n$ odd. Using $\mathcal{E}_{p_1}^n$ as a test function in \eqref{QB}, as a view of Lemma \ref{Lemma:WeakFormulation}, we get
\begin{equation}\label{Propo:MomentsPropa:E1}
\frac{d}{dt}\int_{\mathbb{R}^3}f(p_1)\mathcal{E}_{p_1}^ndp_1=\int_{\mathbb{R}^3}C_{12}[f](p_1)\mathcal{E}_{p_1}^ndp_1+\int_{\mathbb{R}^3}C_{22}[f](p_1)\mathcal{E}_{p_1}^ndp_1.
\end{equation}
 For the sake of simplicity, we denote $m_k[f(t)]$ as $m_k(t)$. First, let us consider the $C_{12}$ collision operator. By Proposition \ref{Propo:C12} 
$$\int_{\mathbb{R}^3}C_{12}[f](p_1)\mathcal{E}^n(p_1)dp_1\le \sum_{k=1}^{n-1}{C} (m_k(t)+m_{k-1}(t))(m_{n-k-1}(t)+m_{n-k}(t))-{C} m_{n+1}(t)+{C} m_1(t).
$$Since, according to Proposition \ref{Propo:Mass}, $m_0(t)$ is bounded by a constant $C$ on $[0,T]$, we deduce from Lemma  \ref{Lemma:Holder} that
\begin{eqnarray*}
m_k(t)&\leq& C m_n(t)^{\frac{k}{n}},~~~~~m_{k-1}(t)\leq C m_n(t)^{\frac{k-1}{n}},~~~~~m_{n-k-1}(t)\leq C m_n(t)^{\frac{n-k-1}{n}},\\
~~~~~m_{n-k}(t)&\leq &C m_n(t)^{\frac{n-k}{n}}, ~~~~~C m_{n+1}(t)\geq m_n(t)^{\frac{n+1}{n}}, ~~~~~C m_1(t)\leq m_n(t)^{\frac{1}{n}},
\end{eqnarray*}
where $C$ depends on  $n$, $k$, and the bound of the mass on $[0,T]$ in Proposition \ref{Propo:Mass}. As a consequence, we obtain the following estimate for $C_{12}$
\begin{equation}\label{Propo:MomentsPropa:E2}
\int_{\mathbb{R}^3}C_{12}[f](p_1)\mathcal{E}^n(p_1)dp_1\le {C} m_n(t)+Cm_n(t)^{\frac{n-1}{n}}+Cm_n(t)^{\frac{n-2}{n}}+Cm_n(t)^{\frac{1}{n}}-Cm_{n}(t)^{\frac{n+1}{n}}.\end{equation}
Now, for the $C_{22}$ collision operator, according to Proposition \ref{Propo:C22},
$$\begin{aligned}
&~~\int_{\mathbb{R}^3}C_{22}[f](p_1)\mathcal{E}_{p_1}^{n}dp_1\le\\
\le &~~C\sum_{0\le i, j, k<n;~ i+j+k=n}\sum_{s=0}^{k+1}\left(m_{i+s}(t)+m_{j+k-s}(t)+m_{j+k-s+1/2}(t)\right)+\\
&~~+C\sum_{0\le i, j, k<n;~ i+j+k=n:~j,k>0}m_{i}(t)\left(m_{j-1}(t)+m_{j-1/2}(t)\right)\left(m_{k-1}(t)+m_{k-1/2}(t)\right).
\end{aligned}
$$
Again, by  Proposition \ref{Propo:Mass}, and Lemma  \ref{Lemma:Holder} 
\begin{eqnarray*}
m_{i+s}(t)&\leq& C m_n(t)^{\frac{i+s}{n}},~~~~~m_{j+k-s}(t)\leq C m_n(t)^{\frac{j+k-s}{n}},\\
~~~~~m_{j+k-s+1/2}(t)&\leq &C m_n(t)^{\frac{j+k-s+1/2}{n}}, ~~~~~m_{i}(t)\leq C m_n(t)^{\frac{i}{n}},\\
~~~~~m_{j-1}(t)&\leq &C m_n(t)^{\frac{j-1}{n}}, ~~~~~m_{j-1/2}(t)\leq C m_n(t)^{\frac{j-1/2}{n}}
,\\
~~~~~m_{k-1}(t)&\leq &C m_n(t)^{\frac{k-1}{n}}, ~~~~~m_{k-1/2}(t)\leq C m_n(t)^{\frac{k-1/2}{n}},
\end{eqnarray*}
we obtain
\begin{equation}\label{Propo:MomentsPropa:E3}
\begin{aligned}
&~~\int_{\mathbb{R}^3}C_{22}[f](p_1)\mathcal{E}_{p_1}^{n}dp_1\le\\
\le &~~C\sum_{0\le i, j, k<n;~ i+j+k=n}\sum_{s=0}^{k+1}m_n(t)^{\frac{i+s}{n}}\left(m_n(t)^{\frac{j+k-s}{n}}+m_n(t)^{\frac{j+k-s+1/2}{n}}\right)+\\
&~~+C\sum_{0\le i, j, k<n;~ i+j+k=n:~j,k>0}m_n(t)^{\frac{i}{n}}\left(m_n(t)^{\frac{j-1}{n}}+m_n(t)^{\frac{j-1/2}{n}}\right)\left(m_n(t)^{\frac{k-1}{n}}+m_n(t)^{\frac{k-1/2}{n}}\right).
\end{aligned}
\end{equation}
Combining \eqref{Propo:MomentsPropa:E1}, \eqref{Propo:MomentsPropa:E2} and \eqref{Propo:MomentsPropa:E3} yields
 \begin{equation}\label{Propo:MomentsPropa:E4}\begin{aligned}
&~~\frac{d}{dt} m_n(t)\\
\le&~~{C} m_n(t)+Cm_n(t)^{\frac{n-1}{n}}+Cm_n(t)^{\frac{n-2}{n}}+ Cm_n(t)^{\frac{1}{n}}-Cm_{n}^{\frac{n+1}{n}}\\
 &~~+C\sum_{0\le i, j, k<n;~ i+j+k=n}\sum_{s=0}^{k+1}m_n(t)^{\frac{i+s}{n}}\left(m_n(t)^{\frac{j+k-s}{n}}+m_n(t)^{\frac{j+k-s+1/2}{n}}\right)+\\
&~~+C\sum_{0\le i, j, k<n;~ i+j+k=n;~j,k>0}m_n(t)^{\frac{i}{n}}\left(m_n(t)^{\frac{j-1}{n}}+m_n(t)^{\frac{j-1/2}{n}}\right)\left(m_n(t)^{\frac{k-1}{n}}+m_n(t)^{\frac{k-1/2}{n}}\right),
\end{aligned}
\end{equation}
where $C$ depends on  $n$, $k$, and the bound of the mass on $[0,T]$ in Proposition \ref{Propo:Mass}. 
Notice that $-Cm_n(t)^{\frac{n+1}{n}}$ has the highest order on the right hand side of \eqref{Propo:MomentsPropa:E4}. By the same argument as in \cite{Wennberg:1997:EDM}, the conclusion of the theorem then follows.
\end{proof}
\subsection{Holder estimates for the collision operators}\label{Sec:HolderEstimate}
In this section, we will provide Holder estimates for the two collision operators $C_{12}$ and $C_{22}$. For $C_{22}$, we split it into two operators
\begin{equation}\begin{aligned}\label{Def:C221}
C_{22}^1[f](p_1)=&~~\kappa_3\iiint_{\mathbb{R}_+\times\mathbb{R}_+\times\mathbb{R}_+ }K^{22}(p_1,p_2,p_3,p_4)\frac{\min\{|p_1|,|p_2|,|p_3|,|p_4|\}|p_1||p_2||p_3||p_4|}{|p_1|^2}\\
&~~\times\delta(\mathcal{E}_{p_1}+\mathcal{E}_{p_2}-\mathcal{E}_{p_3}-\mathcal{E}_{p_4})[f({p}_3)f({p}_4)-f({p}_1)f({p}_2)]d|{p}_2|d|{p}_3|d|{p}_4|,
\end{aligned}\end{equation}
and
\begin{equation}\begin{aligned}\label{Def:C222}
C_{22}^2[f](p_2)=&~~\kappa_3\iiint_{\mathbb{R}_+\times\mathbb{R}_+\times\mathbb{R}_+ }K^{22}(p_1,p_2,p_3,p_4)\frac{\min\{|p_1|,|p_2|,|p_3|,|p_4|\}|p_1||p_2||p_3||p_4|}{|p_1|^2}\\
&~~\times\delta(\mathcal{E}_{p_1}+\mathcal{E}_{p_2}-\mathcal{E}_{p_3}-\mathcal{E}_{p_4})[f({p}_3)f({p}_4)(f({p}_1)+f({p}_2))-\\
&~~-f({p}_1)f({p}_2)(f({p}_3)+f({p}_4))]d|{p}_2|d|{p}_3|d|{p}_4|,
\end{aligned}\end{equation}
We will show in Proposition \ref{Propo:C12}, Proposition \ref{Propo:HolderC221} and Proposition \ref{Propo:HolderC222N} that $C_{12}$, $C_{22}^1$ and $C_{22}^{2}$ are Holder continuous.  
\subsubsection{Holder estimates for $C_{12}$}\label{Sec:HolderEstimateC12}
\begin{proposition}\label{Propo:HolderC12} Let $f$ and $g$ be two functions in $L^1_{n+3}(\mathbb{R}^3)\cap L^1(\mathbb{R}^3)$, $n\in\mathbb{R}_+$, $n$ can be $0$; then there exists a constant $\mathcal{C}$ depending on $\|f\|_{L^1_{n+3}}, \|f\|_{L^1}, \|g\|_{L^1_{n+3}}, \|g\|_{L^1}$ such that
\begin{equation}\label{Propo:HolderC12:1} 
\|C_{12}[f]-C_{12}[g]\|_{L^1_{n}}\le \mathcal{C}\left(\|f-g\|_{L^1_{n+3}}+\|f-g\|_{L^1}\right).\end{equation}
If $\|f\|_{{L}^1_{n+4}},\|g\|_{{L}^1_{n+4}}<\mathcal{C}_0$, then
\begin{equation}\label{Propo:HolderC12:2} 
\|C_{12}[f]-C_{12}[g]\|_{L^1_{n}}\le \mathcal{C}_1\left(\|f-g\|_{L^1}^{\frac{1}{n+4}}+\|f-g\|_{L^1}\right),\end{equation}
where $\mathcal{C}_1$ is a constant depending on $\mathcal{C}_0$, $\mathcal{C}$.
\end{proposition}
\begin{proof}
First, let us consider the $L^1_n$ norm of the difference $C_{12}[f]-C_{12}[g]$. As a view of Lemma \ref{Lemma:WeakFormulation}
\begin{equation}\label{Propo:HolderC12:E1}\begin{aligned}
\|C_{12}[f]-C_{12}[g]\|_{L^1_n}~~=&~~\int_{\mathbb{R}^3}|p_1|^n|C_{12}[f]-C_{12}[g]|dp_1\\
~~\le&~~n_c\lambda_1\iiint_{\mathbb{R}^{3\times3}}K^{12}(p_1,p_2,p_3)\delta(p_1-p_2-p_3)\delta(\mathcal{E}_{p_1}-\mathcal{E}_{p_2}-\mathcal{E}_{p_3})\\
~~&~~\times|f(p_2)f(p_3)-2f(p_3)f(p_1)-f(p_1)-g(p_2)g(p_3)\\
~~&~~+2g(p_3)g(p_1)+g(p_1)|\left[|p_1|^n+|p_2|^n+|p_3|^n\right]dp_1dp_2dp_3.
\end{aligned}
\end{equation}
The above identity implies that $\|C_{12}[f]-C_{12}[g]\|_{L^1_n}$ can be bounded by the sum of the following three terms
$$\begin{aligned}
N_1~~=&~~n_c\lambda_1\iiint_{\mathbb{R}^{3\times3}}K^{12}(p_1,p_2,p_3)\delta(p_1-p_2-p_3)\delta(\mathcal{E}_{p_1}-\mathcal{E}_{p_2}-\mathcal{E}_{p_3})\\
~~&~~\times|f(p_2)f(p_3)-g(p_2)g(p_3)|\left[|p_1|^n+|p_2|^n+|p_3|^n\right]dp_1dp_2dp_3,
\end{aligned}
$$$$\begin{aligned}
N_2~~=&~~2n_c\lambda_1\iiint_{\mathbb{R}^{3\times3}}K^{12}(p_1,p_2,p_3)\delta(p_1-p_2-p_3)\delta(\mathcal{E}_{p_1}-\mathcal{E}_{p_2}-\mathcal{E}_{p_3})\\
~~&~~\times|f(p_3)f(p_1)-g(p_3)g(p_1)|\left[|p_1|^n+|p_2|^n+|p_3|^n\right]dp_1dp_2dp_3,
\end{aligned}
$$and
$$\begin{aligned}
N_3~~=&~~n_c\lambda_1\iiint_{\mathbb{R}^{3\times3}}K^{12}(p_1,p_2,p_3)\delta(p_1-p_2-p_3)\delta(\mathcal{E}_{p_1}-\mathcal{E}_{p_2}-\mathcal{E}_{p_3})\\
~~&~~\times|f(p_1)-g(p_1)|\left[|p_1|^n+|p_2|^n+|p_3|^n\right]dp_1dp_2dp_3.
\end{aligned}
$$In the sequel, we will estimate $N_1$, $N_2$, $N_3$ in three steps.
{\\\bf Step 1: Estimating $N_1$.} 
\\ By the definition of $\delta(p_1-p_2-p_3)$, $N_1$ can be rewritten as:
$$\begin{aligned}
N_1~~=&~~n_c\lambda_1\iint_{\mathbb{R}^{3\times2}}K^{12}(p_2+p_3,p_2,p_3)\delta(\mathcal{E}_{p_2+p_3}-\mathcal{E}_{p_2}-\mathcal{E}_{p_3})\\
~~&~~\times|f(p_2)f(p_3)-g(p_2)g(p_3)|\left[|p_2+p_3|^n+|p_2|^n+|p_3|^n\right]dp_2dp_3.
\end{aligned}
$$
By the triangle inequality,
$$|f(p_2)f(p_3)-g(p_2)g(p_3)|\leq |f(p_2)-g(p_2)||f(p_3)|+|f(p_3)-g(p_3)||g(p_2)|,$$
the term $N_1$ can be bounded as 
$$\begin{aligned}
N_1~~\le&~~n_c\lambda_1\iint_{\mathbb{R}^{3\times2}}K^{12}(p_2+p_3,p_2,p_3)\delta(\mathcal{E}_{p_2+p_3}-\mathcal{E}_{p_2}-\mathcal{E}_{p_3})\\
~~&~~\times|f(p_2)-g(p_2)||f(p_3)|\left[|p_2+p_3|^n+|p_2|^n+|p_3|^n\right]dp_2dp_3\\
&~~+n_c\lambda_1\iint_{\mathbb{R}^{3\times2}}K^{12}(p_2+p_3,p_2,p_3)\delta(\mathcal{E}_{p_2+p_3}-\mathcal{E}_{p_2}-\mathcal{E}_{p_3})\\
~~&~~\times|f(p_3)-g(p_3)||g(p_2)|\left[|p_2+p_3|^n+|p_2|^n+|p_3|^n\right]dp_2dp_3
.
\end{aligned}
$$
Again, by the triangle inequality
$$|p_2+p_3|^n\le (|p_2|+|p_3|)^n\le 2^{n-1}(|p_2|^n+|p_3|^n),$$
one can estimate $N_1$ as
$$\begin{aligned}
N_1~~\le&~~C\iint_{\mathbb{R}^{3\times2}}\delta(\mathcal{E}_{p_2+p_3}-\mathcal{E}_{p_2}-\mathcal{E}_{p_3})K^{12}(p_2+p_3,p_2,p_3)\times\\
~~&~~\times|f(p_2)-g(p_2)||f(p_3)|\left[|p_2|^n+|p_3|^n\right]dp_2dp_3\\
&~~+C\iint_{\mathbb{R}^{3\times2}}\delta(\mathcal{E}_{p_2+p_3}-\mathcal{E}_{p_2}-\mathcal{E}_{p_3})K^{12}(p_2+p_3,p_2,p_3)\times\\
&~~\times|f(p_3)-g(p_3)||g(p_2)|\left[|p_2|^n+|p_3|^n\right]dp_2dp_3
,
\end{aligned}
$$
where $C$ is a constant varying from line to line. The above estimate can be rewritten, taking into account the definition of $\delta(\mathcal{E}_{p_2+p_3}-\mathcal{E}_{p_2}-\mathcal{E}_{p_3})$, as
$$\begin{aligned}
N_1~~\le&~~C\int_{\mathbb{R}^{3}}\int_{S_{p_3}^1}K^{12}_1(p_2+p_3,p_2,p_3)|f(p_2)-g(p_2)||f(p_3)|\left[|p_2|^n+|p_3|^n\right]d\sigma(p_3)dp_2\\
&~~+C\int_{\mathbb{R}^{3}}\int_{S_{p_2}^1}K^{12}_1(p_2+p_3,p_2,p_3)|f(p_3)-g(p_3)||g(p_2)|\left[|p_2|^n+|p_3|^n\right]d\sigma(p_2)dp_3
.
\end{aligned}
$$
By Lemma \ref{lem-Sp1}, one can estimate $N_1$ as follows
$$\begin{aligned}
N_1~~\le&~~C\iint_{\mathbb{R}^{3\times 2}}|f(p_2)-g(p_2)||f(p_3)|\frac{K^{12}(p_2+p_3,p_2,p_3)}{|p_2||p_3|}\left[|p_2|^n+|p_3|^n\right]dp_3dp_2\\
&~~+C\iint_{\mathbb{R}^{3\times 2}}|f(p_3)-g(p_3)||g(p_2)|\frac{K^{12}(p_2+p_3,p_2,p_3)}{|p_2||p_3|}\left[|p_2|^n+|p_3|^n\right]dp_2dp_3
.
\end{aligned}
$$
Since $\frac{K^{12}(p_2+p_3,p_2,p_3)}{|p_2||p_3|}$ and $\frac{K^{12}(p_2+p_3,p_2,p_3)}{|p_2||p_3|}$ are bounded, $N_1$ is bounded as
$$\begin{aligned}
N_1~~\le&~~C\iint_{\mathbb{R}^{3\times 2}}|f(p_2)-g(p_2)||f(p_3)|\left[|p_2|^n+|p_3|^n\right]dp_3dp_2\\
&~~+C\iint_{\mathbb{R}^{3\times 2}}|f(p_3)-g(p_3)||g(p_2)|\left[|p_2|^n+|p_3|^n\right]dp_2dp_3,
\end{aligned}
$$
which leads to the following straightforward estimates on $N_1$
\begin{equation}\label{Propo:HolderC12:E2}\begin{aligned}
N_1\le&~~C\int_{\mathbb{R}^{3}}|f(p_2)-g(p_2)||p_2|^ndp_2\int_{\mathbb{R}^{3}}|f(p_3)|dp_3\\
&~~+C\int_{\mathbb{R}^{3}}|f(p_2)-g(p_2)|dp_2\int_{\mathbb{R}^{3}}|f(p_3)||p_3|^ndp_3\\
&~~+C\int_{\mathbb{R}^{3}}|f(p_3)-g(p_3)||p_3|^ndp_3\int_{\mathbb{R}^{3}}|f(p_2)|dp_2\\
&~~+C\int_{\mathbb{R}^{3}}|f(p_3)-g(p_3)|dp_3\int_{\mathbb{R}^{3}}|f(p_2)||p_2|^ndp_2\\
\le&~~C\int_{\mathbb{R}^{3}}|f(p_1)-g(p_1)||p_1|^ndp_1+C\int_{\mathbb{R}^{3}}|f(p_1)-g(p_1)|dp_1.
\end{aligned}
\end{equation}
{\bf Step 2: Estimating $N_2$.} 
\\ By the definition of $\delta(p_1-p_2-p_3)$, $N_2$ can be rewritten as:
$$\begin{aligned}
N_2~~=&~~2n_c\lambda_1\iint_{\mathbb{R}^{3\times2}}K^{12}(p_1,p_1-p_3,p_3)\delta(\mathcal{E}_{p_1}-\mathcal{E}_{p_1-p_3}-\mathcal{E}_{p_3})\\
~~&~~\times|f(p_3)f(p_1)-g(p_3)g(p_1)|\left[|p_1|^n+|p_1-p_3|^n+|p_3|^n\right]dp_1dp_3,
\end{aligned}
$$
which, by the inequality,
$$|p_1-p_3|^n\le (|p_1|+|p_3|)^n\le 2^{n-1}(|p_1|^n+|p_3|^n),$$
can be bounded as
$$\begin{aligned}
N_2~~\le &~~C\iint_{\mathbb{R}^{3\times2}}K^{12}(p_1,p_1-p_3,p_3)\delta(\mathcal{E}_{p_1}-\mathcal{E}_{p_1-p_3}-\mathcal{E}_{p_3})\\
~~&~~\times|f(p_3)-g(p_3)||f(p_1)|\left[|p_1|^n+|p_3|^n\right]dp_1dp_3\\
&~~+C\iint_{\mathbb{R}^{3\times2}}K^{12}(p_1,p_1-p_3,p_3)\delta(\mathcal{E}_{p_1}-\mathcal{E}_{p_1-p_3}-\mathcal{E}_{p_3})\\
~~&~~\times|f(p_1)-g(p_1)||g(p_3)|\left[|p_1|^n+|p_3|^n\right]dp_1dp_3.
\end{aligned}
$$
Employing the definition of $\delta(\mathcal{E}_{p_1}-\mathcal{E}_{p_1-p_3}-\mathcal{E}_{p_3})$, one can estimate $N_2$ as
$$\begin{aligned}
N_2~~\le &~~C\int_{\mathbb{R}^{3}}\int_{S_{p_1}^0}K^{12}_0(p_1,p_1-p_3,p_3)|f(p_3)-g(p_3)||f(p_1)|\left[|p_1|^n+|p_3|^n\right]d\sigma(p_3)dp_1\\
&~~+C\int_{\mathbb{R}^{3}}\int_{S_{p_1}^0}K^{12}_0(p_1,p_1-p_3,p_3)|f(p_1)-g(p_1)||g(p_3)|\left[|p_1|^n+|p_3|^n\right]dd\sigma(p_3)dp_1,
\end{aligned}
$$
which, by Lemma \ref{lem-Sp}, yields
$$\begin{aligned}
N_2~~\le &~~\int_{\mathbb{R}^3}\int_0^{|p_1|}K^{12}(p_1,p_1-p_3,p_3)|f(p_1)-g(p_1)||g(p_3)|\left[|p_1|^n+|p_3|^n\right]|p_3|d|p_3|dp_1\\
 &~~+\int_{\mathbb{R}^3}\int_0^{|p_1|}K^{12}(p_1,p_1-p_3,p_3)|f(p_3)-g(p_3)||f(p_1)|\left[|p_1|^n+|p_3|^n\right]|p_3|d|p_3|dp_1.
\end{aligned}
$$
Bounding the integral from $0$ to $|p_1|$ by an integral from $0$ to $\infty$ implies
$$\begin{aligned}
N_2~~\le &~~\int_{\mathbb{R}^3}\int_0^{\infty}K^{12}(p_1,p_1-p_3,p_3)|f(p_1)-g(p_1)||g(p_3)|\left[|p_1|^n+|p_3|^n\right]|p_3|d|p_3|dp_1\\
 &~~+\int_{\mathbb{R}^3}\int_0^{\infty}K^{12}(p_1,p_1-p_3,p_3)|f(p_3)-g(p_3)||f(p_1)|\left[|p_1|^n+|p_3|^n\right]|p_3|d|p_3|dp_1.
\end{aligned}
$$
We now switch the integral from $d|p_3|$ to $dp_3$ from the above inequality to obtain
$$\begin{aligned}
N_2~~\le &~~\int_{\mathbb{R}^{3\times2}}\frac{K^{12}(p_1,p_1-p_3,p_3)}{|p_3|}|f(p_1)-g(p_1)||g(p_3)|\left[|p_1|^n+|p_3|^n\right]dp_3dp_1\\
&~~+\int_{\mathbb{R}^{3\times2}}(1+|p_1|)\frac{K^{12}(p_1,p_1-p_3,p_3)}{|p_3|}|f(p_3)-g(p_3)||f(p_1)|\left[|p_1|^n+|p_3|^n\right]dp_3dp_1.
\end{aligned}
$$
Applying the inequality
$$|p_1|^n+|p_3|^n\le C(1+|p_1|^{n}+|p_3|^{n})$$
to the above bound on $N_2$, we get
$$\begin{aligned}
&~~N_2\le\\
~~\le &~~C\int_{\mathbb{R}^{3\times 2}}\frac{K^{12}(p_1,p_1-p_3,p_3)}{|p_3|}|f(p_3)-g(p_3)||f(p_1)|\left[1+|p_1|^{n}+|p_3|^{n}\right]dp_3dp_1\\
&~~+C\int_{\mathbb{R}^{3\times 2}}\frac{K^{12}(p_1,p_1-p_3,p_3)}{|p_3|}|f(p_1)-g(p_1)||g(p_3)|\left[1+|p_1|^{n}+|p_3|^{n}\right]dp_3dp_1.
\end{aligned}
$$
The same argument as for \eqref{Propo:HolderC12:E2} yields
\begin{equation}\label{Propo:HolderC12:E2b}\begin{aligned}
N_2
\le&~~C\int_{\mathbb{R}^{3}}|f(p_1)-g(p_1)||p_1|^{n}dp_1+C\int_{\mathbb{R}^{3}}|f(p_1)-g(p_1)|dp_1.
\end{aligned}\end{equation}
{\bf Step 3: Estimating $N_3$.} 
\\ By the definition of $\delta(p_1-p_2-p_3)$, $N_3$ can be rewritten as:
$$\begin{aligned}
N_3~~=&~~n_c\lambda_1\iint_{\mathbb{R}^{3\times2}}K^{12}(p_1,p_2,p_1-p_2)\delta(\mathcal{E}_{p_1}-\mathcal{E}_{p_2}-\mathcal{E}_{p_1-p_2})\\
~~&~~\times|f(p_1)-g(p_1)|\left[|p_1|^n+|p_2|^n+|p_1-p_2|^n\right]dp_1dp_2,
\end{aligned}$$
which, by the inequality,
$$|p_1-p_2|^n\le (|p_1|+|p_2|)^n\le 2^{n-1}(|p_1|^n+|p_2|^n),$$
can be bounded as
$$\begin{aligned}
N_3~~\le&~~C\int_{\mathbb{R}^{3}}\int_{S_{p_1}}K^{12}(p_1,p_2,p_1-p_2)|f(p_1)-g(p_1)|\left[|p_1|^n+|p_2|^n\right]d\sigma(p_2)dp_1.
\end{aligned}$$
Now, as an application of Lemma \ref{lem-Sp},
$$\begin{aligned}
\int_{S_{p_1}}(|p_1|^n+|p_2|^n)d\sigma(p_2)~~&\le~~ C\left(|p_1|^{n+2}+\int_{S_{p_1}}|p_2|^nd\sigma(p_2)\right)\\
~~&\le~~  C\left(|p_1|^{n+2}+\int_{0}^{|p_1|}|p_2|^{n+1}d|p_2|\right)\\
~~&\le~~ C\left(1+|p_1|^{n+3}\right),
\end{aligned}$$
which together with the fact that ${K^{12}(p_1,p_2,p_1-p_2)}$ is bounded, implies 
\begin{equation}\label{Propo:HolderC12:E3}\begin{aligned}
N_3~~\le&~~C\int_{\mathbb{R}^{3}}|f(p_1)-g(p_1)|\left[|p_1|^{n+3}+1\right]dp_1.
\end{aligned}\end{equation}
Combining \eqref{Propo:HolderC12:E2}, \eqref{Propo:HolderC12:E2b}, and \eqref{Propo:HolderC12:E3} yields
\begin{equation}\label{Propo:HolderC12:E4}\begin{aligned}
\|C_{12}[f]-C_{12}[g]\|_{L^1_n}~~\le&~~C\int_{\mathbb{R}^{3}}|f(p_1)-g(p_1)|\left[|p_1|^{n+3}+|p_1|^{n+1}+|p_1|^{n}+1\right]dp_1.
\end{aligned}
\end{equation}
Since $$|p|^n\leq C\left(|p|^{n+3}+1\right), |p|^{n+1}\leq C\left(|p|^{n+3}+1\right),$$
Inequality \eqref{Propo:HolderC12:1} follows from \eqref{Propo:HolderC12:E4}.  Inequality \eqref{Propo:HolderC12:2} is a consequence of Inequality \eqref{Propo:HolderC12:1}, Lemma \ref{Lemma:Holder} and 
$$\|f-g\|_{L^1_{n+3}}\le \|f-g\|_{L^1}^{\frac{1}{n+4}}\left(\|f\|_{L^1_{n+4}}+\|g\|_{L^1_{n+4}}\right)^{\frac{n+3}{n+4}}.$$
\end{proof}
\subsubsection{Holder estimates for $C_{22}^1$}\label{Sec:HolderEstimateC221}
\begin{proposition}\label{Propo:HolderC221} Let $f$ and $g$ be two functions in $L^1_n(\mathbb{R}^3)\cap L^1(\mathbb{R}^3)$, $n\in\mathbb{N},$ $n/2$ is an odd number, or $n=0$, then there exists a constant $\mathcal{C}$ depending on $\|f\|_{L^1_{n+1}}, \|f\|_{L^1}, \|g\|_{L^1_{n+1}}, \|g\|_{L^1}$ such that 
\begin{equation}\label{Propo:HolderC221:1} 
\|C_{22}^1[f]-C_{22}^1[g]\|_{L^1_{n}}\le \mathcal{C}\left(\|f-g\|_{L^1_{n+1}}+\|f-g\|_{L^1}\right).\end{equation}
If $\|f\|_{{L}^1_{n+2}},\|g\|_{{L}^1_{n+2}}<\mathcal{C}_0$, then
\begin{equation}\label{Propo:HolderC221:2} 
\|C_{22}^1[f]-C_{22}^1[g]\|_{L^1_{n}}\le \mathcal{C}_1\left(\|f-g\|_{L^1}^{\frac{1}{n+2}}+\|f-g\|_{L^1}\right),\end{equation}
where $\mathcal{C}_1$ is a constant depending on $\mathcal{C}_0$, $\mathcal{C}$.
\end{proposition}
\begin{proof}
Let us consider the $L^1_n$ norm of the difference $C_{22}^1[f]-C_{22}^1[g]$. As a view of Lemma \ref{Lemma:WeakFormulation}
$$\begin{aligned}
\begin{aligned}
&~~\int_{\mathbb{R}^3}\left|C_{22}^1[f](p_1)-C_{22}^1[g](p_1)\right||p_1|^{n}dp_1\\
\le&~~C\int_{\mathbb{R}_+^{4}}K^{22}(p_1,p_2,p_3,p_4){\min\{|p_1|,|p_2|,|p_3|,|p_4|\}}|p_1||p_2||p_3||p_4|\delta(\mathcal{E}_{p_1}+\mathcal{E}_{p_2}-\mathcal{E}_{p_3}-\mathcal{E}_{p_4}) \\
&~~\times|f(p_1)f(p_2)-g(p_1)g(p_2)|\Big[|p_4|^{n}+|p_3|^{n}+|p_2|^{n}+|p_1|^{n}\Big]d|{p}_1|d|{p}_2|d|{p}_3|d|p_4|,
\end{aligned}
\end{aligned}
$$
By the inequality $$|p|^n\leq C\mathcal{E}_p^{n/2},$$
one gets
$$|p_4|^{n}+|p_3|^{n}+|p_2|^{n}+|p_1|^{n}\leq C\mathcal{E}_{p_4}^{n/2}+C\mathcal{E}_{p_3}^{n/2}+C\mathcal{E}_{p_2}^{n/2}+C\mathcal{E}_{p_1}^{n/2},$$
which implies
$$\begin{aligned}
&~~\int_{\mathbb{R}^3}\left|C_{22}^1[f](p_1)-C_{22}^1[g](p_1)\right||p_1|^{n}dp_1\\
\le&~~C\int_{\mathbb{R}_+^{4}}K^{22}(p_1,p_2,p_3,p_4){\min\{|p_1|,|p_2|,|p_3|,|p_4|\}}|p_1||p_2||p_3||p_4|\delta(\mathcal{E}_{p_1}+\mathcal{E}_{p_2}-\mathcal{E}_{p_3}-\mathcal{E}_{p_4}) \\
&~~\times|f(p_1)f(p_2)-g(p_1)g(p_2)|\Big[\mathcal{E}_{p_4}^{n/2}+\mathcal{E}_{p_3}^{n/2}+\mathcal{E}_{p_2}^{n/2}+\mathcal{E}_{p_1}^{n/2}\Big]d|{p}_1|d|{p}_2|d|{p}_3|d|p_4|.
\end{aligned}
$$
Now, thanks to the Dirac function $\delta(\mathcal{E}_{p_1}+\mathcal{E}_{p_2}-\mathcal{E}_{p_3}-\mathcal{E}_{p_4})$, one can write $\mathcal{E}_{p_4}$ as $\mathcal{E}_{p_1}+\mathcal{E}_{p_2}-\mathcal{E}_{p_3}$, which implies
$$\begin{aligned}
\begin{aligned}
&~~\int_{\mathbb{R}^3}\left|C_{22}^1[f](p_1)-C_{22}^1[g](p_1)\right||p_1|^{n}dp_1\\
\le&~~C\int_{\mathbb{R}_+^{4}}K^{22}(p_1,p_2,p_3,p_4){\min\{|p_1|,|p_2|,|p_3|,|p_4|\}}|p_1||p_2||p_3||p_4|\delta(\mathcal{E}_{p_1}+\mathcal{E}_{p_2}-\mathcal{E}_{p_3}-\mathcal{E}_{p_4}) \\
&~~\times|f(p_1)f(p_2)-g(p_1)g(p_2)|\Big[\left(\mathcal{E}_{p_1}+\mathcal{E}_{p_2}-\mathcal{E}_{p_3}\right)^{n/2}+\mathcal{E}_{p_3}^{n/2}+\mathcal{E}_{p_2}^{n/2}+\mathcal{E}_{p_1}^{n/2}\Big]d|{p}_1|d|{p}_2|d|{p}_3|d|p_4|.
\end{aligned}
\end{aligned}
$$
Similar as for \eqref{Propo:C22:E3}, $|p_4|d|p_4|$ can be bounded by $Cd\mathcal{E}_{p_4}$ and $\min\{|p_1|,|p_2|,|p_3|,|p_4|\}$ can be bounded by $\min\{|p_1|,|p_2|,|p_3|\}$. Moreover, $K^{22}(p_1,p_2,p_3,p_4)$ is bounded by $\Gamma$ due to \eqref{K22Bound}. As a consequence,
$$\begin{aligned}
\begin{aligned}
&~~\int_{\mathbb{R}^3}\left|C_{22}^1[f](p_1)-C_{22}^1[g](p_1)\right||p_1|^{n}dp_1\\
\le&~~C\int_{\mathbb{R}_+^{4}}{\min\{|p_1|,|p_2|,|p_3|\}}|p_1||p_2||p_3|\delta(\mathcal{E}_{p_1}+\mathcal{E}_{p_2}-\mathcal{E}_{p_3}-\mathcal{E}_{p_4}) |f(p_1)f(p_2)-\\
&~~-g(p_1)g(p_2)|\Big[\left(\mathcal{E}_{p_1}+\mathcal{E}_{p_2}-\mathcal{E}_{p_3}\right)^{n/2}+\mathcal{E}_{p_3}^{n/2}+\mathcal{E}_{p_2}^{n/2}+\mathcal{E}_{p_1}^{n/2}\Big]d|{p}_1|d|{p}_2|d|{p}_3|d\mathcal{E}_{p_4}\\
\le&~~C\int_{\mathcal{E}_{p_3}\leq \mathcal{E}_{p_1}+\mathcal{E}_{p_2}}{\min\{|p_1|,|p_2|,|p_3|\}}|p_1||p_2||p_3| |f(p_1)f(p_2)-\\
&~~-g(p_1)g(p_2)|\Big[\left(\mathcal{E}_{p_1}+\mathcal{E}_{p_2}-\mathcal{E}_{p_3}\right)^{n/2}+\mathcal{E}_{p_3}^{n/2}+\mathcal{E}_{p_2}^{n/2}+\mathcal{E}_{p_1}^{n/2}\Big]d|{p}_1|d|{p}_2|d|{p}_3|,
\end{aligned}
\end{aligned}
$$
where in the last inequality, we have taken the integration with respect to $d\mathcal{E}_{p_4}$.\\
Since $n/2$ is an odd number, by Newton formula
$$\left(\mathcal{E}_{p_1}+\mathcal{E}_{p_2}-\mathcal{E}_{p_3}\right)^{n/2}+\mathcal{E}_{p_3}^{n/2}+\mathcal{E}_{p_2}^{n/2}+\mathcal{E}_{p_1}^{n/2}=\sum_{0\le i,j,k~;~i+j+k=n/2~;~k\ne n/2} B_{i,j,k,n}\mathcal{E}_{p_1}^i\mathcal{E}_{p_2}^j\mathcal{E}_{p_3}^k,$$
we obtain
\begin{equation}\label{Propo:HolderC221:E1}\begin{aligned}
\int_{\mathbb{R}^3}\left|C_{22}^1[f](p_1)-C_{22}^1[g](p_1)\right||p_1|^{n}dp_1
\le X,
\end{aligned}
\end{equation}
where 
$$\begin{aligned}
X:=&~~C\int_{\mathcal{E}_{p_3}\leq \mathcal{E}_{p_1}+\mathcal{E}_{p_2}}{\min\{|p_1|,|p_2|,|p_3|\}}|p_1||p_2||p_3| \Big|f(p_1)f(p_2)-\\
&~~-g(p_1)g(p_2)\Big|\left[\sum_{0\le i,j,k~;~i+j+k=n/2~;~k\ne n/2} |B_{i,j,k,n}|\mathcal{E}_{p_1}^i\mathcal{E}_{p_2}^j\mathcal{E}_{p_3}^k\right]d|{p}_1|d|{p}_2|d|{p}_3|.
\end{aligned}
$$
The rest of the proof is devoted to estimates of $X$.
\\ Similar as for \eqref{Propo:C22:E3}, $|p_3|d|p_3|$ can be bounded by $Cd\mathcal{E}_{p_3}$ and $\min\{|p_1|,|p_2|,|p_3|\}$ can be bounded by $\min\{|p_1|,|p_2|\}$:
$$\begin{aligned}
X_1\le&~~C\int_{\mathcal{E}_{p_3}\leq \mathcal{E}_{p_1}+\mathcal{E}_{p_2}}{\min\{|p_1|,|p_2|\}}|p_1||p_2| \Big|f(p_1)f(p_2)-\\
&~~-g(p_1)g(p_2)\Big|\left[\sum_{0\le i,j,k~;~i+j+k=n/2~;~k\ne n/2} \mathcal{E}_{p_1}^i\mathcal{E}_{p_2}^j\mathcal{E}_{p_3}^k\right]d|{p}_1|d|{p}_2|\mathcal{E}_{p_3}.
\end{aligned}
$$
Integrating with respect to $d\mathcal{E}_{p_3}$ the above integral and using Newton formula yields
$$\begin{aligned}
X\le&~~C\int_{\mathbb{R}_+^2}{\min\{|p_1|,|p_2|\}}|p_1||p_2| |f(p_1)f(p_2)-\\
&~~-g(p_1)g(p_2)|\left[\sum_{0\le i,j,k~;~i+j+k=n/2~;~k\ne n/2} \sum_{s=0}^{k+1} {{k+1}\choose{s}}\mathcal{E}_{p_1}^{i+s}\mathcal{E}_{p_2}^{k+1+j-s}\right]d|{p}_1|d|{p}_2|\\
\le&~~\sum_{0\le i,j,k~;~i+j+k=n/2~;~k\ne n/2} \sum_{s=0}^{k+1}C\int_{\mathbb{R}_+^2}{\min\{|p_1|,|p_2|\}}|p_1||p_2| |f(p_1)f(p_2)-\\
&~~-g(p_1)g(p_2)|\mathcal{E}_{p_1}^{i+s}\mathcal{E}_{p_2}^{k+1+j-s}d|{p}_1|d|{p}_2|\\
\le&~~\sum_{0\le i,j,k~;~i+j+k=n/2~;~k\ne n/2} \sum_{s=0;i+s\ne 0}^{k+1}C\int_{\mathbb{R}_+^2}{\min\{|p_1|,|p_2|\}}|p_1||p_2| |f(p_1)f(p_2)-\\
&~~-g(p_1)g(p_2)|\mathcal{E}_{p_1}^{i+s}\mathcal{E}_{p_2}^{k+1+j-s}d|{p}_1|d|{p}_2|+\\
&~~+C\int_{\mathbb{R}_+^2}{\min\{|p_1|,|p_2|\}}|p_1||p_2| |f(p_1)f(p_2)-g(p_1)g(p_2)|\mathcal{E}_{p_2}^{n/2+1}d|{p}_1|d|{p}_2|
.
\end{aligned}
$$
By the inequalities
$${\min\{|p_1|,|p_2|\}}|p_1||p_2|\le |p_1||p_2|^2,$$
and
$${\min\{|p_1|,|p_2|\}}|p_1||p_2|\le |p_1|^2|p_2|,$$
one deduce that
\begin{equation}\label{Propo:HolderC221:E2}\begin{aligned}
X
\le X_{1}+ X_{2},
\end{aligned}
\end{equation}
where
$$\begin{aligned}
X_{1}:=&~~\sum_{0\le i,j,k~;~i+j+k=n/2~;~k\ne n/2} \sum_{s=0;i+s\ne 0}^{k+1}C\int_{\mathbb{R}_+^2}|p_1||p_2|^2\Big|f(p_1)f(p_2)-\\
&~~-g(p_1)g(p_2)\Big|\mathcal{E}_{p_1}^{i+s}\mathcal{E}_{p_2}^{k+1+j-s}d|{p}_1|d|{p}_2|;\\
X_{2}:=&~~C\int_{\mathbb{R}_+^2}|p_1|^2|p_2| \Big|f(p_1)f(p_2)-g(p_1)g(p_2)\Big|\mathcal{E}_{p_2}^{n/2+1}d|{p}_1|d|{p}_2|
.
\end{aligned}
$$
Let us first estimate $X_{1}$ by looking at the terms inside the sum
$$\begin{aligned}
&~~\int_{\mathbb{R}_+^2}|p_1||p_2|^2|f(p_1)f(p_2)-g(p_1)g(p_2)|\mathcal{E}_{p_1}^{i+s}\mathcal{E}_{p_2}^{k+1+j-s}d|{p}_1|d|{p}_2|\\
\le&~~\int_{\mathbb{R}_+^2}|p_1||p_2|^2|f(p_1)-g(p_1)||g(p_2)|\mathcal{E}_{p_1}^{i+s}\mathcal{E}_{p_2}^{k+1+j-s}d|{p}_1|d|{p}_2|\\
 &~~+\int_{\mathbb{R}_+^2}|p_1||p_2|^2|f(p_2)-g(p_2)||f(p_1)|\mathcal{E}_{p_1}^{i+s}\mathcal{E}_{p_2}^{k+1+j-s}d|{p}_1|d|{p}_2|,
\end{aligned}
$$
where we have used the triangle inequality
$$|f(p_1)f(p_2)-g(p_1)g(p_2)|\le |f(p_1)-g(p_1)||g(p_2)|+|f(p_2)-g(p_2)||f(p_1)|.$$
Since $0<i+s\le n/2+1$ and $0\le k+1+j-s\le n/2+1$, we have
$$\mathcal{E}_{p_1}^{i+s}\leq C\left(|p_1|+|p_1|^{n+2}\right),$$
and
$$\mathcal{E}_{p_2}^{k+1+j-s}\leq C\left(1+|p_2|^{n+2}\right),$$
which yields
$$\begin{aligned}
&~~\int_{\mathbb{R}_+^2}|p_1||p_2|^2|f(p_1)f(p_2)-g(p_1)g(p_2)|\mathcal{E}_{p_1}^{i+s}\mathcal{E}_{p_2}^{k+1+j-s}d|{p}_1|d|{p}_2|\\
\le&~~C\int_{\mathbb{R}_+}|p_1|\left(|p_1|+|p_1|^{n+2}\right)|f(p_1)-g(p_1)|d|{p}_1|\int_{\mathbb{R}_+}|p_1|^2\left(1+|p_1|^{n+2}\right)|g(p_1)|d|{p}_1|\\
 &~~+C\int_{\mathbb{R}_+}|p_1|\left(|p_1|+|p_1|^{n+2}\right)|f(p_1)-g(p_1)|d|{p}_1|\int_{\mathbb{R}_+}|p_1|^2\left(1+|p_1|^{n+2}\right)|f(p_1)|d|{p}_1|\\
\le&~~C\int_{\mathbb{R}^3}\left(1+|p_1|^{n+1}\right)|f(p_1)-g(p_1)|d{p}_1\int_{\mathbb{R}^3}\left(1+|p_1|^{n+2}\right)|g(p_1)|d{p}_1\\
 &~~+C\int_{\mathbb{R}^3}\left(1+|p_1|^{n+1}\right)|f(p_1)-g(p_1)|d{p}_1\int_{\mathbb{R}^3}\left(1+|p_1|^{n+2}\right)|f(p_1)|d{p}_1,
\end{aligned}
$$
where in the last inequality, we have switched the integration on $\mathbb{R}_+$ to $\mathbb{R}^3$, by a spherical change of variables. Now, by the boundedness of $f$ and $g$ in $L^1$ and $L^1_{n+2}$,
$$\begin{aligned}
&~~\int_{\mathbb{R}_+^2}|p_1||p_2|^2|f(p_1)f(p_2)-g(p_1)g(p_2)|\mathcal{E}_{p_1}^{i+s}\mathcal{E}_{p_2}^{k+1+j-s}d|{p}_1|d|{p}_2|\\
\le&~~C\int_{\mathbb{R}^3}\left(1+|p_1|^{n+1}\right)|f(p_1)-g(p_1)|d{p}_1,
\end{aligned}
$$
which implies the following estimate on $X_{1}$
\begin{equation}\label{Propo:HolderC221:E3}\begin{aligned}
X_{1}
\le C\|f-g\|_{L^1}+C\|f-g\|_{L^1_{n+1}}.
\end{aligned}
\end{equation}
We now estimate $X_{2}$. As an application of the inequality
$$\mathcal{E}_{p_2}^{n/2+1}\leq C\left(|p_2|+|p_2|^{n+2}\right),$$
$X_2$ can be bounded as follows
$$\begin{aligned}
X_{2}\le &~~C\int_{\mathbb{R}_+^2}|p_1|^2|p_2| |f(p_1)f(p_2)-g(p_1)g(p_2)|\left(|p_2|+|p_2|^{n+2}\right)d|{p}_1|d|{p}_2|\\
\le &~~C\int_{\mathbb{R}_+^2}|p_1|^2|f(p_1)-g(p_1)||g(p_2)|\left(|p_2|^2+|p_2|^{n+3}\right)d|{p}_1|d|{p}_2|\\
&~~+C\int_{\mathbb{R}_+^2}|p_1|^2 |f(p_2)-g(p_2)||f(p_1)|\left(|p_2|^2+|p_2|^{n+3}\right)d|{p}_1|d|{p}_2|
.
\end{aligned}
$$
The same argument as for \eqref{Propo:HolderC221:E3} leads to
\begin{equation}\label{Propo:HolderC221:E4}\begin{aligned}
X_{2}
\le C\|f-g\|_{L^1}+C\|f-g\|_{L^1_{n+1}}.
\end{aligned}\end{equation}
Combining \eqref{Propo:HolderC221:E2}, \eqref{Propo:HolderC221:E3} and \eqref{Propo:HolderC221:E4} yields
\begin{equation}\label{Propo:HolderC221:E5}X
\le C\left(\|f-g\|_{L^1}+\|f-g\|_{L^1_{n+1}}\right).
\end{equation}
The two inequalities \eqref{Propo:HolderC221:E1} and \eqref{Propo:HolderC221:E5} lead to
\begin{equation}\label{Propo:HolderC221:E6}\begin{aligned}
\int_{\mathbb{R}^3}\left|C_{22}^{1}[f](p_1)-C_{22}^{1}[g](p_1)\right||p_1|^{n}dp_1
\le C\left(\|f-g\|_{L^1}+\|f-g\|_{L^1_{n+1}}\right).
\end{aligned}
\end{equation}
Inequality \eqref{Propo:HolderC221:2} is a consequence of Inequality \eqref{Propo:HolderC221:1}, Lemma \ref{Lemma:Holder} and
$$\|f-g\|_{L^1_{n+1}}\le \|f-g\|_{L^1}^{\frac{1}{n+2}}\left(\|f\|_{L^1_{n+2}}+\|g\|_{L^1_{n+2}}\right)^{\frac{n+1}{n+2}}.$$
\end{proof} 
\subsubsection{Holder estimates for $C_{22}^{2}$}\label{Sec:HolderEstimateC222}
\begin{proposition}\label{Propo:HolderC222N} Let $f$ and $g$ be two functions in $L^1_n(\mathbb{R}^3)\cap L^1(\mathbb{R}^3)$, $n/2\in\mathbb{N}$, $n$ can be $0$, then there exists a constant $\mathcal{C}$ depending on $\|f\|_{L^1_n}, \|f\|_{L^1}, \|g\|_{L^1_n}, \|g\|_{L^1}$, such that 
\begin{equation}\label{Propo:HolderC222N:1} \|C_{22}^{2}[f]-C_{22}^{2}[g]\|_{L^1_{n}}\le \mathcal{C}\left(\|f-g\|_{L^1_n}+\|f-g\|_{L^1}\right).\end{equation}
If $\|f\|_{{L}^1_{n+1}},\|g\|_{{L}^1_{n+1}}<\mathcal{C}_0$, then
\begin{equation}\label{Propo:HolderC222N:2} 
\|C_{22}^2[f]-C_{22}^2[g]\|_{L^1_{n}}\le \mathcal{C}_1\left(\|f-g\|_{L^1}^{\frac{1}{n+1}}+\|f-g\|_{L^1}\right),\end{equation}
where $\mathcal{C}_1$ is a constant depending on $\mathcal{C}_0$, $\mathcal{C}$.
\end{proposition}
\begin{proof}
As a view of Lemma \ref{Lemma:WeakFormulation}, the $L^1_n$ norm of the difference $C_{22}^{2}[f]-C_{22}^{2}[g]$ can be written as
$$\begin{aligned}
\begin{aligned}
&~~\int_{\mathbb{R}^3}\left|C_{22}^{2}[f](p_1)-C_{22}^{2}[g](p_1)\right||p_1|^{n}dp_1\\
\le&~~C\int_{\mathbb{R}_+^{4}}K^{22}(p_1,p_2,p_3,p_4){\min\{|p_1|,|p_2|,|p_3|,|p_4|\}}|p_1||p_2||p_3||p_4|\delta(\mathcal{E}_{p_1}+\mathcal{E}_{p_2}-\mathcal{E}_{p_3}-\mathcal{E}_{p_4}) \times\\
&~~\times|f(p_1)f(p_2)f(p_3)-g(p_1)g(p_2)g(p_3)|\Big[|p_4|^{n}+|p_3|^{n}+|p_2|^{n}+|p_1|^{n}\Big]d|{p}_1|d|{p}_2|d|{p}_3|d|p_4|,
\end{aligned}
\end{aligned}
$$
Similar as for Proposition \ref{Propo:HolderC221}, by the inequality 
$$|p_4|^{n}+|p_3|^{n}+|p_2|^{n}+|p_1|^{n}\leq C\mathcal{E}_{p_4}^{n/2}+C\mathcal{E}_{p_3}^{n/2}+C\mathcal{E}_{p_2}^{n/2}+C\mathcal{E}_{p_1}^{n/2},$$
one has
$$\begin{aligned}
&~~\int_{\mathbb{R}^3}\left|C_{22}^{2}[f](p_1)-C_{22}^{2}[g](p_1)\right||p_1|^{n}dp_1\\
\le&~~C\int_{\mathbb{R}_+^{4}}K^{22}(p_1,p_2,p_3,p_4){\min\{|p_1|,|p_2|,|p_3|,|p_4|\}}|p_1||p_2||p_3||p_4|\delta(\mathcal{E}_{p_1}+\mathcal{E}_{p_2}-\mathcal{E}_{p_3}-\mathcal{E}_{p_4})\times \\
&~~\times|f(p_1)f(p_2)f(p_3)-g(p_1)g(p_2)g(p_3)|\Big[\mathcal{E}_{p_4}^{n/2}+\mathcal{E}_{p_3}^{n/2}+\mathcal{E}_{p_2}^{n/2}+\mathcal{E}_{p_1}^{n/2}\Big]d|{p}_1|d|{p}_2|d|{p}_3|d|p_4|,
\end{aligned}
$$
By the Dirac function $\delta(\mathcal{E}_{p_1}+\mathcal{E}_{p_2}-\mathcal{E}_{p_3}-\mathcal{E}_{p_4})$,  $\mathcal{E}_{p_4}$ can be written as $\mathcal{E}_{p_1}+\mathcal{E}_{p_2}-\mathcal{E}_{p_3}$, which implies
$$\begin{aligned}
\begin{aligned}
&~~\int_{\mathbb{R}^3}\left|C_{22}^{2}[f](p_1)-C_{22}^{2}[g](p_1)\right||p_1|^{n}dp_1\\
\le&~~C\int_{\mathbb{R}_+^{4}}K^{22}(p_1,p_2,p_3,p_4){\min\{|p_1|,|p_2|,|p_3|,|p_4|\}}|p_1||p_2||p_3||p_4|\delta(\mathcal{E}_{p_1}+\mathcal{E}_{p_2}-\mathcal{E}_{p_3}-\mathcal{E}_{p_4})\times \\
&~~\times|f(p_1)f(p_2)f(p_3)-g(p_1)g(p_2)g(p_3)|\Big[\left(\mathcal{E}_{p_1}+\mathcal{E}_{p_2}-\mathcal{E}_{p_3}\right)^{n/2}+\\
&~~+\mathcal{E}_{p_3}^{n/2}+\mathcal{E}_{p_2}^{n/2}+\mathcal{E}_{p_1}^{n/2}\Big]d|{p}_1|d|{p}_2|d|{p}_3|d|p_4|.
\end{aligned}
\end{aligned}
$$
Similar as for \eqref{Propo:C22:E3}, $|p_4|d|p_4|$ can be bounded by $Cd\mathcal{E}_{p_4}$ and $\min\{|p_1|,|p_2|,|p_3|,|p_4|\}$ can be bounded by $\min\{|p_1|,|p_2|,|p_3|\}$, which leads to
$$\begin{aligned}
&~~\int_{\mathbb{R}^3}\left|C_{22}^{2}[f](p_1)-C_{22}^{2}[g](p_1)\right||p_1|^{n}dp_1\\
\le&~~C\int_{\mathbb{R}_+^{4}}K^{22}(p_1,p_2,p_3,p_4){\min\{|p_1|,|p_2|,|p_3|\}}|p_1||p_2||p_3|\delta(\mathcal{E}_{p_1}+\mathcal{E}_{p_2}-\mathcal{E}_{p_3}-\mathcal{E}_{p_4})\times\\
&~~ \times|f(p_1)f(p_2)f(p_3)-g(p_1)g(p_2)g(p_3))|\times\\
&~~\times \Big[\left(\mathcal{E}_{p_1}+\mathcal{E}_{p_2}-\mathcal{E}_{p_3}\right)^{n/2}+\mathcal{E}_{p_3}^{n/2}+\mathcal{E}_{p_2}^{n/2}+\mathcal{E}_{p_1}^{n/2}\Big]d|{p}_1|d|{p}_2|d|{p}_3|d\mathcal{E}_{p_4}\\
\le&~~C\int_{\mathcal{E}_{p_3}\leq \mathcal{E}_{p_1}+\mathcal{E}_{p_2}}K^{22}(p_1,p_2,p_3,p_4){\min\{|p_1|,|p_2|,|p_3|\}}|p_1||p_2||p_3| |f(p_1)f(p_2)f(p_3)-\\
&~~-g(p_1)g(p_2)g(p_3)|\Big[\left(\mathcal{E}_{p_1}+\mathcal{E}_{p_2}-\mathcal{E}_{p_3}\right)^{n/2}+\mathcal{E}_{p_3}^{n/2}+\mathcal{E}_{p_2}^{n/2}+\mathcal{E}_{p_1}^{n/2}\Big]d|{p}_1|d|{p}_2|d|{p}_3|,
\end{aligned}
$$
where we have taken the integration with respect to $d\mathcal{E}_{p_4}$.\\
Since $n/2$ is a natural number, by Newton formula
$$\left(\mathcal{E}_{p_1}+\mathcal{E}_{p_2}-\mathcal{E}_{p_3}\right)^{n/2}+\mathcal{E}_{p_3}^{n/2}+\mathcal{E}_{p_2}^{n/2}+\mathcal{E}_{p_1}^{n/2} = \sum_{0\le i,j,k~;~i+j+k=n/2~} D_{i,j,k,n}\mathcal{E}_{p_1}^i\mathcal{E}_{p_2}^j\mathcal{E}_{p_3}^k,$$
where $D_{i,j,k,n}$ are positive constants. As an application of the above Newton formula, one has
$$\begin{aligned}
&~~\int_{\mathbb{R}^3}\left|C_{22}^{2}[f](p_1)-C_{22}^{2}[g](p_1)\right||p_1|^{n}dp_1\\
\le&~~C\sum_{0\le i,j,k~;~i+j+k=n/2~}\int_{\mathcal{E}_{p_3}\leq \mathcal{E}_{p_1}+\mathcal{E}_{p_2}}K^{22}(p_1,p_2,p_3,p_4){\min\{|p_1|,|p_2|,|p_3|\}}|p_1||p_2||p_3|\\
&~~\times|f(p_1)f(p_2)f(p_3)-g(p_1)g(p_2)g(p_3)|\mathcal{E}_{p_1}^i\mathcal{E}_{p_2}^j\mathcal{E}_{p_3}^kd|{p}_1|d|{p}_2|d|{p}_3|,
\end{aligned}
$$
where $C$ is a positive constant varying from line to line.\\
By using the fact that
$$K^{22}(p_1,p_2,p_3,p_4){\min\{|p_1|,|p_2|,|p_3|\}}|p_1||p_2||p_3| \le C|p_1|^2|p_2|^2|p_3|^2,$$
where $C$ is a positive constant depending on $p_*$ defined in \eqref{Def:TransitionProbabilityKernel:K22A}, we get
$$\begin{aligned}
&~~\int_{\mathbb{R}^3}\left|C_{22}^{2}[f](p_1)-C_{22}^{2}[g](p_1)\right||p_1|^{n}dp_1\\
\le&~~C\sum_{0\le i,j,k~;~i+j+k=n/2~}\int_{\mathcal{E}_{p_3}\leq \mathcal{E}_{p_1}+\mathcal{E}_{p_2}}|p_1|^2|p_2|^2|p_3|^2 |f(p_1)f(p_2)f(p_3)-\\
&~~-g(p_1)g(p_2)g(p_3)|\mathcal{E}_{p_1}^i\mathcal{E}_{p_2}^j\mathcal{E}_{p_3}^kd|{p}_1|d|{p}_2|d|{p}_3|\\
\le&~~C\sum_{0\le i,j,k~;~i+j+k=n/2~}\int_{\mathbb{R}_+^3}|p_1|^2|p_2|^2|p_3|^2 |f(p_1)f(p_2)f(p_3)-\\
&~~-g(p_1)g(p_2)g(p_3)|\mathcal{E}_{p_1}^i\mathcal{E}_{p_2}^j\mathcal{E}_{p_3}^kd|{p}_1|d|{p}_2|d|{p}_3|.
\end{aligned}
$$
Changing from the radial integration on $\mathbb{R}_+$ to the integration on $\mathbb{R}^3$ in the above inequality, by a spherical coordinate change of variables, yields
 $$\begin{aligned}
&~~\int_{\mathbb{R}^3}\left|C_{22}^{2}[f](p_1)-C_{22}^{2}[g](p_1)\right||p_1|^{n}dp_1\\
\le&~~C\sum_{0\le i,j,k~;~i+j+k=n/2~}\int_{\mathbb{R}^{3\times3}}|f(p_1)f(p_2)f(p_3)-g(p_1)g(p_2)g(p_3)|\mathcal{E}_{p_1}^i\mathcal{E}_{p_2}^j\mathcal{E}_{p_3}^kd{p}_1d{p}_2d{p}_3.
\end{aligned}
$$
Applying the triangle inequality 
\begin{eqnarray*}
&&|f(p_1)f(p_2)f(p_3)-g(p_1)g(p_2)g(p_3)|\\
&\le &|f(p_1)-g(p_1)||f(p_2)||f(p_3)|+|f(p_2)-g(p_2)||g(p_1)||f(p_3)|+|f(p_3)-g(p_3)||g(p_1)||g(p_3)|,
\end{eqnarray*}
to the previous inequality gives
 $$\begin{aligned}
&~~\int_{\mathbb{R}^3}\left|C_{22}^{2}[f](p_1)-C_{22}^{2}[g](p_1)\right||p_1|^{n}dp_1\\
\le&~~C\sum_{0\le i,j,k~;~i+j+k=n/2~}\int_{\mathbb{R}^{3\times3}}|f(p_1)-g(p_1)||f(p_2)||f(p_3)|\mathcal{E}_{p_1}^i\mathcal{E}_{p_2}^j\mathcal{E}_{p_3}^kd{p}_1d{p}_2d{p}_3\\
&~~+C\sum_{0\le i,j,k~;~i+j+k=n/2~}\int_{\mathbb{R}^{3\times3}}|f(p_2)-g(p_2)||g(p_1)||f(p_3)|\mathcal{E}_{p_1}^i\mathcal{E}_{p_2}^j\mathcal{E}_{p_3}^kd{p}_1d{p}_2d{p}_3\\
&~~+C\sum_{0\le i,j,k~;~i+j+k=n/2~}\int_{\mathbb{R}^{3\times3}}|f(p_3)-g(p_3)||g(p_1)||g(p_3)|\mathcal{E}_{p_1}^i\mathcal{E}_{p_2}^j\mathcal{E}_{p_3}^kd{p}_1d{p}_2d{p}_3
.
\end{aligned}
$$
Notice that we can estimate $\mathcal{E}_{p_1}^i$, $\mathcal{E}_{p_2}^j$ and $\mathcal{E}_{p_3}^k$ as
\begin{eqnarray*}
\mathcal{E}_{p}^i\le C(1+|p|^n),~~~
\mathcal{E}_{p}^j\le C(1+|p|^n),~~~
\mathcal{E}_{p}^k\le C(1+|p|^n),
\end{eqnarray*}
which leads to the following estimate on the norm of $C_{22}^{2}[f]-C_{22}^{2}[g]$
 $$\begin{aligned}
&~~\int_{\mathbb{R}^3}\left|C_{22}^{2}[f](p_1)-C_{22}^{2}[g](p_1)\right||p_1|^{n}dp_1\\
\le&~~C\sum_{0\le i,j,k~;~i+j+k=n/2~}\int_{\mathbb{R}^{3\times3}}|f(p_1)-g(p_1)||f(p_2)||f(p_3)|\\
&~~\times (1+|p_1|^n)(1+|p_2|^n)(1+|p_3|^n)d{p}_1d{p}_2d{p}_3\\
&~~+C\sum_{0\le i,j,k~;~i+j+k=n/2~}\int_{\mathbb{R}^{3\times3}}|f(p_2)-g(p_2)||g(p_1)||f(p_3)|\\
&~~\times(1+|p_1|^n)(1+|p_2|^n)(1+|p_3|^n)d{p}_1d{p}_2d{p}_3\\
&~~+C\sum_{0\le i,j,k~;~i+j+k=n/2~}\int_{\mathbb{R}^{3\times3}}|f(p_3)-g(p_3)||g(p_1)||g(p_3)|\\
&~~\times(1+|p_1|^n)(1+|p_2|^n)(1+|p_3|^n)d{p}_1d{p}_2d{p}_3
.
\end{aligned}
$$
Now, since  
$$\int_{\mathbb{R}^3}|f(p)|(1+|p|^n)=\|f\|_{L^1}+\|f\|_{L^1_n},~~~\int_{\mathbb{R}^3}|g(p)|(1+|p|^n)=\|g\|_{L^1}+\|g\|_{L^1_n},$$
we get from the above inequality that
 $$\begin{aligned}
~\int_{\mathbb{R}^3}\left|C_{22}^{2}[f](p_1)-C_{22}^{2}[g](p_1)\right||p_1|^{n}dp_1
\le C\left(\|f-g\|_{L^1}+\|f-g\|_{L^1_n}\right)
.
\end{aligned}
$$
Inequality \eqref{Propo:HolderC222N:2} is a consequence of Inequality \eqref{Propo:HolderC222N:1}, Lemma \ref{Lemma:Holder} and
$$\|f-g\|_{L^1_{n}}\le \|f-g\|_{L^1}^{\frac{1}{n+1}}\left(\|f\|_{L^1_{n+1}}+\|g\|_{L^1_{n+1}}\right)^{\frac{n}{n+1}}.$$
\end{proof}
\subsection{Proof of Theorem \ref{Theorem:ExistenceKinetic}}\label{Sec:ExistenceKinetic}
In order to prove Theorem \ref{Theorem:ExistenceKinetic}, we will use Theorem \ref{Theorem:ODE}. Choose $E=\mathbb{L}^{1}_{2n}\big(\mathbb{R}^{3}\big)$. 
We define the function $|\cdot|_*$ to be
$$|f|_*=\int_{\mathbb{R}^3}f(p)dp.$$
Set
$$\|f\|_*=\int_{\mathbb{R}^3}|f(p)|dp.$$
By \eqref{Propo:Mass:E5}, it is clear that for all $f\geq0$, $f\in E$, the following inequality holds true
\begin{equation}\label{CStar}
|Q[f]|_*\le C^*\left(1+\|f\|_*\right),
\end{equation}
where $C^*$ depends on $\|f\|_{\mathcal{L}^1_2(\mathbb{R}^3)}$. We then choose $C_*$ in Theorem \ref{Theorem:ODE} as $C^*$.\\ 
The set $\mathcal{S}_T$ is defined as follows: 
\begin{align}\label{SetS} 
\begin{split}
\mathcal{S}_T:=\Big\{& f \in \mathbb{L}^{1}_{2n}\big(\mathbb{R}^{3}\big) \;\big| \; (S_1)\;f \geq 0, ~~~ f(p)=f(|p|),\; (S_2) \int_{\mathbb{R}^3}\,f(|p|) dp\le \mathfrak{c}_0,\;\\
& (S_3) \int_{\mathbb{R}_+}\,f(|p|)\mathcal{E}_p dp= \mathfrak{c}_1,\; (S_4)\int_{\mathbb{R}_+}\,f(|p|)\mathcal{E}_p^{n^*}dp\leq\mathfrak{c}_{n^*}\Big\},
\end{split}
\end{align}
where 
\begin{equation}\label{CONDITION:INITIALMASS}
\mathfrak{c}_0:=(2\mathcal{R}+1)e^{(C^*+1)T},
\end{equation}
 and 
\begin{equation}\label{CONDITION:INITIALMOMENT}
\mathfrak{c}_{n^*}=\frac{3\rho_{n_*}}{2},
\end{equation}
with $\rho_{n_*}$ defined in \eqref{RHON}.  It is clear that $\mathcal{S}_T$ is a bounded, convex and closed subset of $\mathbb{L}^{1}_{2n}\big(\mathbb{R}^{3}\big)$. Moreover for all $f$ in $\mathcal{S}_T$, it is straightforward that $|f|_*=\|f\|_*$.\\
In the four Sections \ref{Subtangent}, \ref{Con2}, \ref{Holder}, \ref{Lipschitz}, we will verify  the four conditions $(\mathfrak{A})$, $(\mathfrak{B})$, $(\mathfrak{C})$ and $(\mathfrak{D})$ of Theorem \ref{Theorem:ODE}. Then, Theorem \ref{Theorem:ExistenceKinetic} follows as an application of Theorem \ref{Theorem:ODE}. 
\subsubsection{Checking Condition $(\mathfrak{A})$}\label{Con2}
We choose the constant $R_*$ to be $\mathcal{R}+1$, then for all $u$ in $\mathcal{S}_T$,  $\|u\|_*\le (2{R}_*+1)e^{(C^*+1)T}$. Condition  $(\mathfrak{A})$ is satisfied.

\subsubsection{Checking Condition $(\mathfrak{B})$}\label{Subtangent}
First, the same argument as for \eqref{Propo:MomentsPropa:E4} gives
\begin{equation}\label{Subtangent:E1}\begin{aligned}
&~~\int_{\mathbb{R}^3}Q[f]\mathcal{E}_p^{n^*} dp \leq \mathcal{P}\big[m_{n^*}(f)\big]:=\\
&~~{C} m_{n^*}(f)+Cm_{n^*}(f)^{\frac{{n^*}-1}{{n^*}}}+Cm_{n^*}(f)^{\frac{{n^*}-2}{{n^*}}}+ Cm_{n^*}(f)^{\frac{1}{{n^*}}}-Cm_{{n^*}}(f)^{\frac{{n^*}+1}{{n^*}}}\\
 &~~+C\sum_{0\le i, j, k<{n^*};~ i+j+k={n^*}}\sum_{s=0}^{k+1}m_{n^*}(f)^{\frac{i+s}{{n^*}}}\left(m_{n^*}(f)^{\frac{j+k-s}{{n^*}}}+m_{n^*}(f)^{\frac{j+k-s+1/2}{{n^*}}}\right)+\\
&~~+C\sum_{0\le i, j, k<{n^*};~ i+j+k={n^*};~j,k>0}m_{n^*}(f)^{\frac{i}{{n^*}}}\left(m_{n^*}(f)^{\frac{j-1}{{n^*}}}+m_{n^*}(f)^{\frac{j-1/2}{{n^*}}}\right)\times\\
&~~\times\left(m_{n^*}(f)^{\frac{k-1}{{n^*}}}+m_{n^*}(f)^{\frac{k-1/2}{{n^*}}}\right),~~~\forall f\in\mathcal{S}_T,
\end{aligned}
\end{equation}
 where $C$ is a positive constant  depending on $\mathfrak{c}_0$.

\begin{equation}\label{RHON}
\mbox{Let $\rho_{n^*}$ be the solution of $\mathcal{P}(\rho)=0$: if $0<\rho<\rho_{n^*}$, $\mathcal{P}(\rho)<0$; if $\rho>\rho_{n^*}$, $\mathcal{P}(\rho)>0$.}
\end{equation} 
Notice that $\rho_{n^*}$ depends on $\mathfrak{c}_0$.
\\ Let $f$ be an arbitrary element of the set $\mathcal{S}_T\cap B_*\Big(O,(2R_*+1)e^{(C_*+1)T}\Big)$ and consider the element $f+hQ[f]$. We will show that for all $\epsilon>0$, there exists $h_*$ depending on $f$ and $\epsilon$ such that $B(f+hQ[f],h\epsilon)\cap\mathcal{S}_T$ is not empty for all $0<h<h_*$. Define $\chi_R(p)$ to be the characteristic function of the ball $B(O,R)$ centered at the origin with radius $R$. Set $f_R(p)=\chi_R(p)f(p)$ and $w_R=f+hQ[f_R]$.  Since $Q[f_R]\in \mathbb{L}^1_{2n}(\mathbb{R}^3)$, we find that $w_R\in \mathbb{L}^1_{2n}(\mathbb{R}^3)$. We will prove that for $h_*$ small enough and $R$ large enough, $w_R$ belongs to $\mathcal{S}_T$. We now verify the four conditions $(S_1)$, $(S_2)$, $(S_3)$ and $(S_4)$.
\begin{itemize}
\item {\it Condition $(S_1)$:} Since $f_R$ is compactly supported, it is clear that $Q^{-}[f_R]$, with $Q^-$ defined in \eqref{QGainLoss}, is bounded by $C(f,R,\mathfrak{c}_0,\mathfrak{c}_{n^*})$, a positive constant depending on $f$, $R$, $\mathfrak{c}_0$, $\mathfrak{c}_{n^*}$, which implies
$$w_R\geq f-h f_RQ^{-}[f_R]\geq f(1-hQ^{-}[f_R])\geq 0,$$
for $h< C(f,R,\mathfrak{c}_0,\mathfrak{c}_{n^*})^{-1}$.
\item {\it Condition $(S_2)$:} Since $$\|f\|_*<(2R_*+1)e^{(C_*+1)T},$$
and 
$$\lim_{h\to 0}\|f-w_R\|_*=0,$$
we can choose $h_*$ small enough such that 
$$\|w_R\|_*<(2R_*+1)e^{(C_*+1)T}.$$
\item {\it Condition $(S_3)$:} By the conservation of energy, we have
$$\int_{\mathbb{R}^3}w_R\mathcal{E}_pdp=\int_{\mathbb{R}^3}(f+hQ[f_R])\mathcal{E}_pdp=\int_{\mathbb{R}^3}f\mathcal{E}_pdp=\mathfrak{c}_1.$$
\item {\it Condition $(S_4)$:}
Now, we claim that $R$ and $h_*$ can be chosen, such that
$$\int_{\mathbb{R}^3}w_R\mathcal{E}_p^{n^*} dp<\frac{3\rho_{n^*}}{2}.$$
In order to see this, we consider two cases:
\\ If 
$$\int_{\mathbb{R}^3}f\mathcal{E}_p^{n^*} dp<\frac{3\rho_{n^*}}{2},$$
we deduce from the fact  
$$\lim_{h\to 0}\int_{\mathbb{R}^3}|w_R-f|\mathcal{E}_p^{n^*}dp=0,$$
that we can choose $h_*$ small enough such that
$$\int_{\mathbb{R}^3}w_R\mathcal{E}_p^{n^*} dp<\frac{3\rho_{n^*}}{2}.$$
\\  If, on the other hand, we have
$$\int_{\mathbb{R}^3}f\mathcal{E}_p^{n^*} dp=\frac{3\rho_{n^*}}{2},$$
we can choose $R$ large enough such that
$$\int_{\mathbb{R}^3}f_R\mathcal{E}_p^{n^*} dp>{\rho_{n^*}},$$
which implies, by \eqref{RHON}, that
$$\int_{\mathbb{R}^3}Q[f_R]<0.$$
As a consequence, 
$$\int_{\mathbb{R}^3}w_R\mathcal{E}^{n^*}_pdp<\int_{\mathbb{R}^3}f\mathcal{E}^{n^*}_pdp=\frac{3\rho_{n^*}}{2}.$$
\end{itemize}
Finally, we have $w_R\in\mathcal{S}_T$ for all $0<h<h_*$.
\\ Now since 
$$\lim_{R\to\infty}\frac{1}{h}\|w_R-f-hQ[f_R]\|_{\mathbb{L}^1_{2n}(\mathbb{R}^3)}=\lim_{R\to\infty}\|Q[f]-Q[f_R]\|_{\mathbb{L}^1_{2n}(\mathbb{R}^3)}=0,$$
then for $R$ large enough, $w_R\in B(f+hQ[f],h\epsilon)$, which implies $B(f+hQ[f],h\epsilon)\cap \mathcal{S}_T\backslash\{f+hQ[f]\}.$ Condition $(\mathfrak{B})$ is verified.  

\subsubsection{Checking Condition $(\mathfrak{C})$}\label{Holder}
Condition $(\mathfrak{C})$ follows from Propositions \ref{Propo:HolderC12}, \ref{Propo:HolderC221},  and \ref{Propo:HolderC222N}.
\subsubsection{Checking Condition $(\mathfrak{D})$}\label{Lipschitz}
By the Lebesgue dominated convergence theorem, we have that
\begin{equation}\label{Lipschitz:E1}
\Big[\varphi,\phi\Big]\le \int_{\mathbb{R}^3}\varphi(p)\mathrm{sign}(\phi(p))(1+\mathcal{E}_p^n)dp,\end{equation}
which means that Condition $(\mathfrak{D})$ is satisfied if we have the following inequality
\begin{equation}\label{Lipschitz:E2}
\mathcal{M}_0:=\int_{\mathbb{R}^3}[Q[f](p)-Q[g](p)]\mathrm{sign}((f-g)(p))(1+\mathcal{E}_p^n)dp\leq C\|f-g\|_{\mathbb{L}^1_{2n}}.
\end{equation}
Since $Q=C_{12}+C_{22}$, let us split $$\mathcal{M}_0=\mathcal{M}_1+\mathcal{M}_2,$$
where 
$$\mathcal{M}_1:=\int_{\mathbb{R}^3}[C_{12}[f](p)-C_{12}[g](p)]\mathrm{sign}((f-g)(p))(1+\mathcal{E}_p^n)dp,$$
and
$$\mathcal{M}_2:=\int_{\mathbb{R}^3}[C_{22}[f](p)-C_{22}[g](p)]\mathrm{sign}((f-g)(p))(1+\mathcal{E}_p^n)dp.$$
{\bf Step 1: Estimating $\mathcal{M}_1$.}
\\ Define $\varphi_k(p)=\mathrm{sign}((f-g)(p))\mathcal{E}^k_p$, $k\in\mathbb{Z}, k\geq 0$, $k\ne 1$. Let us consider the following generalized term of $\mathcal{M}_1$
\begin{equation}\label{Lipschitz:E3}
\mathcal{N}_0:=\int_{\mathbb{R}^3}[C_{12}[f](p)-C_{12}[g](p)]\varphi_k(p)dp,
\end{equation}
which by Lemma \ref{Lemma:WeakFormulation} can be rewritten as
\begin{equation}\label{Lipschitz:E3a}
\begin{aligned}
\mathcal{N}_0:=&~~~\int_{\mathbb{R}^{3\times 3}}[R_{12}[f](p_1)-R_{12}[g](p_1)][\varphi_k(p_1)-\varphi_k(p_2)-\varphi_k(p_3)]dp_1dp_2dp_3\\
=&~~~\int_{\mathbb{R}^{3\times 2}}\bar{K}^{12}(p_2+p_3,p_2,p_3)\delta({\mathcal{E}_{p_2+p_3}-\mathcal{E}_{p_2}-\mathcal{E}_{p_3}})[(f(p_2)f(p_3)-g(p_2)g(p_3))\\
&~~~-2(f(p_2)f(p_2+p_3)-g(p_2)g(p_2+p_3))-(f(p_2+p_3)-g(p_2+p_3))]\times\\
&~~~\times[\varphi_k(p_2+p_3)-\varphi_k(p_2)-\varphi_k(p_3)]dp_2dp_3.
\end{aligned}
\end{equation}
Split $\mathcal{N}_0$ into the sum of three terms:
\begin{equation}\label{Lipschitz:E4}
\begin{aligned}
\mathcal{N}_1:=&~~~\int_{\mathbb{R}^{3\times 2}}\bar{K}^{12}(p_2+p_3,p_2,p_3)\delta({\mathcal{E}_{p_2+p_3}-\mathcal{E}_{p_2}-\mathcal{E}_{p_3}})[f(p_2)f(p_3)-g(p_2)g(p_3)]\\
&~~~\times[\varphi_k(p_2+p_3)-\varphi_k(p_2)-\varphi_k(p_3)]dp_2dp_3,
\end{aligned}
\end{equation}
\begin{equation}\label{Lipschitz:E5}
\begin{aligned}
\mathcal{N}_2:=&~~~-2\int_{\mathbb{R}^{3\times 2}}\bar{K}^{12}(p_2+p_3,p_2,p_3)\delta({\mathcal{E}_{p_2+p_3}-\mathcal{E}_{p_2}-\mathcal{E}_{p_3}})[f(p_2)f(p_2+p_3)-g(p_2)g(p_2+p_3)]\\
&~~~\times[\varphi_k(p_2+p_3)-\varphi_k(p_2)-\varphi_k(p_3)]dp_2dp_3,
\end{aligned}
\end{equation}
and
\begin{equation}\label{Lipschitz:E6}
\begin{aligned}
\mathcal{N}_3:=
&~~~-\int_{\mathbb{R}^{3\times 2}}\bar{K}^{12}(p_2+p_3,p_2,p_3)\delta({\mathcal{E}_{p_2+p_3}-\mathcal{E}_{p_2}-\mathcal{E}_{p_3}})[f(p_2+p_3)-g(p_2+p_3)]\times\\
&~~~\times[\varphi_k(p_2+p_3)-\varphi_k(p_2)-\varphi_k(p_3)]dp_2dp_3.
\end{aligned}
\end{equation}
The same arguments as for \eqref{Propo:HolderC12:E2} and \eqref{Propo:HolderC12:E2b} give
\begin{equation}\label{Lipschitz:E7}
\begin{aligned}
\mathcal{N}_1\leq C\|f-g\|_{\mathbb{L}^1_{2k}(\mathbb{R}^3)},
\end{aligned}
\end{equation}
and
\begin{equation}\label{Lipschitz:E8}
\begin{aligned}
\mathcal{N}_2\leq C\|f-g\|_{\mathbb{L}^1_{2k}(\mathbb{R}^3)},
\end{aligned}
\end{equation}
where $C$ is a positive constant varying from line to line.
\\ The third term $\mathcal{N}_3$ can be estimated as
\begin{equation}\label{Lipschitz:E9}
\begin{aligned}
\mathcal{N}_3=
&~~~-\int_{\mathbb{R}^{3\times 2}}\bar{K}^{12}(p_2+p_3,p_2,p_3)\delta({\mathcal{E}_{p_2+p_3}-\mathcal{E}_{p_2}-\mathcal{E}_{p_3}})[f(p_2+p_3)-g(p_2+p_3)]\times\\
&~~~\times [\mathcal{E}_{p_2+p_3}^k\mathrm{sign}((f(p_2+p_3)-g(p_2+p_3))-\mathcal{E}_{p_2}^k\mathrm{sign}((f(p_2)-g(p_2))-\\
&~~~-\mathcal{E}_{p_3}^k\mathrm{sign}((f(p_3)-g(p_3))]dp_2dp_3\\
\le
&~~~\int_{\mathbb{R}^{3\times 2}}\bar{K}^{12}(p_2+p_3,p_2,p_3)\delta({\mathcal{E}_{p_2+p_3}-\mathcal{E}_{p_2}-\mathcal{E}_{p_3}})|f(p_2+p_3)-g(p_2+p_3)|\times\\
&~~~\times[\mathcal{E}_{p_2}^k+\mathcal{E}_{p_3}^k-\mathcal{E}_{p_2+p_3}^k]dp_2dp_3.
\end{aligned}
\end{equation}
Now, let us consider the two cases $k=0$ and $k>1$ separately.
\begin{itemize}
\item If $k=0$,
\begin{equation}\label{Lipschitz:E10}
\begin{aligned}
\mathcal{N}_3
\le
&~~~\int_{\mathbb{R}^{3\times 2}}\bar{K}^{12}(p_2+p_3,p_2,p_3)\delta({\mathcal{E}_{p_2+p_3}-\mathcal{E}_{p_2}-\mathcal{E}_{p_3}})|f(p_2+p_3)-g(p_2+p_3)|dp_2dp_3,
\end{aligned}
\end{equation}
which, by the same arguments that lead to \eqref{Propo:HolderC12:E3}, can be bounded as
\begin{equation}\label{Lipschitz:E11}
\begin{aligned}
\mathcal{N}_3
\le
C\|f-g\|_{\mathbb{L}^1_{3}(\mathbb{R}^3)}.
\end{aligned}
\end{equation}
\item 
If $k>1$, since $\mathcal{E}_{p_2+p_3}=\mathcal{E}_{p_2}+\mathcal{E}_{p_3}$, it is straight forward that
$$\mathcal{E}_{p_2}^k+\mathcal{E}_{p_3}^k-\mathcal{E}_{p_2+p_3}^k=\mathcal{E}_{p_2}^k+\mathcal{E}_{p_3}^k-(\mathcal{E}_{p_2}+\mathcal{E}_{p_3})^k\le - k\mathcal{E}_{p_2}\mathcal{E}_{p_3}^{k-1}\le 0.$$
As a consequence, we can estimate $\mathcal{N}_3$ as
\begin{equation}\label{Lipschitz:E12}
\begin{aligned}
&\mathcal{N}_3\le\\
\le
&~-\int_{\mathbb{R}^{3\times 2}}\bar{K}^{12}(p_2+p_3,p_2,p_3)\delta({\mathcal{E}_{p_2+p_3}-\mathcal{E}_{p_2}-\mathcal{E}_{p_3}})|f(p_2+p_3)-g(p_2+p_3)|k\mathcal{E}_{p_2}\mathcal{E}_{p_3}^{k-1}dp_2dp_3\\
\le
&~-\int_{\mathbb{R}^{3}}\int_{S_{p_1}^0}\bar{K}_0^{12}(p_1,p_2,p_1-p_2)|f(p_1)-g(p_1)|k\mathcal{E}_{p_2}\mathcal{E}_{p_1-p_2}^{k-1}d\sigma(p_2)dp_1.
\end{aligned}
\end{equation}
As a view of Lemma \ref{lem-Sp}, we find the following bound on $\mathcal{N}_3$
\begin{equation}\label{Lipschitz:E13}
\begin{aligned}
\mathcal{N}_3
\le
&~-C\int_{\mathbb{R}^{3}}|f(p)-g(p)|\Big(|p|^{2k+1}\min\{1,|p|\}^{2k+7}\Big)dp.
\end{aligned}
\end{equation}
\end{itemize}
Combining \eqref{Lipschitz:E7}, \eqref{Lipschitz:E8}, \eqref{Lipschitz:E11} and \eqref{Lipschitz:E13} for the two cases $k=0$ and $k=n$, yields
\begin{equation}\label{Lipschitz:E14}\begin{aligned}
\mathcal{M}_1
\le&~~ C\int_{\mathbb{R}^3}|f(p)-g(p)|\Big(1+|p|+|p|^3+|p|^{2n}+|p|^{2n+1}\\
&~~-|p|^{2n+1}\min\{1,|p|\}^{2n+7}\Big)dp.
\end{aligned}
\end{equation} 
{\bf Step 2: Estimating $\mathcal{M}_2$.}
\\ We can estimate $\mathcal{M}_2$ in a straightforward manner by employing Propositions \ref{Propo:HolderC221} and \ref{Propo:HolderC222N}, as follows
\begin{equation}\label{Lipschitz:E15}\begin{aligned}
\mathcal{M}_2
\le&~~ C\int_{\mathbb{R}^3}|f(p)-g(p)|\Big(1+|p|+|p|^{2n}+|p|^{2n+1}\Big)dp.
\end{aligned}
\end{equation} 
{\bf Step 3: Estimating $\mathcal{M}_0$.}
\\ Combining \eqref{Lipschitz:E14} and \eqref{Lipschitz:E15} yields
\begin{equation}\label{Lipschitz:E16}\begin{aligned}
\mathcal{M}_0
\le&~~ C\int_{\mathbb{R}^3}|f(p)-g(p)|\Big(1+|p|+|p|^3+|p|^{2n}\\
&~~-|p|^{2n+1}\min\{1,|p|\}^{2n+7}\Big)dp.
\end{aligned}
\end{equation} 
Since for $|p|\leq1$, 
$$1+|p|+|p|^3+|p|^{2n}-|p|^{4n+8}\le 5,$$
and for $|p|>1$, there exists $C>0$ independent of $p$ such that
$$1+|p|+|p|^3+|p|^{2n}-|p|^{2n+1}\le C,$$
we find that the weight 
$$1+|p|+|p|^3+|p|^{2n}-|p|^{2n+1}\min\{1,|p|\}^{2n+7}$$
of \eqref{Lipschitz:E16} is bounded uniformly in $p$ by a universal positive constant $C$. As a consequence,
Inequality \eqref{Lipschitz:E16} implies
\begin{equation}\label{Lipschitz:E17}\begin{aligned}
\mathcal{M}_0
\le&~~ C\int_{\mathbb{R}^3}|f(p)-g(p)|dp,
\end{aligned}
\end{equation} 
which concludes the proof of \eqref{Lipschitz:E2}. 
 \section{Proof of Theorem \ref{Theorem:ODE}}\label{Appendix}
Our proof is a extension and generalization of the  framework proposed in \cite{Bressan}. The proof is divided into four parts:
{\\\\\bf Part 1:}
Fix a element $v$ of $\mathcal{S}_T$, due to the Holder continuity property of $Q$, we have $$\|Q(u)\|\le \|Q(v)\|+ C\|u-v\|^\beta,~~~~~\forall u\in\mathcal{S}_T.$$ 
According to our assumption, $\mathcal{S}_T$ is bounded by a constant $C_S$. We deduce from the above inequality that 
$$\|Q(u)\|\le \|Q(v)\|+ C\left(\|u\|+\|v\|\right)^\beta\le \|Q(v)\|+ C\left(C_S+\|v\|\right)^\beta=:C_Q, \ \forall u \in \mathcal{S}_T.$$ 
For an element $u$ be in $\mathcal{S}_T$, there exists $\xi_u>0$ such that for $0<\xi<\xi_u$, $u+\xi Q(u)\in{\mathcal{S}_T}$, which implies $$B(u+\xi Q(u),\delta)\cap {\mathcal{S}_T}\backslash\{u+\xi Q(u)\}\ne {\O},$$ for $\delta$ small enough. Choose $\epsilon=2C((C_Q+1)\xi)^\beta$, then $\|Q(u)-Q(v)\|\leq \frac{\epsilon}{2}$ if $\|u-v\|\leq (C_Q+1)\xi$, by the Holder continuity of $Q$. Let $z\in B\left(u+\xi Q(u),\frac{\epsilon \xi}{2}\right)\cap {\mathcal{S}_T}\backslash\{u+\xi Q(u)\} $ and define
$$t\mapsto \vartheta(t)=u+\frac{t(z-u)}{\xi},~~~~t\in[0,\xi].$$
Since  $\mathcal{S}_T$ is convex, $\vartheta$ maps $[0,\xi]$ into $\mathcal{S}_T$. 
It is straightforward that $$\|\vartheta(t)-u\|\leq \xi\|Q(u)\|+\frac{\epsilon \xi}{2}<(C_Q+1)\xi,$$
which implies $$\|Q({\vartheta}(t))-Q(u)\|\leq \frac{\epsilon}{2},~~\forall t\in[0,\xi].$$
The above inequality and the fact that $$\|\dot{\vartheta}(t)-Q(u)\|\leq \frac{\epsilon}{2},$$
leads to 
\begin{equation}\label{Theorem:ODE:E1}
\|\dot{\vartheta}(t)-Q({\vartheta}(t))\|\leq \epsilon,~~\forall t\in[0,\xi].
\end{equation}*
{\\\bf Part 2:} Let $\vartheta$ be a solution to \eqref{Theorem:ODE:E1} on $[0,\tau]$  constructed in Part 1. From Part 1, we have that
$$\begin{aligned}
\|\vartheta(t)\|_* =|\vartheta(t)|_* = \left|u+\frac{t(z-u)}{\tau}\right|_* & \le  |u|_*+\left|\frac{t(z-u)}{\tau}\right|_* \le  |u|_* + |u|_*\frac{tC_*}{2} \\
&= \|\vartheta(0)\|_*\left(1+\frac{tC_*}{2}\right).
\end{aligned}$$
We then obtain
\begin{equation}\label{Theorem:ODE:E3}
\|\vartheta(t)\|_*\le (\|\vartheta(0)\|_*+1)e^{C_*t}-1 <
 (2R_*+1)e^{C_*t}.\end{equation}

Using the procedure of Part 1, we assume that $\vartheta$ can be extended to the interval $[\tau,\tau+\tau']$.
 \\ The same arguments that lead to \eqref{Theorem:ODE:E3} imply
  $$\|\vartheta(\tau+t)\|_*\le \left( (\|\vartheta(\tau)\|_*+1)e^{C_*t}-1\right), \ \  \ t\in[0,\tau'].$$
Combining the above inequality with \eqref{Theorem:ODE:E3} yields
\begin{equation}\label{Theorem:ODE:E5}
\begin{aligned}
\|\vartheta(\tau+t)\|_*\
\le&~~\left(\left(\|\vartheta(0)\|_*+1\right)e^{C_*\tau}-1+1\right)e^{C_*t}-1\\
\le&~~\left(\|\vartheta(0)\|_*+1\right)e^{C_*(\tau+t)}-1\\
<&~~ (2R_*+1)e^{C_*(\tau+t)},
\end{aligned}
\end{equation}
where the last inequality follows from the fact that $R_*\ge 1$.
\\\\\bf Part 3:} From Part 1, there exists a solution $\vartheta$ to the equation \eqref{Theorem:ODE:E1} on an interval $[0,h]$. Now, we have the following procedure.
\begin{itemize}
\item {\it Step 1:}  Suppose that  we can construct the solution $\vartheta$ of \eqref{Theorem:ODE:E1} on $[0,\tau]$ $(\tau<T)$. Since $\vartheta(\tau)\in\mathcal{S}_T$, by the same process as in Part 1 and by \eqref{Theorem:ODE:E3} and \eqref{Theorem:ODE:E5}, the solution $\vartheta$ could be extended to  $[\tau,\tau+h_\tau]$ where $\tau+h_\tau\le T, h_\tau\le \tau$. 
\item {\it Step 2:} Suppose that we can construct the solution $\vartheta$ of \eqref{Theorem:ODE:E1}  on a series of intervals $[0,\tau_1]$, $[\tau_1,\tau_2]$, $\cdots$, $[\tau_n,\tau_{n+1}]$, $\cdots$. Observe that the increasing sequence $\{\tau_n\}$ is bounded by $T$, the sequence has a limit, defined by $\tau.$
Recall that $Q({\vartheta})$ is bounded by $C_Q$ on $[\tau_n,\tau_{n+1}]$ for all $n\in\mathbb{N},$ then $\dot{\vartheta}$ is bounded by $\epsilon+C_Q$ on $[0,\tau)$. As a consequence  $\vartheta(\tau)$ can be defined as $$\vartheta(\tau)=\lim_{n\to\infty}\vartheta(\tau_n), \dot{\vartheta}(\tau)=\lim_{n\to\infty}\dot{\vartheta}(\tau_n),$$
which, together with the fact that $\mathcal{S}_T$ is closed, implies that $\vartheta$ is a solution of \eqref{Theorem:ODE:E1} on $[0,\tau]$. 
\end{itemize}
By Step 2, if
 the solution $\vartheta$ can be defined on $[0,T_0)$, $T_0<T$, it could be extended to $[0,T_0]$. Now, we suppose that $[0,T_0]$ is the maximal  closed interval that $\vartheta$ could be defined, by Step 1, $\vartheta$ could be extended to a larger interval $[T_0,T_0+T_h]$, which means that $T=T_0$ and $\vartheta$ is defined on the whole interval $[0,T]$.
{\\\\\bf Part 4:} Finally, let us consider a sequence of solution $\{u^\epsilon\}$ to \eqref{Theorem:ODE:E1} on $[0,T]$. We will prove that this is a Cauchy sequence. 
Let $\{u^\epsilon\}$ and $\{v^\epsilon\}$ be two sequences of solutions to \eqref{Theorem:ODE:E1} on $[0,T]$. We note that $u^\epsilon$ and $v^\epsilon$ are affine functions on $[0,T]$. Moreover by the one-side Lipschitz condition
\begin{eqnarray*}
\frac{d}{dt} \|u^\epsilon(t)-v^\epsilon(t)\|&=&\Big[u^\epsilon(t)-v^\epsilon(t),\dot{u}^\epsilon(t)-\dot{v}^\epsilon(t)\Big]\\
&\le& \Big[u^\epsilon(t)-v^\epsilon(t),Q[u^\epsilon(t)]-Q[v^\epsilon(t)]\Big]+2\epsilon\\
&\le& C\|u^\epsilon(t)-v^\epsilon(t)\|+2\epsilon, 
\end{eqnarray*}
for a.e. $t\in[0,T]$, which leads to
$$\|u^\epsilon(t)-v^\epsilon(t)\|\le 2\epsilon\frac{e^{LT}}{L}.$$
Let $\epsilon$ tend to $0$, $u^\epsilon\to u$ uniformly on $[0,T]$. It is straightforward that $u$ is a solution to \eqref{Theorem_ODE_Eq}.

  ~~ \\{\bf Acknowledgements.} 
This work was partially supported by a grant from the Simons Foundation ($\#$395767 to Avraham Soffer). A. Soffer is partially supported by NSF grant DMS 1201394.
M.-B Tran is partially supported by NSF Grant DMS (Ki-Net) 1107291, ERC Advanced Grant DYCON. M.-B Tran would like to thank Professor Daniel Heinzen, Professor Linda Reichl, Professor Mark Raizen, Professor Robert Dorfman and Professor Jeremie Szeftel for fruiful discussions on the topic. A part of the research was carried on while M.-B. Tran was visiting  University of Texas at Austin. A part of this work was done when both of the authors were visiting the Central China Normal University, China. The authors would like to acknowledge the institutions for the hospitality. The authors would like to thank the referee for useful comments.
 \bibliographystyle{plain}
 \bibliography{QuantumBoltzmann}

\def\cprime{$'$} \def\cprime{$'$} \def\cprime{$'$} \def\cprime{$'$}
  \def\cprime{$'$} \def\cprime{$'$} \def\cprime{$'$} \def\cprime{$'$}
  \def\cprime{$'$}
\begin{thebibliography}{10}

\bibitem{Allemand:DOF:2009}
Thibaut Allemand.
\newblock Derivation of a two-fluids model for a {B}ose gas from a quantum
  kinetic system.
\newblock {\em Kinet. Relat. Models}, 2(2):379--402, 2009.

\bibitem{ABCL:2016}
R.~Alonso, V.~Bagland, Y.~Cheng, and B.~Lods.
\newblock One dimensional dissipative boltzmann equation: measure solutions,
  cooling rate and self-similar profile.
\newblock {\em Submitted}, 2016.

\bibitem{AlonsoGambaBinh}
Ricardo Alonso, Irene~M Gamba, and Minh-Binh Tran.
\newblock The cauchy problem for the quantum boltzmann equation for bosons at
  very low temperature.
\newblock {\em arXiv preprint arXiv:1609.07467}, 2016.

\bibitem{WiemanCornell}
M.H. Anderson, J.R. Ensher, M.R. Matthews, C.E. Wieman, and E.A. Cornell.
\newblock Observation of {B}ose–{E}instein {C}ondensation in a dilute atomic
  vapor.
\newblock {\em Science}, 269(5221):198--201, 1995.

\bibitem{Ketterle}
M.~R. Andrews, C.~G. Townsend, H.-J. Miesner, D.~S. Durfee, D.~M. Kurn, and
  W.~Ketterle.
\newblock Observation of interference between two {B}ose condensates.
\newblock {\em Science}, 275 (5300):637--641, 1997.

\bibitem{anglin2002bose}
James~R. Anglin and Wolfgang Ketterle.
\newblock Bose--einstein condensation of atomic gases.
\newblock {\em Nature}, 416(6877):211--218, 2002.

\bibitem{Arkeryd:1972:OBE}
Leif Arkeryd.
\newblock On the {B}oltzmann equation. {I}. {E}xistence.
\newblock {\em Arch. Rational Mech. Anal.}, 45:1--16, 1972.

\bibitem{ArkerydNouri:2012:BCI}
Leif Arkeryd and Anne Nouri.
\newblock Bose condensates in interaction with excitations: a kinetic model.
\newblock {\em Comm. Math. Phys.}, 310(3):765--788, 2012.

\bibitem{ArkerydNouri:AMP:2013}
Leif Arkeryd and Anne Nouri.
\newblock A {M}ilne problem from a {B}ose condensate with excitations.
\newblock {\em Kinet. Relat. Models}, 6(4):671--686, 2013.

\bibitem{ArkerydNouri:2015:BCI}
Leif Arkeryd and Anne Nouri.
\newblock Bose condensates in interaction with excitations: a two-component
  space-dependent model close to equilibrium.
\newblock {\em J. Stat. Phys.}, 160(1):209--238, 2015.

\bibitem{bach2016time}
Volker Bach, S{\'e}bastien Breteaux, Thomas Chen, J{\"u}rg Fr{\"o}hlich, and
  Israel~Michael Sigal.
\newblock The time-dependent {H}artree-{F}ock-{B}ogoliubov equations for
  bosons.
\newblock {\em arXiv preprint arXiv:1602.05171}, 2016.

\bibitem{BarbaraGoss}
L.~Barbara~Goss.
\newblock {C}ornell, {K}etterle, and {W}ieman share {N}obel {P}rize for
  {B}ose-{E}instein {C}ondensates.
\newblock {\em Search and Discovery. Physics Today online.}, 2001.

\bibitem{KirkpatrickSchlein:2013:ACL}
G{\'e}rard Ben~Arous, Kay Kirkpatrick, and Benjamin Schlein.
\newblock A central limit theorem in many-body quantum dynamics.
\newblock {\em Comm. Math. Phys.}, 321(2):371--417, 2013.

\bibitem{BennemannBennemann:TPL:1976}
K.~H. Bennemann and J.~B. Bennemann.
\newblock {\em The Physics of Liquid and Solid Helium}, volume~1 of {\em
  Interscience Monographs and Texts in Physics And Astronomy}.
\newblock Wiley, New York, Wiley, New York, 1976.

\bibitem{bijlsma2000condensate}
M.~J. Bijlsma, E.~Zaremba, and H.~T.~C. Stoof.
\newblock Condensate growth in trapped bose gases.
\newblock {\em Physical Review A}, 62(6):063609, 2000.

\bibitem{Bressan}
A.~Bressan.
\newblock Notes on the {B}oltzmann equation.
\newblock {\em Lecture notes for a summer course, S.I.S.S.A. Trieste}, 2005.

\bibitem{BriantEinav:2016:OTC}
Marc Briant and Amit Einav.
\newblock On the {C}auchy problem for the homogeneous {B}oltzmann-{N}ordheim
  equation for bosons: local existence, uniqueness and creation of moments.
\newblock {\em J. Stat. Phys.}, 163(5):1108--1156, 2016.

\bibitem{buckmaster2016effective}
T.~Buckmaster, P.~Germain, Z.~Hani, and J.~Shatah.
\newblock Effective dynamics of the nonlinear schr\"odinger equation on large
  domains.
\newblock {\em arXiv preprint arXiv:1610.03824}, 2016.

\bibitem{buckmaster2016analysis}
T.~Buckmaster, P.~Germain, Z.~Hani, and J.~Shatah.
\newblock Analysis of the {(CR)} equation in higher dimensions.
\newblock {\em International Mathematics Research Notices Accepted}, 2017.

\bibitem{Carleman:1933:TEI}
Torsten Carleman.
\newblock Sur la th\'eorie de l'\'equation int\'egrodiff\'erentielle de
  {B}oltzmann.
\newblock {\em Acta Math.}, 60(1):91--146, 1933.

\bibitem{Cercignani:1975:TAB}
Carlo Cercignani.
\newblock {\em Theory and application of the {B}oltzmann equation}.
\newblock Elsevier, New York, 1975.

\bibitem{Cercignani:1988:BEI}
Carlo Cercignani.
\newblock {\em The {B}oltzmann equation and its applications}, volume~67 of
  {\em Applied Mathematical Sciences}.
\newblock Springer-Verlag, New York, 1988.

\bibitem{CercignaniIllnerPulvirenti:1994:TMT}
Carlo Cercignani, Reinhard Illner, and Mario Pulvirenti.
\newblock {\em The mathematical theory of dilute gases}, volume 106 of {\em
  Applied Mathematical Sciences}.
\newblock Springer-Verlag, New York, 1994.

\bibitem{CraciunBinh}
Gheorghe Craciun and Minh-Binh Tran.
\newblock A reaction network approach to the convergence to equilibrium of
  quantum boltzmann equations for bose gases.
\newblock {\em arXiv preprint arXiv:1608.05438}, 2016.

\bibitem{DeckertFrohlichPickl:2016:DSW}
D-A. Deckert, J.~Fr{\"o}hlich, P.~Pickl, and A.~Pizzo.
\newblock Dynamics of sound waves in an interacting {B}ose gas.
\newblock {\em Adv. Math.}, 293:275--323, 2016.

\bibitem{E}
U.~Eckern.
\newblock Relaxation processes in a condensed bose gas.
\newblock {\em J. Low Temp. Phys.}, 54:333--359, 1984.

\bibitem{EscobedoVelazquez:2015:FTB}
M.~Escobedo and J.~J.~L. Vel{\'a}zquez.
\newblock Finite time blow-up and condensation for the bosonic {N}ordheim
  equation.
\newblock {\em Invent. Math.}, 200(3):761--847, 2015.

\bibitem{EscobedoVelazquez:2015:OTT}
M.~Escobedo and J.~J.~L. Vel{\'a}zquez.
\newblock On the theory of weak turbulence for the nonlinear {S}chr\"odinger
  equation.
\newblock {\em Mem. Amer. Math. Soc.}, 238(1124):v+107, 2015.

\bibitem{EPV}
Miguel Escobedo, Federica Pezzotti, and Manuel Valle.
\newblock Analytical approach to relaxation dynamics of condensed {B}ose gases.
\newblock {\em Ann. Physics}, 326(4):808--827, 2011.

\bibitem{Binh9}
Miguel Escobedo and Minh-Binh Tran.
\newblock Convergence to equilibrium of a linearized quantum {B}oltzmann
  equation for bosons at very low temperature.
\newblock {\em Kinetic and Related Models}, 8(3):493---531, 2015.

\bibitem{FaouGermainHani:TWN:2016}
E.~Faou, P.~Germain, and Z.~Hani.
\newblock The weakly nonlinear large-box limit of the 2{D} cubic nonlinear
  {S}chr\"odinger equation.
\newblock {\em J. Amer. Math. Soc.}, 29(4):915--982, 2016.

\bibitem{GambaSmithBinh}
Irene~M Gamba, Leslie~M Smith, and Minh-Binh Tran.
\newblock On the wave turbulence theory for stratified flows in the ocean.
\newblock {\em arXiv preprint arXiv:1709.08266}, 2017.

\bibitem{QK1}
C.~Gardiner and P.~Zoller.
\newblock Quantum kinetic theory. {A} quantum kinetic master equation for
  condensation of a weakly interacting {B}ose gas without a trapping potential.
\newblock {\em Phys. Rev. A}, 55:2902, 1997.

\bibitem{QK3}
C.~Gardiner and P.~Zoller.
\newblock Quantum kinetic theory. {III}. {Q}uantum kinetic master equation for
  strongly condensed trapped systems.
\newblock {\em Phys. Rev. A}, 58:536, 1998.

\bibitem{QK5}
C.~Gardiner and P.~Zoller.
\newblock Quantum kinetic theory. {V}. {Q}uantum kinetic master equation for
  mutual interaction of condensate and noncondensate.
\newblock {\em Phys. Rev. A}, 61:033601, 2000.

\bibitem{QK0}
C.~Gardiner, P.~Zoller, R.~J. Ballagh, and M.~J. Davis.
\newblock Kinetics of {B}ose-{E}instein condensation in a trap.
\newblock {\em Phys. Rev. Lett.}, 79:1793, 1997.

\bibitem{QK6}
C.~W. Gardiner, M.~D. Lee, R.~J. Ballagh, M.~J. Davis, and P.~Zoller.
\newblock Quantum kinetic theory of condensate growth: Comparison of experiment
  and theory.
\newblock {\em Phys. Rev. Lett.}, 81:5266, 1998.

\bibitem{germain2015continuous}
P.~Germain, Z.~Hani, and L.~Thomann.
\newblock On the continuous resonant equation for {NLS}, {II: S}tatistical
  study.
\newblock {\em Analysis \& PDE}, 8(7):1733--1756, 2015.

\bibitem{germain2016continuous}
P.~Germain, Z.~Hani, and L.~Thomann.
\newblock On the continuous resonant equation for {NLS. I. D}eterministic
  analysis.
\newblock {\em Journal de Math{\'e}matiques Pures et Appliqu{\'e}es},
  105(1):131--163, 2016.

\bibitem{germain2015high}
P.~Germain and L.~Thomann.
\newblock On the high frequency limit of the lll equation.
\newblock {\em arXiv preprint arXiv:1509.09080}, 2015.

\bibitem{GermainIonescuTran}
Pierre Germain, Alexandru~D Ionescu, and Minh-Binh Tran.
\newblock Optimal local well-posedness theory for the kinetic wave equation.
\newblock {\em arXiv preprint arXiv:1711.05587}, 2017.

\bibitem{Glassey:1996:CPK}
Robert~T. Glassey.
\newblock {\em The {C}auchy problem in kinetic theory}.
\newblock Society for Industrial and Applied Mathematics (SIAM), Philadelphia,
  PA, 1996.

\bibitem{GriffinNikuniZaremba:2009:BCG}
Nikuni T. Zaremba~E. Griffin, A.
\newblock Bose-condensed gases at finite temperatures.
\newblock {\em Cambridge University Press, Cambridge, 2009}.

\bibitem{GrillakisMachedon:2013:BMF}
M.~Grillakis and M.~Machedon.
\newblock Beyond mean field: on the role of pair excitations in the evolution
  of condensates.
\newblock {\em J. Fixed Point Theory Appl.}, 14(1):91--111, 2013.

\bibitem{GrillakisMachedon:2013:PEA}
M.~Grillakis and M.~Machedon.
\newblock Pair excitations and the mean field approximation of interacting
  bosons, {I}.
\newblock {\em Comm. Math. Phys.}, 324(2):601--636, 2013.

\bibitem{GrillakisMachedonMargetis:2011:SOC}
M.~Grillakis, M.~Machedon, and D.~Margetis.
\newblock Second-order corrections to mean field evolution of weakly
  interacting bosons. {II}.
\newblock {\em Adv. Math.}, 228(3):1788--1815, 2011.

\bibitem{ReichlGust:2012:CII}
E.~D Gust and L.~E. Reichl.
\newblock Collision integrals in the kinetic equations ofdilute bose-einstein
  condensates.
\newblock {\em arXiv:1202.3418}, 2012.

\bibitem{ReichlGust:2013:RRA}
E.~D Gust and L.~E. Reichl.
\newblock Relaxation rates and collision integrals for bose-einstein
  condensates.
\newblock {\em Phys. Rev. A}, 170:43--59, 2013.

\bibitem{gust2013transport}
E.~D. Gust and L.~E. Reichl.
\newblock Transport coefficients from the boson {U}ehling-{U}hlenbeck equation.
\newblock {\em Physical Review E}, 87(4):042109, 2013.

\bibitem{IG}
M.~Imamovic-Tomasovic and A.~Griffin.
\newblock Quasiparticle kinetic equation in a trapped bose gas at low
  temperatures.
\newblock {\em J. Low Temp. Phys.}, 122:617--655, 2001.

\bibitem{inguscio1999bose}
Massimo Inguscio, Sandro Stringari, and Carl~E. Wieman.
\newblock {\em {B}ose-{E}instein condensation in atomic gases}, volume 140.
\newblock IOS Press, Amsterdam, 1999.

\bibitem{QK4}
D.~Jaksch, C.~Gardiner, K.~M. Gheri, and P.~Zoller.
\newblock Quantum kinetic theory. {IV}. {I}ntensity and amplitude fluctuations
  of a {B}ose-{E}instein condensate at finite temperature including trap loss.
\newblock {\em Phys. Rev. A}, 58:1450, 1998.

\bibitem{QK2}
D.~Jaksch, C.~Gardiner, and P.~Zoller.
\newblock Quantum kinetic theory. {II}. {S}imulation of the quantum {B}oltzmann
  master equation.
\newblock {\em Phys. Rev. A}, 56:575, 1997.

\bibitem{JinBinh}
Shi Jin and Minh-Binh Tran.
\newblock Quantum hydrodynamic approximations to the finite temperature trapped
  bose gases.
\newblock {\em arXiv preprint arXiv:1703.00825}, 2017.

\bibitem{kagan1997evolution}
Y.~Kagan and B.~V. Svistunov.
\newblock Evolution of correlation properties and appearance of broken symmetry
  in the process of bose-einstein condensation.
\newblock {\em Physical review letters}, 79(18):3331, 1997.

\bibitem{KD1}
T.~R. Kirkpatrick and J.~R. Dorfman.
\newblock Transport theory for a weakly interacting condensed {B}ose gas.
\newblock {\em Phys. Rev. A (3)}, 28(4):2576--2579, 1983.

\bibitem{KD2}
T.~R. Kirkpatrick and J.~R. Dorfman.
\newblock Transport in a dilute but condensed nonideal bose gas: Kinetic
  equations.
\newblock {\em J. Low Temp. Phys.}, 58:301--331, 1985.

\bibitem{lieb2002proof}
Elliott~H Lieb and Robert Seiringer.
\newblock Proof of {B}ose-{E}instein condensation for dilute trapped gases.
\newblock {\em Physical review letters}, 88(17):170409, 2002.

\bibitem{Lu:2004:OID}
Xuguang Lu.
\newblock On isotropic distributional solutions to the {B}oltzmann equation for
  {B}ose-{E}instein particles.
\newblock {\em J. Statist. Phys.}, 116(5-6):1597--1649, 2004.

\bibitem{Lu:2005:TBE}
Xuguang Lu.
\newblock The {B}oltzmann equation for {B}ose-{E}instein particles: velocity
  concentration and convergence to equilibrium.
\newblock {\em J. Stat. Phys.}, 119(5-6):1027--1067, 2005.

\bibitem{Lu:2013:TBE}
Xuguang Lu.
\newblock The {B}oltzmann equation for {B}ose-{E}instein particles:
  condensation in finite time.
\newblock {\em J. Stat. Phys.}, 150(6):1138--1176, 2013.

\bibitem{LukkarinenSpohn:WNS:2011}
Jani Lukkarinen and Herbert Spohn.
\newblock Weakly nonlinear {S}chr\"odinger equation with random initial data.
\newblock {\em Invent. Math.}, 183(1):79--188, 2011.

\bibitem{Martin}
R.~H. Martin.
\newblock {\em Nonlinear operators and differential equations in Banach
  spaces}.
\newblock Pure and Applied Mathematics. Wiley-Interscience, 1976.

\bibitem{mitrouskas2016bogoliubov}
David Mitrouskas, S{\"o}ren Petrat, and Peter Pickl.
\newblock Bogoliubov corrections and trace norm convergence for the hartree
  dynamics.
\newblock {\em arXiv preprint arXiv:1609.06264}, 2016.

\bibitem{Nazarenko:2011:WT}
Sergey Nazarenko.
\newblock {\em Wave turbulence}, volume 825 of {\em Lecture Notes in Physics}.
\newblock Springer, Heidelberg, 2011.

\bibitem{Nepomnyashchii:1978}
Yu.~A. Nepomnyashchii and A.~A. Nepomnyashchii.
\newblock Infrared divergence in field theory of a base system with a
  condensate.
\newblock {\em Soy. Phys.-JETP}, 493(48), 1978.

\bibitem{ToanBinh2}
T.~T. Nguyen and M.-B. Tran.
\newblock On the kinetic equation in zakharov's wave turbulence theory for
  capillary waves.
\newblock {\em arXiv preprint arXiv:1702.03892}, 2017.

\bibitem{ToanBinh}
Toan~T Nguyen and Minh-Binh Tran.
\newblock Uniform in time lower bound for solutions to a quantum boltzmann
  equation of bosons.
\newblock {\em arXiv preprint arXiv:1605.07890}, 2016.

\bibitem{Nordheim:OTK:1928}
L.W. Nordheim.
\newblock On the kinetic methods in the new statistics and its applications in
  the electron theory of conductivity.
\newblock {\em Proc. Roy. Soc. London Ser. A}, 119:689--698, 1928.

\bibitem{pomeau1999theorie}
Y.~Pomeau, M.-{E}. Brachet, S.~M{\'e}tens, and S.~Rica.
\newblock Th{\'e}orie cin{\'e}tique d'un gaz de bose dilu{\'e} avec condensat.
\newblock {\em Comptes Rendus de l'Acad{\'e}mie des Sciences-Series
  IIB-Mechanics-Physics-Astronomy}, 327(8):791--798, 1999.

\bibitem{PopovSeredniakov:1978}
V.~N. Popov and A.~V. Seredniakov.
\newblock Low-frequency asymptotic form of the self-energy parts of a
  superfluid bose system at {T}=0.
\newblock {\em Soy. Phys.-JETP}, 193(50), 1979.

\bibitem{ColdAtoms1}
N.~Proukakis, S.~Gardiner, M.~Davis, and M.~Szymanska.
\newblock {\em Cold Atoms: Volume 1 Quantum Gases Finite Temperature and
  Non-Equilibrium Dynamics}.
\newblock Imperial College Press, 2013.

\bibitem{proukakis2008finite}
Nick~P Proukakis and Brian Jackson.
\newblock Finite-temperature models of bose--einstein condensation.
\newblock {\em Journal of Physics B: Atomic, Molecular and Optical Physics},
  41(20):203002, 2008.

\bibitem{ReichlGust:2013:TTF}
L.~E. Reichl and E.~D Gust.
\newblock Transport theory for a dilute bose-einstein condensate.
\newblock {\em J Low Temp Phys}, 88:053603, 2013.

\bibitem{ReichlTran}
Linda~E Reichl and Minh-Binh Tran.
\newblock A kinetic model for very low temperature dilute bose gases.
\newblock {\em arXiv preprint arXiv:1709.09982}, 2017.

\bibitem{seiringer2011excitation}
Robert Seiringer.
\newblock The excitation spectrum for weakly interacting bosons.
\newblock {\em Communications in mathematical physics}, 306(2):565--578, 2011.

\bibitem{semikoz1995kinetics}
D.~V. Semikoz and Igor~I. Tkachev.
\newblock Kinetics of bose condensation.
\newblock {\em Physical review letters}, 74(16):3093, 1995.

\bibitem{semikoz1997condensation}
D.~V. Semikoz and Igor~I. Tkachev.
\newblock Condensation of bosons in the kinetic regime.
\newblock {\em Physical Review D}, 55(2):489, 1997.

\bibitem{soffer2016coupling}
A.~Soffer and M.-B. Tran.
\newblock On coupling kinetic and schrodinger equations.
\newblock {\em arXiv preprint arXiv:1610.04496}, 2016.

\bibitem{Spohn:2010:KOT}
H.~Spohn.
\newblock Kinetics of the bose-einstein condensation.
\newblock {\em Physica D}, 239:627--634, 2010.

\bibitem{Spohn:WNW:2010}
Herbert Spohn.
\newblock Weakly nonlinear wave equations with random initial data.
\newblock In {\em Proceedings of the {I}nternational {C}ongress of
  {M}athematicians. {V}olume {III}}, pages 2128--2143. Hindustan Book Agency,
  New Delhi, 2010.

\bibitem{Stoof:1999:CVI}
H.~Stoof.
\newblock Coherent versus incoherent dynamics during bose-einstein condensation
  in atomic gases.
\newblock {\em J. Low Temp. Phys.}, 114:11--108, 1999.

\bibitem{UehlingUhlenbeck:TPI:1933}
Uhlenbeck~G.E. Uehling, E.A.
\newblock Transport phenomena in einstein-bose and fermi-dirac gases.
\newblock {\em Phys. Rev.}, 43:552--561, 1933.

\bibitem{Villani:2002:RMT}
C{\'e}dric Villani.
\newblock A review of mathematical topics in collisional kinetic theory.
\newblock In {\em Handbook of mathematical fluid dynamics, {V}ol. {I}}, pages
  71--305. North-Holland, Amsterdam, 2002.

\bibitem{Wennberg:1997:EDM}
Bernt Wennberg.
\newblock Entropy dissipation and moment production for the {B}oltzmann
  equation.
\newblock {\em J. Statist. Phys.}, 86(5-6):1053--1066, 1997.

\bibitem{williams2002dynamical}
J.~E. Williams, E.~Zaremba, B.~Jackson, T.~Nikuni, and A.~Griffin.
\newblock Dynamical instability of a condensate induced by a rotating thermal
  gas.
\newblock {\em Physical review letters}, 88(7):070401, 2002.

\bibitem{Zakharov:1998:NWA}
V.~E. Zakharov, editor.
\newblock {\em Nonlinear waves and weak turbulence}, volume 182 of {\em
  American Mathematical Society Translations, Series 2}.
\newblock American Mathematical Society, Providence, RI, 1998.
\newblock Advances in the Mathematical Sciences, 36.

\bibitem{zakharov2012kolmogorov}
V.~E. Zakharov, V.~S. L'vov, and G.~Falkovich.
\newblock {\em Kolmogorov spectra of turbulence I: Wave turbulence}.
\newblock Springer Science \& Business Media, 2012.

\bibitem{ZarembaNikuniGriffin:1999:DOT}
Nikuni T. Griffin~A. Zaremba, E.
\newblock Dynamics of trapped bose gases at finite temperatures.
\newblock {\em J. Low Temp. Phys.}, 116:277--345, 1999.

\end{thebibliography}

\end{document}